\pdfoutput=1
\documentclass[letterpaper]{article}
\def\IRISArxiv{1}
\newif\ifIRISIncludeMain
\newif\ifIRISIncludeSupplement
\IRISIncludeMaintrue
\IRISIncludeSupplementtrue
\ifdefined\IRISMainOnly
  \IRISIncludeSupplementfalse
\fi
\ifdefined\IRISSupplementOnly
  \IRISIncludeMainfalse
\fi
\usepackage{aaai2027}  
\nocopyright           
\makeatletter
\newcommand{\IRIScorrespondingplural}{\global\aaai@corrmultitrue}
\makeatother
\usepackage[hyphens]{url}  
\usepackage{graphicx} 
\usepackage{natbib}  
\usepackage{caption} 
\usepackage{amsmath}
\usepackage{amssymb}
\usepackage{amsthm}
\usepackage{booktabs}
\usepackage{longtable}
\usepackage{algorithm}      
\usepackage{algpseudocode}  
\usepackage{bibunits}
\usepackage{xr}             
\graphicspath{{figures/}}

\makeatletter
\newenvironment{breakablealgorithm}
{%
\begin{center}
    \refstepcounter{algorithm}
    \hrule height.8pt depth0pt \kern2pt
    \renewcommand{\caption}[2][\relax]{%
    {\raggedright\textbf{\ALG@name~\thealgorithm} ##2\par}%
    \ifx\relax##1\relax
        \addcontentsline{loa}{algorithm}{\protect\numberline{\thealgorithm}##2}%
    \else
        \addcontentsline{loa}{algorithm}{\protect\numberline{\thealgorithm}##1}%
    \fi
    \kern2pt\hrule\kern2pt
    }
}{%
    \kern2pt\hrule\relax
\end{center}
}
\makeatother

\newtheorem{theorem}{Theorem}
\newtheorem{proposition}{Proposition}
\newtheorem{corollary}{Corollary}
\newtheorem{definition}{Definition}
\newtheorem{lemma}{Lemma}

\newcommand{\IRIS}{\textsc{IRIS}}

\ifdefined\IRISMainOnly
\fi
\ifdefined\IRISSupplementOnly
\fi
\defaultbibliographystyle{aaai2027}
\defaultbibliography{references}
\newwrite\IRISbuaux
\AtBeginDocument{\immediate\openout\IRISbuaux=bu.aux\relax\immediate\closeout\IRISbuaux}

\title{Which Model Is Actually Serving You?\\
\IRIS{}: Budgeted Black-Box Auditing of Model Substitution and Routing Dilution in LLM Gateways}
\author{Yuewei Zhang\textsuperscript{\rm 1},
Zhi-Hai Zhang\textsuperscript{\rm 1}\corresponding,
Hanzhang Qin\textsuperscript{\rm 2}\corresponding\IRIScorrespondingplural
}
\affiliations{
  \textsuperscript{\rm 1}Department of Industrial Engineering, Tsinghua University, Beijing, China\\
  \textsuperscript{\rm 2}Department of Industrial Systems Engineering and Management, National University of Singapore, Singapore\\
  \{zyw23@mails.tsinghua.edu.cn, zhzhang@tsinghua.edu.cn, hzqin@nus.edu.sg\}
}

\begin{document}
\ifIRISIncludeMain
\begin{bibunit}
\maketitle

\begin{abstract}
    Commercial LLM gateways mediate access to hosted models, but the served backend may not match the advertised one: it may substitute a cheaper model on every request or route only a fraction $\epsilon$ of requests to it. Prior black-box auditors often need a privileged signal (log-probabilities, token ranks, or reference samples) or a target-specific probe, fix the query budget in advance, and return a yes/no verdict. We present \IRIS{}, an audit that needs only the returned text: it asks endpoints to generate random numbers or strings, fingerprints the backend, and is the first to combine, in one text-only audit, detection of whole-stream substitution and fractional dilution, attribution of the served backend, routing-fraction ($\epsilon$) estimation, and a query budget it sizes itself. A cheap pilot fits the exponential query-error decay and freezes that budget before any suspect query is issued. On an intra-family Qwen3 ladder \IRIS{} verifies the backend at $0.99$ AUROC and sharpens attribution as queries accumulate; across a commercial OpenRouter library it catches $\epsilon{=}0.3$ dilution on margin-qualified pairs at $0.85$ mean power ($0.017$ false-positive rate) and recovers $\epsilon$ to within $0.04$ for enrolled diluents; and a live cross-provider audit flags $14$ of $15$ same-model provider pairs by genuine quantization and kernel deviations, corroborated on third-party MET traces. Against comparable black-box auditors, \IRIS{} matches or beats detection on shared tasks, and adaptive allocation lifts the matched-budget target-hit rate from $73\%$ to $87\%$. Further experiments cover adversarial gateways, knob identifiability, unseen diluents, and false-positive control.
\end{abstract}

\ifdefined\IRISArxiv
\begin{links}
  \link{Code and data}{https://github.com/Photen/IRIS-audit}
\end{links}
\fi
\section{Introduction}
\label{sec:intro}

Commercial API gateways and shadow APIs increasingly aggregate LLMs from competing vendors behind one endpoint and in some markets are resold as low-cost commodity routers~\citep{supplychain2026}. This indirection breaks the link between the model a client requests and the backend that actually serves the call. Measurements document silent downgrades, model switches, billing deviations, and fingerprint failures~\citep{realmoney2026,gatescope2026}; formal work studies cheaper, quantized, or randomized substitutes~\citep{substitution2025}; and GhostPrint shows a weaker model can be fine-tuned to spoof fingerprinting audits~\citep{ghostprint2026}. We call whole-stream replacement \emph{substitution} and fractional replacement \emph{dilution}. Either may arise through misconfiguration, cost control, or misrepresentation, and either requires an audit that infers from returned text alone which model served the request and how much traffic was replaced.

\begin{figure}[t]
  \centering
  \includegraphics[width=\columnwidth]{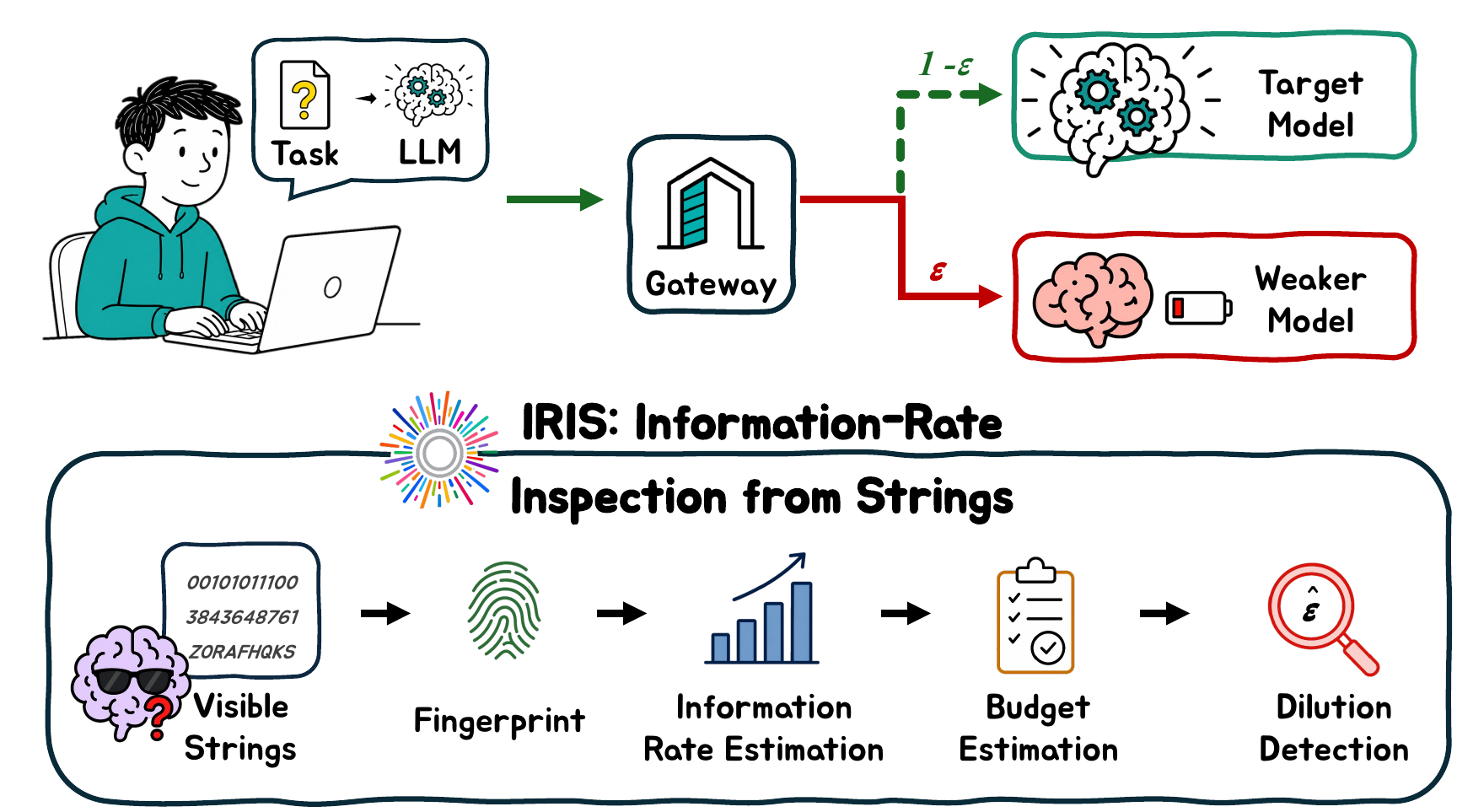}
  \caption{\IRIS{} audits gateway routing from visible strings only: a user task may be served by the intended model or by a weaker substitute; \IRIS{} turns the returned strings into fingerprints, estimates the endpoint information rate, sizes the audit budget, and detects routing dilution.}
  \label{fig:overview}
\end{figure}

Output-only auditors make served-model checking practical~\citep{met2025,rut2026,b3it2026,kbf2026}, yet four gaps still separate them from a general endpoint-level audit: (i)~probes and references are often endpoint-, reference-, or method-specific rather than reusable across endpoints; (ii)~most methods return only a binary substituted-or-not verdict; (iii)~the method that does estimate a routing fraction assumes a known substitute under fixed routing; and (iv)~query counts are fixed by design, so the auditor cannot tell in advance how many interactions a target error rate requires, or whether a given dilution is detectable within a reasonable budget.

\IRIS{} (\textbf{I}nformation-\textbf{R}ate \textbf{I}nspection from \textbf{S}trings) addresses these gaps by making the audit budget an object of estimation. The auditor probes each endpoint with a random-generation challenge whose visible output exposes backend-specific sampling biases~\citep{hopkins2023random,baddice2026}. \emph{Audit-Plan Construction} turns these strings into a reusable candidate library, calibrates thresholds, and freezes the query budget before any suspect traffic is queried. \emph{Audit Execution} then spends that budget on the gateway, attributes the served backend, and, when the diluting substitute is enrolled, estimates the routing fraction (Figure~\ref{fig:overview}).
\paragraph{Contributions.}
\begin{itemize}
  \item \textbf{A reusable audit for cross-model substitution and dilution.} To our knowledge, \IRIS{} is the first visible-string random-generation audit that uses reusable probes to detect cross-model gateway substitution, estimate the routing fraction, attribute the served backend, and set the live-query budget before the audit.
  \item \textbf{An estimate-then-budget strategy for endpoint audits.} A cheap pilot fits the exponential rank-error decay, calibrates response-level tell rates (rates of reference-atypical responses), and prescribes the query count a target reliability requires, all before Audit Execution begins.
  \item \textbf{Theory and guarantees for audit effectiveness.} Evidence accumulates across independent queries, not within one long response. When enrollment observes enough response-level tells, dilution cost grows only logarithmically in the target miss probability and roughly inversely with the routing fraction; the measured tail exponent identifies departures from this regime.
  \item \textbf{Validation on real commercial models.} Over claimed--substitute pairs from a $K{=}6$ ladder and $17$ OpenRouter APIs, the pilot exponent predicts query difficulty and, on the ladder, drives adaptive allocation that meets the target on $87\%$ of pairs versus $73\%$ for a fixed rule at matched budget. At scale, on margin-qualified pairs, \IRIS{} catches $\epsilon{=}0.3$ dilution at mean power $0.85$ with a $0.017$ false-positive rate, names the substitute, and flags an unenrolled one.
\end{itemize}

\section{Related Work}
\label{sec:related}

\paragraph{The market reality of substitution.}
Measurement studies establish that substitution is real, but they characterize gateway populations rather than giving an error-controlled audit of one endpoint. \citet{realmoney2026} find cross-model impersonation on real shadow APIs, including a GPT-5 endpoint fingerprinting as GLM-4-9B, and report that about $45.8\%$ of endpoints fail fingerprint verification; GateScope records silent downgrades, model switches, and billing deviations across live gateways~\citep{gatescope2026}; and \citet{supplychain2026} broaden the market view to $428$ commodity LLM API routers. Model-substitution theory formalizes the same threat class~\citep{substitution2025}.

\paragraph{Identifying the served model.}
Black-box identification methods differ mainly in which behavioral signal they read. Input-based methods craft the query, using either discriminative prompt banks (LLMmap; \citealp{llmmap2025}) or adversarial suffixes (TRAP; \citealp{trap2024}). Knowledge-based methods read what the model knows, through factual-capacity probes (IKP; \citealp{ikp2026}), combined trigger-pattern and knowledge-level fingerprints (DuFFin; \citealp{duffin2026}), or recent-event recall that locates a training cutoff (LLMLagBench; \citealp{llmlagbench2025}). Output-based methods read how the model writes, from characteristic wrong-answer patterns~\citep{madcrowds2024}, word-level and semantic idiosyncrasies~\citep{idiosyncrasies2025}, or provenance and derivation between models~\citep{provenance2025}. Close to our signal, FLIPS~\citep{flips2026} trains a supervised classifier on pseudo-random sequence bias to fingerprint deployment configurations at a fixed query count, while \IRIS{} can further estimate the audit budget and the routing fraction. Most methods require a model-specific probe or reference and return a yes/no verdict, not a routing fraction that a client may tolerate, price, or contest.

\paragraph{Detecting substitution and dilution.}
Substitution and dilution auditors also vary by access assumption. KBF~\citep{kbf2026} uses numerical recall at the knowledge boundary and diagnostically estimates the routing fraction for a known substitute under fixed routing; B3IT~\citep{b3it2026} uses low-temperature border inputs; Log Probability Tracking~\citep{logprobtrack2026} uses a single token's mean log-probability; MET~\citep{met2025} and RUT~\citep{rut2026} use reference-based equality and rank-uniformity. Classical testing theory gives exponential error decay in independent samples~\citep{chernoff1952,hoeffding1965,cover2006elements}, but dilution is harder than substitution: only an $\epsilon$ fraction of responses carries substitute evidence, a low-rate regime that \citet{substitution2025} treat as a reason to fall back on trusted hardware. The novelty is therefore not dilution alone: \IRIS{} pairs reusable random-generation probes with an estimate-then-budget plan frozen before any suspect query, and one batch of visible strings supports detection, attribution, and enrolled-diluent rate estimation.

\section{Problem Formulation}
\label{sec:prelim}

We formalize endpoint auditing as follows. A client requests a nominal model through an API gateway. A trusted reference endpoint serves that model faithfully, producing visible outputs from distribution $P$ that the auditor may sample during Audit-Plan Construction. A suspect endpoint advertises the same model but may instead serve a different, cheaper base model $R$ on some or all requests. We model its served distribution as
\begin{equation}
  Q_\epsilon=(1-\epsilon)\,P+\epsilon\,R,\qquad \epsilon\in[0,1],
\label{eq:mixture}
\end{equation}
so $\epsilon{=}0$ is honest service, $\epsilon{=}1$ is \emph{substitution}, and $\epsilon\in(0,1)$ is \emph{dilution}, with the substitute $R$ serving as the \emph{diluent}. Thus, substitution replaces the advertised model on every request, whereas dilution replaces it on only a fraction $\epsilon$ of requests. Temperature, decoding, or reasoning-effort changes may affect output quality, but under token-count billing they do not create extra gateway profit and therefore are not treated as auditing targets.

\paragraph{Observation model.}
The audit is strictly black-box: the auditor fixes a probe (a query prompt) and issues $m$ independent calls to the suspect endpoint, observing only the returned visible strings $y_1,\dots,y_m\mathrel{\smash{\overset{\scriptscriptstyle\mathrm{i.i.d.}}{\sim}}}Q_\epsilon$, with no weights, log-probabilities, or token ranks exposed. The reference $P$ and a library $\{R_j\}$ of candidate substitutes can be sampled under known labels during Audit-Plan Construction, whereas the suspect stream is unlabeled.

\paragraph{Audit goals.}
From the $m$ suspect outputs, the auditor performs three increasingly informative tasks: (1)~\emph{detect} any deviation from the nominal model by testing $H_0{:}\,\epsilon{=}0$ against $H_1{:}\,\epsilon{>}0$ at false-positive rate $\alpha$ (flagging an honest endpoint) and miss probability $\delta$ (failing to flag substitution or dilution); (2)~\emph{estimate} the routing fraction $\epsilon$; and (3)~given the candidate library $\{R_j\}$, \emph{attribute} the served substitute $R$. The returned decision is $(\text{flag}\in\{0,1\},\,\hat\epsilon,\,\hat R)$, with $\hat R=\varnothing$ when no substitute is named.

\section{The \IRIS{} Audit}
\label{sec:method}

\IRIS{} has two stages. \emph{Audit-Plan Construction} collects labeled responses from the trusted reference and candidate models, builds their fingerprints, and generates a user-selected audit plan for either substitution or dilution. \emph{Audit Execution} uses the plan to query the suspect endpoint, detect the selected deviation type, and identify the served backend when applicable.

Although substitution is the $\epsilon{=}1$ special case of dilution (Eq.~\eqref{eq:mixture}), we retain a dedicated mode because market audits often need the simpler, direct decision of whether an endpoint consistently serves the advertised backend. Accordingly, substitution aggregates episode-level evidence under a one-backend assumption, while dilution uses response-level tells to detect and quantify fractional routing.

\subsection{Audit-Plan Construction}
\label{sec:audit-plan-construction}

\paragraph{Probe design and labeled data collection.}
The audit uses generation probes that ask the endpoint for short random strings: a coin-style binary draw, a single digit, an integer in a range, or a short bit/digit sequence. We write $c_{n,L}$ for a probe with per-draw alphabet or range size $n$ and requested length $L$, so a single binary draw is $c_{2,1}$; below, $c$ denotes an arbitrary probe. Because these tasks have no factual answer to memorize, the returned strings expose backend-specific sampling biases under black-box access.

\IRIS{} uses a task-agnostic probe set $\mathcal C$ that varies alphabet size and response length; the full set is provided in App.~\ref{app:probes}. For every $c\in\mathcal C$, Audit-Plan Construction collects labeled responses from the trusted reference $P$ and candidate substitutes $\{R_j\}_{j=1}^{J}$. An audit plan covers candidate indices $\mathcal J_{\mathrm{audit}}\subseteq\{1,\dots,J\}$ and follows the user-selected mode $z\in\{\mathrm{s},\mathrm{d}\}$.

\paragraph{Visible-string classifier and score.}
\IRIS{} maps each response $y$ to a fixed $179$-dimensional visible-string vector $\phi(y)$ of format-compliance, symbol-frequency, transition, run-structure, and positional-balance statistics (listed in App.~\ref{app:features}). For a fixed probe $c$, the labeled responses train a multiclass classifier $g_c$ whose output is a posterior over the enrolled endpoint labels $M\in\mathcal{M}=\{P,R_1,\dots,R_J\}$:
\begin{align}
  \hat{\mathbb P}_c(M\mid y)&=\big[g_c(\phi(y))\big]_M .
\label{eq:posterior-map}
\end{align}
The deployed score is the reference negative log-posterior,
\begin{align}
  s_c(y)&=-\log\hat{\mathbb P}_c(P\mid y).
\label{eq:score}
\end{align}
Larger values mean that the response is less typical of the trusted reference for probe $c$. Given $m$ independent responses $y_1,\ldots,y_m$, \IRIS{} averages the score into the episode-level evidence statistic
\begin{align}
  S_m(c)&=\frac{1}{m}\sum_{i=1}^{m}s_c(y_i).
\label{eq:evidence}
\end{align}
When the selected probe is clear from context, we write $s(y)$ and $S_m$. Audit-Plan Construction also stores the posterior signature $\bar g_{M,c}$, the mean of $u_c(y)=\hat{\mathbb P}_c(\cdot\mid y)$ over endpoint-$M$ enrollment responses. The stored posterior vector, signature, and score are available to either execution mode. Audit-Plan Construction calibrates the two branch-specific budgets independently, but the deployed plan freezes only the branch requested by audit mode $z$.

\paragraph{Substitution audit budget.}
The substitution branch uses a mean-evidence threshold and a query budget. For each probe $c$ and enrolled pair $(P,R_j)$, a cheap labeled pilot estimates the rate at which episode-level rank error decays. Define
\begin{align}
  \mathrm{AUROC}_{cj}(m)
  &=\mathbb P\{S_m^{P,c}<S_m^{R_j,c}\}\notag\\
  &\quad+\tfrac12\mathbb P\{S_m^{P,c}=S_m^{R_j,c}\},
\label{eq:auroc-pilot}
\end{align}
where $S_m^{P,c}$ and $S_m^{R_j,c}$ are independent $m$-query episode means from $P$ and $R_j$. Because $1-\mathrm{AUROC}$ decays exponentially in $m$ (see Prop.~\ref{prop:exp}), the pilot fits
\begin{align}
  \log(1-\mathrm{AUROC}_{cj}(m))
  &\approx \hat a_{cj}-\widehat I_{\mathrm{auc},cj}m,~ m\le4,
\label{eq:auc-fit}
\end{align}
where $\hat a_{cj}$ is the fitted intercept and $\widehat I_{\mathrm{auc},cj}$ is the operational rank-error exponent. Because the $m\le4$ extrapolation is mildly optimistic, \IRIS{} uses a one-sided $(1{-}\gamma)$ lower bound $\underline I_{\mathrm{auc},cj}$ and upper bound $\overline a_{cj}$, both $t$-based over split-level fits (App.~\ref{app:lcb}). With target rank error $\delta_{\mathrm{auc}}{=}\delta$, the substitution budget is
\begin{align}
  m_{\mathrm{sub}}(c)=
  \max_{j\in\mathcal J_{\mathrm{audit}}}
  \left\lceil\frac{\overline a_{cj}-\log\delta_{\mathrm{auc}}}
  {\underline I_{\mathrm{auc},cj}}\right\rceil .
\label{eq:sub-budget}
\end{align}

The mean-evidence threshold is also generated during Audit-Plan Construction. For a probe $c$ and live budget $m$, group held-out reference responses into $B_{\mathrm{mean}}$ batches, compute $S_{m,b}^{P,c}$ by Eq.~\eqref{eq:evidence}, and take
\begin{align}
  \tau_{\mathrm{mean}}(c,m)
  =\widehat Q_{1-\alpha}\!\left(\{S_{m,b}^{P,c}\}_{b=1}^{B_{\mathrm{mean}}}\right).
\label{eq:mean-threshold}
\end{align}
After probe selection, the substitution plan freezes $\tau_{\mathrm{mean}}(c^\star,m^\star)$, so an honest reference episode exceeds it with calibrated probability $\le\alpha$.

\paragraph{Dilution audit budget.}
Dilution needs a response-level test because a substitute may serve only an $\epsilon$ fraction of queries. On a calibration split separate from the data used to train $g_c$, a threshold $\tau$ turns a response into a \emph{tell} when $s_c(y)>\tau$. \IRIS{} searches only reference quantiles: for a small honest-tail grid $\mathcal A_{\mathrm{resp}}$,
\begin{align}
  \mathcal T_c
  =\left\{\widehat Q_{1-\alpha_0}\!\left(\{s_c(y_r^P)\}_{r=1}^{n_P}\right)
  :\alpha_0\in\mathcal A_{\mathrm{resp}}\right\}.
\label{eq:resp-threshold-grid}
\end{align}
Here $y_{1:n_P}^P$ are held-out reference responses. For each candidate substitute $R_j$ with held-out responses $y_{1:n_j}^{R_j}$, threshold $\tau$ has calibrated honest and substitute tell rates
\begin{align}
  \widehat\alpha_{1,cj}(\tau)
  &=\frac{1}{n_P}\sum_{r=1}^{n_P}\mathbf{1}\{s_c(y_r^P)>\tau\},
\label{eq:honest-tell-rate}\\
  \widehat q_{cj}(\tau)
  &=\frac{1}{n_j}\sum_{r=1}^{n_j}\mathbf{1}\{s_c(y_r^{R_j})>\tau\}.
\label{eq:substitute-tell-rate}
\end{align}
Let $\mathcal T_{cj}^{+}$ collect thresholds satisfying $\widehat q_{cj}(\tau)>\widehat\alpha_{1,cj}(\tau)$; an empty set marks $(c,j)$ low-margin. At a usable threshold, an $\epsilon$-diluted response is a tell with rate $\widehat p_{\epsilon,cj}(\tau)=(1-\epsilon)\widehat\alpha_{1,cj}(\tau)+\epsilon\,\widehat q_{cj}(\tau)$. Let $U_m(b;p)=\mathbb P_{K\sim\mathrm{B}(m,p)}(K\ge b)$ and $b_\alpha(m,p_0)=\min\{b:U_m(b;p_0)\le\alpha\}$. The resulting level-$\alpha$ power is
\[
  \pi_{\epsilon,cj}(m,\tau)
  =U_m\!\left(b_\alpha(m,\widehat\alpha_{1,cj}(\tau));
  \widehat p_{\epsilon,cj}(\tau)\right).
\]
For target dilution $\epsilon_{\min}$, the query budget at threshold $\tau$ is
\begin{align}
  m_{\mathrm{dil}}(c,j;\tau)
  &=\min\Big\{m:\pi_{\epsilon_{\min},cj}(m',\tau)\ge1-\delta
  \notag\\
  &\hspace{2.7cm}\text{for all }m'\ge m\Big\},
\label{eq:mdil}
\end{align}
\IRIS{} sets $\tau_{cj}\in\arg\min_{\tau\in\mathcal T_{cj}^{+}}m_{\mathrm{dil}}(c,j;\tau)$ and writes the corresponding rates and budget as $\alpha_{1,cj},q_{cj},m_{\mathrm{dil}}(c,j)$.

\paragraph{Plan selection.}
Plan selection is mode-specific. For $z=\mathrm{s}$, all probes with valid pilot fits are eligible. For $z=\mathrm{d}$, \IRIS{} first forms $\mathcal C_{\mathrm{feas}}$, the probes with $\mathcal T_{cj}^{+}\neq\emptyset$ for every $j\in\mathcal J_{\mathrm{audit}}$. The live-query budget for probe $c$ is
\begin{align*}
  m^\star(c;z)=
  \begin{cases}
  m_{\mathrm{sub}}(c), & z=\mathrm{s},\\
  \max_{j\in\mathcal J_{\mathrm{audit}}}m_{\mathrm{dil}}(c,j),
  & z=\mathrm{d}.
  \end{cases}
\end{align*}
\IRIS{} selects a lowest-budget eligible probe $c^\star$ for the chosen mode, breaking ties by lower expected response cost, and denotes the selected budget $m^\star(c^\star;z)$ by $m^\star$. The selected $m^\star$ fixes the online standards: $\tau_{\mathrm{mean}}(c^\star,m^\star)$ for substitution and, for each audited $j$, $(\tau_{c^\star j},b_j^\star)$ for dilution, with $b_j^\star=b_\alpha(m^\star,\alpha_{1,c^\star j})$.

\subsection{Audit Execution}
\label{sec:audit-execution}
Audit Execution is the only stage that queries the suspect endpoint. It sends probe $c^\star$ exactly $m^\star$ times, collects $\{y_i\}_{i=1}^{m^\star}$, computes the score and posterior quantities fixed during Audit-Plan Construction, and applies only the decision rule for mode $z$.

\paragraph{Substitution audit execution.}
When $z=\mathrm{s}$, \IRIS{} computes $S_{m^\star}(c^\star)$ by Eq.~\eqref{eq:evidence} and sets
$\mathrm{flag}=\mathbf{1}\{S_{m^\star}(c^\star)>\tau_{\mathrm{mean}}(c^\star,m^\star)\}$. All responses are treated as coming from one backend and attributed by aggregated log posterior,
\begin{align}
  \hat M_{\mathrm{sub}}
  &=\arg\max_{M\in\mathcal M}
  \sum_{i=1}^{m^\star}\log \hat{\mathbb P}_{c^\star}(M\mid y_i).
\label{eq:posterior-attribution}
\end{align}
For a flagged episode, the substitute is named as $\hat M_{\mathrm{sub}}$ if $\hat M_{\mathrm{sub}}\neq P$, or $\varnothing$ if the posterior favors the reference. 

\paragraph{Dilution audit execution.}
Each candidate $j\in\mathcal J_{\mathrm{audit}}$ uses its frozen tell threshold:
$k_j=\sum_{i=1}^{m^\star}\mathbf{1}\{s_{c^\star}(y_i)>\tau_{c^\star j}\}$,
$\hat f_j=k_j/m^\star$, and $\mathrm{flag}_j=\mathbf{1}\{k_j\ge b_j^\star\}$.

Let $\mathcal J_+=\{j:\mathrm{flag}_j=1\}$ and $\mathrm{flag}=\mathbf{1}\{\mathcal J_+\neq\emptyset\}$. For every flagged candidate with $q_{c^\star j}>\alpha_{1,c^\star j}$, the same tell fraction estimates the routing fraction. Suppressing the fixed $(c^\star,j)$ indices,
\begin{align}
  \hat\epsilon&=\mathrm{clip}_{[0,1]}\frac{\hat f-\alpha_1}{q-\alpha_1}.
\label{eq:epshat}
\end{align}
For an enrolled diluent, Eq.~\eqref{eq:epshat} gives a calibrated estimate; if the diluent is unenrolled, $q$ is unknown and $q\le1$ gives the lower bound $\hat\epsilon_{\mathrm{lb}}=\mathrm{clip}_{[0,1]}((\hat f-\alpha_1)/(1-\alpha_1))$. A delta-method interval propagates uncertainty in $(\hat f,\alpha_1,q)$.

The branch identifies the diluent from the full posterior vectors of the same batch. Let $\bar v=(m^\star)^{-1}\sum_i u_{c^\star}(y_i)$ and abbreviate the frozen signature $\bar g_{M,c^\star}$ as $\bar g_M$. The execution-stage attribution weights $\hat\eta$ are the nonnegative solution to the least-squares problem
\begin{align}
\min_{\mathbf{1}^{\top}w\le1}
  \Big\|\bar v-(1-\mathbf{1}^{\top}w)\bar g_P-\sum_jw_j\bar g_{R_j}\Big\|_2^2 .
\label{eq:signature-unmixing}
\end{align}
Among flagged candidates, $\hat j=\arg\max_{j\in\mathcal J_+}\hat\eta_j$. Let $h=\max_{M\in\mathcal M}\bar v_M$ and let $\tau_{\mathrm{os}}$ be its fifth percentile on matched held-out enrolled episodes. The diluent is named as $R_{\hat j}$ when $h\ge\tau_{\mathrm{os}}$ and left unknown otherwise.

Algorithm~\ref{alg:iris} uses the thresholds defined above: $\tau_{\mathrm{mean}}$ for substitution, and $(\tau_{cj},b_j^\star,\tau_{\mathrm{os}})$ for dilution.

\begin{breakablealgorithm}
\caption{The \IRIS{} audit}
\label{alg:iris}
\small
\algrenewcommand\algorithmiccomment[1]{\(\triangleright\)\ \textit{#1}}
\begin{algorithmic}[1]
\Require $P$, $\{R_j\}$, suspect endpoint, $\mathcal C$, $\mathcal J_{\mathrm{audit}}$, $z$, $\alpha$, $\delta$, $\gamma$, $\epsilon_{\min}$
\Ensure $(\mathrm{flag},\hat\epsilon,\hat R)$
\Statex \Comment{Stage I: construct the audit plan}
\For{$c\in\mathcal C$}
    \State collect labeled responses from $P$ and $\{R_j\}$
    \State train $g_c$; store score $s_c$ and posterior signatures
    \State calibrate the thresholds and set $m_c\leftarrow m^\star(c;z)$
\EndFor
\State $c^\star\leftarrow\arg\min\{m_c:c\text{ is eligible}\}$
\State $m^\star\leftarrow m_{c^\star}$; freeze the thresholds for $c^\star$
\Statex \Comment{Stage II: execute the frozen audit plan}
\State collect $\{y_i\}_{i=1}^{m^\star}$ from the suspect endpoint using $c^\star$
\If{$z=\mathrm{s}$}
    \State apply the frozen substitution rule to obtain $(\mathrm{flag},\hat R)$
    \State $\hat\epsilon\leftarrow\mathrm{flag}$
\Else
    \For{$j\in\mathcal J_{\mathrm{audit}}$}
        \State compute $\mathrm{flag}_j$, $\hat\epsilon_j$, and lower bound $\ell_j$
    \EndFor
    \State $\mathcal J_+\leftarrow\{j\in\mathcal J_{\mathrm{audit}}:\mathrm{flag}_j=1\}$
    \State $(\mathrm{flag},\hat\epsilon,\hat R)\leftarrow(\mathbf{1}\{\mathcal J_+\neq\emptyset\},0,\varnothing)$
    \If{$\mathrm{flag}=1$}
        \State compute $\hat\eta$; set $\hat j\leftarrow\arg\max_{j\in\mathcal J_+}\hat\eta_j$
        \State $(\hat\epsilon,\hat R)\leftarrow(\hat\epsilon_{\hat j},R_{\hat j})$ if the gate accepts; else $(\ell_{\hat j},\varnothing)$
    \EndIf
\EndIf
\State \Return $(\mathrm{flag},\hat\epsilon,\hat R)$
\end{algorithmic}
\end{breakablealgorithm}

\paragraph{Service deployment.}
When \IRIS{} is exposed as an on-demand audit service, Audit-Plan Construction can run as a low-frequency background job that periodically refreshes. Each user request then triggers only Audit Execution using the latest frozen plan. This temporal decoupling amortizes construction cost across audits and keeps user-triggered service fast and inexpensive.

\section{Audit Sample Complexity}
\label{sec:theory}

The audit budget is a sample-complexity question: at false-positive rate $\alpha$ and miss probability $\delta$, how many independent suspect queries are needed, before the audit begins, to distinguish the suspect endpoint from the trusted reference? Evidence accumulates across the number of queries $m$, not the length of one response (Prop.~\ref{prop:len}). For substitution, the mean score $S_m$ has a positive error exponent (Prop.~\ref{prop:exp}), so verification error falls exponentially in $m$ and a small pilot can estimate the rate. For dilution, only an $\epsilon$ fraction of queries carries substitute evidence, so the response-level tell count becomes the right statistic, and the gap between substitute and reference tell rates sets the cost (Thm.~\ref{thm:mix}): a wide margin gives the favorable $m=\Theta(\epsilon^{-1})$ regime, a narrow margin the $m=\Theta(\epsilon^{-2})$ wall. All proofs are in App.~\ref{app:proofs}.

\paragraph{Response model.}
Open-ended generation ranges over a model's full tokenizer vocabulary, whose support is too large and task-dependent for a direct audit; the probes of Section~\ref{sec:audit-plan-construction} instead concentrate the visible output on a small known alphabet or range, so the parsed symbols support tractable score calibration, tell-rate estimation, and sample-complexity analysis. Fix a probe context $c$ and a requested response length $L$; a query to an endpoint serving model $M$ returns a parsed visible response $y=y_{1:L}$ drawn from the autoregressive law
\begin{equation}
  P_M(y\mid c)=\textstyle\prod_{t=1}^{L} P_M\!\left(y_t\mid y_{<t},c\right).
  \label{eq:ar}
\end{equation}
Under the fixed probe $c$, let $P$ and $R$ denote the reference and substitute response distributions. Symbols within one response are dependent through $y_{<t}$, whereas separate API calls are independent; the audit budget therefore counts calls $m$, not tokens within one call.

\subsection{Substitution}

Four rates play distinct roles: the unattainable oracle benchmark, the information retained by the scalar score, the rate achieved by the deployed fixed-threshold mean test, and the empirical slope \IRIS{} uses for budget prediction.

\begin{definition}[Audit exponents]
\label{def:I}
For a probe $c$, response laws $P,R$, and score $s(y)=-\log\hat{\mathbb P}(P\mid y)$, define:
\begin{enumerate}\itemsep2pt
\item The \emph{oracle exponent} is the full-response Chernoff rate
\[
  I^\star=-\min_{0\le\lambda\le1}\log\sum_y P(y)^{1-\lambda}R(y)^\lambda .
\]
\item The \emph{score exponent} $I^{\mathrm{sc}}$ is the same Chernoff information after mapping the response to the scalar score $s(y)$.
\item The \emph{mean exponent} $I^{\mathrm{mean}}$ is the large-deviations rate at which the deployed fixed-threshold test on the score mean $S_m=\frac1m\sum_{i=1}^m s(y_i)$ decays, i.e.\ the maximin crossing $\max_\tau\min\{I_P^\ast(\tau),I_R^\ast(\tau)\}$ of the Cram\'er rate functions of $s$ under $P$ and $R$ (App.~\ref{app:proofs}, Eq.~\eqref{eq:Imean}).
\item The \emph{pilot exponent} $\widehat I_{\mathrm{auc}}$ is the fitted slope of $\log(1-\mathrm{AUROC}(m))$ on calibration episodes.
\end{enumerate}
\end{definition}

The pilot rate is a budgeting fit, not an information-theoretic lower bound: because AUROC compares an $m$-sample statistic from $P$ with an independent one from $R$, $\widehat I_{\mathrm{auc}}$ is not generally ordered below $I^{\mathrm{sc}}$. \IRIS{} therefore budgets against its lower-confidence bound (Eq.~\eqref{eq:sub-budget}) and verifies the transfer empirically.

\begin{proposition}[Mean-score separation]
\label{prop:exp}
Assume the fixed \IRIS{} score $s(y)$ is bounded and has different means under $P$ and $R$. Then the episode mean $S_m$ admits a fixed threshold whose error decays as
\begin{align}
  P_{\mathrm{err}}(S_m)=\exp\{-I^{\mathrm{mean}}m(1+o(1))\}, ~I^{\mathrm{mean}}>0 .
\end{align}
The deployed rate is bounded by the score-level and full-response oracle rates:
$I^{\mathrm{mean}}\le I^{\mathrm{sc}}\le I^\star$.
\end{proposition}

The oracle rate $I^\star$ provides comparison benchmarks: a full-string Bayes test would achieve $e^{-I^\star m(1+o(1))}$, and $K$-way identification among enrolled endpoints pays roughly $m\gtrsim\log(K/\delta)/I^\star_{\min}$.

\begin{proposition}[Query accumulation]
\label{prop:len}
Let $I_{\mathrm{seq}}(L)$ be the Chernoff information in one full length-$L$ response. For fixed $L$, $m$ independent API calls accumulate linearly,
\[
  C(P^{\otimes m},R^{\otimes m})=m\,I_{\mathrm{seq}}(L).
\]
Longer same-context responses can only add oracle information, but one autoregressive response is not $L$ independent position tests.
\end{proposition}

The length non-monotonicity observed in Section~\ref{sec:exp} concerns the deployed score and prompt-length conditions, not the oracle monotonicity in Prop.~\ref{prop:len}.

\subsection{Dilution}

Under dilution, each suspect query is drawn i.i.d.\ from the mixture $Q_\epsilon$ of Eq.~\eqref{eq:mixture}; the mean statistic $S_m$ still detects large shifts, but low-rate dilution is better read from the tells of Section~\ref{sec:audit-plan-construction}. For a fixed response-level threshold $\tau$ on the score $s(y)$, the population honest and substitute tell rates are $\alpha_1(\tau)=\mathbb P_{y\sim P}(s(y)>\tau)$ and $q(\tau)=\mathbb P_{y\sim R}(s(y)>\tau)$.

\begin{theorem}[Tell-rate budget]
\label{thm:mix}
Fix $\alpha,\delta,\epsilon$ and a threshold $\tau$. Dilution detection is governed by the honest and substitute tell rates $(\alpha_1(\tau),q(\tau))$:
\begin{enumerate}\itemsep2pt
\item[\textup{(a)}] If the suspect routes at least an $\epsilon$ fraction to $R$, the any-tell test (flag if any tell occurs) detects within $m$ queries with probability at least $1-(1-\epsilon q(\tau))^m\ge1-e^{-\epsilon q(\tau)m}$, so
\begin{equation}
  m^\star=\Big\lceil \ln(1/\delta)\,/\,\big(\epsilon\,q(\tau)\big)\Big\rceil
  \label{eq:mstar}
\end{equation}
queries guarantee power at least $1-\delta$.
\item[\textup{(b)}] On an honest endpoint, the false-positive probability is at most $1-(1-\alpha_1(\tau))^m\le m\alpha_1(\tau)$; hence $\alpha_1(\tau)\le\alpha/m$ controls type-I error at level $\alpha$.
\item[\textup{(c)}] A separating tail with $q(\tau)=\Omega(1)$ and $\alpha_1(\tau)=o(\epsilon)$ gives the favorable $m^\star=\Theta(\epsilon^{-1}\log(1/\delta))$ law. Without such a tail, light-tail low-separation mixtures face the local budget
\begin{equation}
  m^\star=\Theta\!\Big((z_{1-\alpha}+z_{1-\delta})^2\big/\big(\epsilon^{2}\,\chi^2(R\|P)\big)\Big)=\Theta(\epsilon^{-2}),
  \label{eq:eps2}
\end{equation}
with an all-test $\Omega(\epsilon^{-2})$ floor for sufficiently low-separation substitutes.
\end{enumerate}
\end{theorem}

The near-$1/\epsilon$ case is a measured tail event, not a generic consequence of $P\neq R$: low-entropy autoregressive conditionals can put most $P$ mass on a thin slice of response space, so an $R$-typical response may fall deep in the reference tail, with $q(\tau)$ staying $\Omega(1)$ while $\alpha_1(\tau)$ is pushed near zero. This is the regime behind many gateway dilutions. With finite enrollment data, confidence bounds on $(\alpha_1,q)$ determine whether the requested cheap audit is certified for the pair or must fall back to the $\epsilon^{-2}$ mean-shift wall of Eq.~\eqref{eq:eps2}.

\paragraph{Probe-design consequences.}
First, requested length is an unreliable budget axis: the token-optimal length maximizes the per-token tell yield $q(L)/L$ rather than simply growing $L$ (Cor.~\ref{cor:len}, App.~\ref{app:length}). Second, temperature insensitivity is intended specificity, not a blind spot: a pure decoding-temperature retune is a rank-one, on-family move with only second-order separation, and at the greedy boundary repeated single draws add no evidence (Cor.~\ref{cor:greedy}, Prop.~\ref{prop:temp}; App.~\ref{app:temp},~\ref{app:tempmag}), so \IRIS{} treats retunes as out of scope rather than as dilution. A true backend substitution, by contrast, perturbs the logits off that on-family curve and can supply first-order evidence when the pair is measurably separated.

\section{Experiments}
\label{sec:exp}

\begin{figure}[t]
  \centering
  \includegraphics[width=0.96\columnwidth]{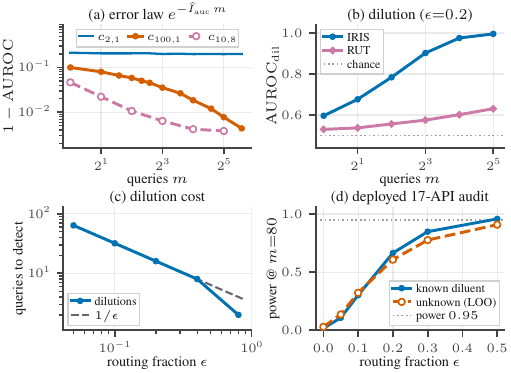}
  \caption{Experimental overview. \textbf{(a)}~Verification error falls as $e^{-\widehat I_{\mathrm{auc}}m}$: the ladder exemplar $c_{100,1}$, the same decay on $17$ commercial APIs ($c_{10,8}$, dashed), and the entropy-capped $c_{2,1}$ (flat). \textbf{(b)}~Dilution evidence accumulates for \IRIS{} but not for RUT ($c_{2,16}$; MET and FLIPS cannot audit dilution, Table~\ref{tab:baseline-main}). \textbf{(c)}~Margin-qualified commercial dilutions cost near-$1/\epsilon$ queries to detect. \textbf{(d)}~Deployed $m{=}80$ gateway audit: known and leave-one-out (unknown) diluent power nearly coincide; $\epsilon{=}0$ realizes the honest FPR $0.017$.}
  \label{fig:mscaling}
\end{figure}

No public benchmark targets gateway substitution/dilution auditing, so we build testbeds following prior output-only auditors~\citep{met2025,rut2026,flips2026,substitution2025} and validate on real endpoints, corroborated by the third-party MET corpus (App.~\ref{app:extended}). The evaluation has two parts: a controlled local testbed (Ollama; $6{\times}$A40, Xeon~8358, $2\,$TiB RAM, Ubuntu~22.04) checks the sample-complexity predictions under same-family substitutions and hosts the baseline head-to-head; a commercial OpenRouter testbed then evaluates the same budget laws under constructed gateway dilution at scale.

\subsection{Controlled local testbed}
\label{sec:exp-local}
The local testbed hosts a deliberately hard $K{=}6$ same-family Qwen3 ladder (0.6B--32B). All six probes of $\mathcal C$ (Section~\ref{sec:audit-plan-construction}) run at $T{=}1.0$ ($120$ repeats per model--probe; a sweep covers $T\in\{0,0.5,1,1.5,2\}$). The classifier is a $300$-tree random forest over the $179$ visible-string features; evidence aggregates $m$ independent responses, and each curve averages $40$ stratified splits with $95\%$ bootstrap intervals. Substitution/endpoint claims use the content-only (length/format-invariant) subset; law, budget, and probe-ranking analyses use the full set with a length ablation (Table~\ref{tab:featureset}, App.~\ref{app:featureset}).

\paragraph{Metrics.}
Detection is scored by \emph{FPR} and \emph{power} ($1{-}$miss), the empirical counterparts of $\alpha$ and $1{-}\delta$ from Section~\ref{sec:prelim}; attribution by \emph{accuracy} in naming the served model among the $K$ candidates; and quantification by the routing-fraction error $|\hat\epsilon-\epsilon|$. \emph{AUROC} ranks honest against deviating $m$-query episodes by the evidence $S_m$ ($0.5$ chance, $1$ perfect), with AUROC$_{\rm sub}$ and AUROC$_{\rm dil}$ denoting substitution and dilution (at the stated $\epsilon$), respectively. \emph{Hit@B} is the fraction of (claimed, candidate) pairs whose audit meets the target reliability at mean query budget $B$; we write $X@m$ for metric $X$ at budget $m$.

\paragraph{Queries, not length.}
Verification error decays exponentially in queries, as Prop.~\ref{prop:exp} predicts: $1-\mathrm{AUROC}$ falls log-linearly, matching the empirical $e^{-\widehat I_{\mathrm{auc}} m}$ rank-error law (Fig.~\ref{fig:mscaling}a; $c_{100,1}$ fit $R^2{=}0.96$), and attribution accuracy rises monotonically with independent queries (Prop.~\ref{prop:len}; $0.62\!\to\!0.95$ over $m{=}1{\to}48$). The fingerprint is randomness, not formatting or length: dropping length/format features ($179{\to}144$) leaves the lead sequence probes $c_{2,16}$ and $c_{10,8}$ effectively unchanged at $m{=}8$ (${\ge}.998$).

\paragraph{Budgets and probes from a cheap pilot.}
Fit on a cheap pilot ($m\le4$), $\widehat I_{\mathrm{auc}}$ predicts held-out $1{-}\mathrm{AUROC}$ at larger $m$ to a median $|\log$-ratio$|$ of $0.19$ over $224$ points, allowing $m_{\mathrm{sub}}$ to be frozen before any suspect query (conservative LCB; App.~\ref{app:lcb}). The same pilot orders probes by realized efficiency: $\rho(\widehat I_{\mathrm{auc}},\mathrm{acc}@8){=}0.82$ over $13$ probes, the six enrolled plus seven further $c_{n,L}$ variants ($p{<}10^{-3}$; App.~\ref{app:probeopt}), with high-rate sequence probes reaching AUROC$\,{\ge}0.95$ in a single query while $c_{2,1}$ never does within $48$. The exponent (Def.~\ref{def:I}; Prop.~\ref{prop:Isc}, App.~\ref{app:score}) also follows the predicted entropy--temperature gate (Prop.~\ref{prop:temp}, Cor.~\ref{cor:greedy}; Fig.~\ref{fig:temp}, App.~\ref{app:temp}): it is $\approx0$ for the one-bit $c_{2,1}$, peaks near $T{\approx}1$, and vanishes under greedy $T{=}0$. Because a pure temperature retune is only second order in this gate (Prop.~\ref{prop:tempmag}), \IRIS{} does not false-flag honest retunes ($\hat\epsilon{=}0.04$ vs.\ $0.53$ for a base-model swap; App.~\ref{app:knobs},~\ref{app:quant}).

\paragraph{Comparison with prior auditors.}
This head-to-head uses the controlled $K{=}6$ Qwen3 ladder, where FLIPS~\citep{flips2026}, MET~\citep{met2025}, and RUT~\citep{rut2026} can all run (MET needs reference samples, RUT log-ranks); Table~\ref{tab:baseline-main} lists probe and budget settings. \IRIS{} gives the strongest results on nearly every reported metric and is the only compared method that simultaneously detects substitution, estimates the routing fraction $\hat\epsilon$, and attributes the served model. MET-MMD is degenerate at a single draw and RUT stays weak on dilution (Fig.~\ref{fig:mscaling}b). KBF~\citep{kbf2026} and B3IT~\citep{b3it2026} use target-specific probes (factual recall, border inputs) we cannot reproduce on this shared random probe, so App.~\ref{app:auditors} compares them qualitatively. Because every compared baseline fixes its query count, we also isolate budget estimation itself: at matched mean budget $m{\approx}9$ on $c_{10,1}$, pilot-driven allocation hits the target (Hit@B) on $0.87$ of pairs versus $0.73$ for fixed $m$ (decomposition in App.~\ref{app:baselines}).
\begin{table}[t]
  \centering\footnotesize\setlength{\tabcolsep}{4pt}
  \begin{tabular}{@{}lcccc@{}}
    \toprule
    Metric & \IRIS{} & FLIPS & MET & RUT \\
    \midrule
    AUROC$_{\rm sub}$ & $\mathbf{.993}/.9997$ & $.990/\mathbf{.9999}$ & $.111/.993$ & $.634/.726$ \\
    AUROC$_{\rm dil}$  & $\mathbf{.902}/\mathbf{.995}$ & -- & -- & $.575/.631$ \\
    $\hat\epsilon$ err.& $\mathbf{.04}$ & -- & -- & -- \\
    Attr.              & $\mathbf{.927}/.998$ & $.921/\mathbf{.999}$ & -- & -- \\
    Hit@B              & $\mathbf{.87}@9$ & fixed & fixed & fixed \\
    \bottomrule
  \end{tabular}
  \caption{Head-to-head with prior black-box auditors on the shared $K{=}6$ Qwen3 ladder and probe $c_{2,16}$ ($30$ ordered pairs). AUROC rows give two budgets ($m{=}1/8$ for substitution; $m{=}8/32$ for dilution at $\epsilon{=}0.2$); Attr.\ is closed-set attribution accuracy ($m{=}1/8$); \textbf{bold} is the column-wise best. ``--'' = task not supported or not reported; ``fixed'' = a fixed, non-estimated budget.}
  \label{tab:baseline-main}
\end{table}

\subsection{Commercial models via OpenRouter}
\label{sec:exp-wild}

The local laws transfer to the $17$ commercial and open models accessed through OpenRouter's API (Fig.~\ref{fig:mscaling}a, dashed): $c_{10,8}$ attribution rises from $0.68$ to $0.98$ over $m{=}1{\to}32$ and AUROC reaches $0.99$ by $m{=}8$, the same Prop.~\ref{prop:exp} signature as on the ladder. Entropy gating persists, requested length again shows no stable relationship with accuracy, and the pattern remains consistent when extended to $45$ models (App.~\ref{app:scale}). App.~\ref{app:gallery} renders the $53$ enrolled endpoints ($18$ families) as an iris gallery.

\paragraph{A near-$1/\epsilon$ dilution budget.}
Dilution across distinct commercial models follows near-$1/\epsilon$ scaling (Fig.~\ref{fig:mscaling}c): in the AUROC-rank diagnostic on $c_{10,8}$ (the median over pairs of the smallest $m$ whose honest-vs-diluted episode AUROC reaches $0.95$), $\epsilon{=}5/10/20/40\%$ is detected in roughly $64/32/16/8$ queries. The deployed fixed-$p_0$ audit below (the honest tell rate $p_0$ is frozen in advance) instead pays the $\epsilon^{-2}$ rate for type-I robustness, but a pre-fixed-FPR test reproduces the same near-$1/\epsilon$ slope given enough reference calibration (App.~\ref{app:fixedfpr}); same-model temperature twins sit at the $\epsilon^{-2}$ wall. This is the $\Theta(\epsilon^{-1}\log(1/\delta))$ law of Thm.~\ref{thm:mix} when a separating tail exists, with the $1/\epsilon$-to-$\epsilon^{-2}$ crossover set by the measured tail exponent (Prop.~\ref{prop:phase}, App.~\ref{app:tail}). Across the $17$- and $45$-model pools, $62$--$85\%$ of margin-qualified pairs are consistent with the exact $1/\epsilon$ law, and budgeting uses each pair's measured exponent (App.~\ref{app:betastab}).

\paragraph{Detection, $\hat\epsilon$, and unseen diluents.}
\label{sec:extra}
We emulate a diluting gateway from real OpenRouter responses: for each ordered (claimed $P$, diluent $R$) pair among the $17$ models, the diluted stream draws each query i.i.d.\ from $Q_\epsilon$ in Eq.~\eqref{eq:mixture}. We fix $\tau$ and the honest tell rate $p_0$ on a fresh honest window of $P$, then run a content-only $m{=}80$ two-proportion test on held-out diluted streams; this is the deployed fixed-$p_0$ audit anticipated above, not the any-tell $1/\epsilon$ diagnostic. At $\epsilon{=}0.3$, \IRIS{} detects the $218/272$ margin-qualified pairs at mean power $0.85$, names the enrolled substitute, and keeps the pooled in-distribution FPR at $0.017$ (Fig.~\ref{fig:mscaling}d). Sweeping $\epsilon\in[0,0.5]$, power rises monotonically and $\hat\epsilon$ tracks the truth; App.~\ref{app:dilscale} adds five live OpenRouter audits, FPR hardening, and a held-out margin-selection check. The $54$ low-margin pairs are known-hard cases: close relatives that may share training lineage (within Qwen3, GPT-4 or Gemini--Gemma), and claimed models whose single-response fingerprint is itself non-distinctive. The fraction estimate tracks truth to $|\text{bias}|\!\approx\!0.04$ conditional on the response library, with near-nominal though slightly optimistic coverage (App.~\ref{app:estimateq}); App.~\ref{app:e2e} audits the population $(\alpha,\delta)$ guarantee (Prop.~\ref{prop:e2e}). Detection does not require enrolling the diluent: leave-one-out barely changes power ($0.78$ vs.\ $0.85$), turns $\hat\epsilon$ into a lower bound, and flags the never-enrolled \texttt{qwen-2.5-72b}.

\paragraph{Real endpoints and stress tests.}
A \emph{real} cross-provider audit (one open-weight model pinned across OpenRouter providers) flags $14/15$ provider pairs as distinguishable, real quantization/kernel deviations rather than author-injected, corroborated on the third-party MET corpus (App.~\ref{app:extended}). The appendix adds the checks a referee expects: probe-paraphrase and probe-aware-gateway robustness, non-i.i.d.\ routing, a \texttt{q4}-for-\texttt{fp16} cheat, knob identifiability, $45$-model scale, exact-binomial FPR hardening, a one-class detector, multi-diluent unmixing, and matched-budget baselines with significance tests (App.~\ref{app:auditors}--\ref{app:extended}).

\section{Conclusion}
\label{sec:conclusion}

\IRIS{} is a budgeted black-box framework for auditing model substitution and routing dilution in LLM gateways from returned text alone: random-generation probes and an estimate-then-budget pilot predict audit difficulty before querying traffic and reuse those responses for detection, attribution, and routing-fraction estimation. Our analysis explains exponential evidence accumulation and the near-$1/\epsilon$ versus $\epsilon^{-2}$ dilution budgets, which a Qwen3 ladder, a $45$-model library, a live cross-provider audit of real deviations, and third-party MET traces confirm, with appendices hardening false-positive control and unseen diluents.


\putbib
\end{bibunit}
\fi

\ifIRISIncludeSupplement
\ifIRISIncludeMain
\clearpage
\fi
\setcounter{page}{1}
\setcounter{section}{0}
\setcounter{subsection}{0}
\setcounter{figure}{0}
\setcounter{table}{0}
\setcounter{equation}{0}
\setcounter{algorithm}{0}
\setcounter{theorem}{0}
\setcounter{proposition}{0}
\setcounter{corollary}{0}
\setcounter{definition}{0}
\setcounter{lemma}{0}
\renewcommand{\thesection}{S\arabic{section}}

\begin{bibunit}
\twocolumn[
\begin{center}
{\Large\bf Supplementary Material for\\
\IRIS{}: Budgeted Black-Box Auditing of Model Substitution and Routing Dilution in LLM Gateways\par}
\vspace{0.5em}
{\normalsize Yuewei Zhang, Zhi-Hai Zhang, Hanzhang Qin\par}
\end{center}
]



\section{Temperature Gating}
\label{app:temp}
\begin{proposition}[Distributional separation]
\label{prop:temp}
For a single-draw probe at temperature $T$ with categorical laws $p^T_P,p^T_R$ on $n$ symbols:
\begin{enumerate}\itemsep2pt
\item One draw carries at most the visible-output information $I(\Theta;Y)\le H(Y)\le\log n$, while the testing exponent satisfies $I^\star\le\min\!\big(\mathrm{KL}(p^T_P\|p^T_R),\mathrm{KL}(p^T_R\|p^T_P)\big)$.
\item At the greedy boundary $T\to0$, repeated identical draws have zero additional-query accumulation exponent.
\end{enumerate}
\end{proposition}

\noindent Thus a small alphabet caps the mutual information in one draw, but weak evidence also requires the two categorical laws to be close. The $c_{2,1}$ probe is capped at one bit per draw; it is not automatically useless. The greedy statement concerns accumulation from repetition: a single deterministic draw may still separate two endpoints, but repeated identical draws do not add independent evidence.

\begin{corollary}[Greedy has no accumulation]
\label{cor:greedy}
Under idealized greedy decoding ($\mathsf T{=}0$), repeated single draws are deterministic and identical, so the accumulation exponent from repeating the same draw is $I=0$.
\end{corollary}

\begin{figure}[t]
  \centering
  \includegraphics[width=0.86\columnwidth]{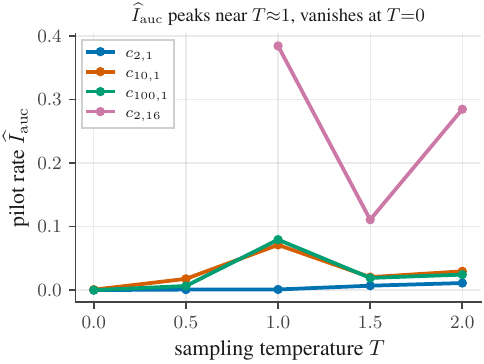}
  \caption{Pilot rank-error rate $\widehat I_{\mathrm{auc}}$ vs.\ sampling temperature $T$ for the entropy-ladder probes: $\widehat I_{\mathrm{auc}}\!\to\!0$ under greedy decoding, peaks near $T{\approx}1$, and the onset temperature scales with output entropy ($c_{2,1}$ only starts at $T{=}2$).}
  \label{fig:temp}
\end{figure}
Figure~\ref{fig:temp} shows the per-query information density $\widehat{I}_{\mathrm{auc}}$ as a function of sampling temperature, for the entropy-ladder probes of Section~\ref{sec:exp}. $\widehat{I}_{\mathrm{auc}}\!\to\!0$ under greedy decoding (repeated single draws are identical), rises to a probe-dependent interior peak near $T{\approx}1$, and falls again as sampling noise grows; the temperature at which accumulation turns on scales with the probe's output entropy ($c_{2,1}$ only starts at $T{=}2$). This is the empirical basis for the entropy-/temperature-gating claim in the body; we do not assert unimodality of $I(T)$ in general (Prop.~\ref{prop:temp}), and note that $T{=}0$ is not perfectly deterministic on every backend, so the greedy $I{=}0$ statement concerns the increment from repeating an identical single draw. The highest-entropy probe $c_{2,16}$ is already saturated at $m{=}1$ for $T{\le}0.5$---a single draw separates the models---so its accuracy-decay exponent is undefined there (zero unsaturated points to fit) and only its $T{\ge}1$ estimates appear in Fig.~\ref{fig:temp}; this is a ceiling on the \emph{estimator}, not an absence of signal, and is itself consistent with entropy gating (the strongest probe saturates first).

\section{Dilution Estimation}
\label{app:estimateq}
Theorem~\ref{thm:mix} flags \emph{that} a suspect is diluted; here we detail the routing-fraction estimator $\hat\epsilon$ of Eq.~\eqref{eq:epshat}. Under $Q_\epsilon$ the per-query tell probability at a threshold $\tau$ is $p_\epsilon=(1-\epsilon)\alpha_1+\epsilon\,q$, with $\alpha_1=\mathbb P_{P}(s>\tau)$ and $q=\mathbb P_{R}(s>\tau)$; inverting the observed tell fraction $\hat f$ over $m$ queries gives the method-of-moments estimate $\hat\epsilon=\mathrm{clip}_{[0,1]}\big((\hat f-\alpha_1)/(q-\alpha_1)\big)$ of Eq.~\eqref{eq:epshat}, identifiable whenever the diluent separates ($q-\alpha_1>0$). We fix $\tau$ at the $(1{-}\alpha_0)$ reference quantile ($\alpha_0{=}0.1$, so there is no max-gap snooping) and estimate $\alpha_1,q$ on an enrollment split disjoint from the audited responses; a delta-method $95\%$ interval propagates the binomial tell variance \emph{and} the anchor ($\alpha_1,q$) uncertainty.

On controlled mixtures of the $K{=}6$ ladder ($30$ reference--diluent pairs passing the margin check, $m{=}128$, probe $c_{2,16}$, $6000$ episodes/cell), $\hat\epsilon$ tracks the true fraction with mean $|\text{bias}|\le0.02$ and $95\%$-CI coverage $0.92$--$0.96$ across $\epsilon\in[0,0.75]$, and the $\epsilon{=}0$ control returns $\hat\epsilon{\approx}0.02$ (Fig.~\ref{fig:estimateq})---\IRIS{} reports \emph{how much} a stream is diluted, not merely that it is. Low-margin pairs ($q\!\to\!\alpha_1$) are excluded as non-identifiable.

\paragraph{Coverage on the gateway.} We re-audit the interval on the full $17$-endpoint gateway (content-only features, $m{=}80$, the $268/272$ \emph{identifiable} pairs with $q{>}\alpha_1$---a looser screen than the deployed $q{-}p_0\ge0.15$ margin that retains $218$ in App.~\ref{app:dilscale}---held-out mixtures at known $\epsilon$). The point estimate stays essentially unbiased (mean $\hat\epsilon=0.110/0.204/0.303/0.501$ at $\epsilon{=}0.1/0.2/0.3/0.5$; signed bias $\le0.010$, per-pair $|\text{bias}|\!\approx\!0.04$), but the delta-method interval is \emph{mildly anti-conservative}: empirical coverage is $0.86$--$0.87$ at the nominal $90\%$ level and $0.92$--$0.93$ at $95\%$---the low end of the $K{=}6$ figure, under-covering by $\approx3$--$4$ points on the harder content-only gateway. Dropping the anchor-uncertainty terms (a Wilson interval on $\hat f$ alone) worsens it to $0.79$--$0.86$ at $90\%$, so the $(\alpha_1,q)$ variance is a non-negligible part of the honest width and must be retained. We therefore describe the interval as approximately calibrated in-distribution, but slightly optimistic rather than exact.

\begin{figure}[t]
  \centering
  \includegraphics[width=\columnwidth]{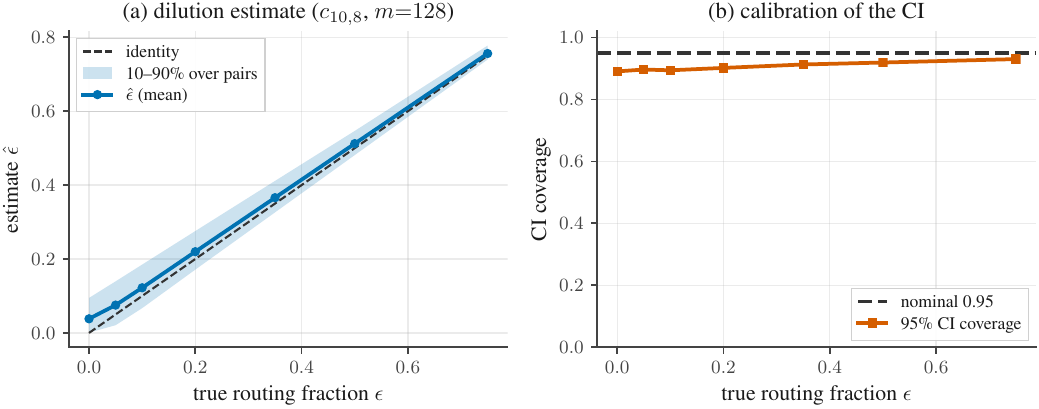}
  \caption{Routing-fraction estimation on controlled mixtures ($K{=}6$ ladder, $30$ margin-qualified pairs,
  $m{=}128$, $c_{2,16}$). Left: $\hat\epsilon$ vs.\ true $\epsilon$ (identity dashed; band is the
  $10$--$90\%$ spread over pairs). Right: empirical $95\%$-CI coverage vs.\ the nominal $0.95$.}
  \label{fig:estimateq}
\end{figure}

\section{Dilution at Scale}
\label{app:dilscale}
We scale the online controlled-dilution audit (Section~\ref{sec:extra}; dilutions we inject, not gateways found diluting) to the full $17$-endpoint OpenRouter library ($10{,}200$ real responses) and add the unknown-diluent regime. \textbf{Protocol.} Posteriors $\hat{\mathbb P}(\cdot\mid y)$ come from a $5$-fold out-of-fold content-only (length/format-invariant) random forest over the three sequence probes, scored as $s=-\log\hat{\mathbb P}(P\mid y)$. For each ordered (claimed $P$, diluent $R$) pair we use three disjoint splits of $P$: a calibration split fixes $\tau$ at the $(1{-}\alpha)$ quantile ($\alpha{=}0.05$); a fresh \emph{honest window} re-measures the null tell rate $p_0$ before the audit; and an audit split supplies the honest queries. The diluted stream draws each of $m$ queries from $R$ w.p.\ $\epsilon$ else $P$, and detection is a one-sided two-proportion test of the audit tell rate against $p_0$ (declared at $p{<}0.01$); the realized FPR is the empirical $\epsilon{=}0$ row of Table~\ref{tab:dilscale} ($0.02$), not the union-bound $\alpha/m$ used in the $m^\star$ derivation. \emph{This at-scale audit therefore runs the \textbf{fixed}-$p_0$ two-proportion test (local complexity $\Theta(\epsilon^{-2})$ by Prop.~\ref{prop:e2e}), \textbf{not} the shrinking-threshold any-tell whose $1/\epsilon$ budget is validated separately in App.~\ref{app:fixedfpr}}---these are distinct tests, and we do not claim the at-scale numbers exhibit the $1/\epsilon$ law. The calibrated estimate is $\hat\epsilon{=}(\hat f-p_0)/(q-p_0)$ with $q$ the diluent's tell rate (needs $R$ enrolled); the deployed lower bound $\hat\epsilon_{\mathrm{lb}}{=}(\hat f-p_0)/(1-p_0)$ uses only $P$. We report pairs passing the pre-audit margin check $q-p_0\ge0.15$ (a \emph{combined}-probe tell-rate margin; this differs from the \emph{per-probe} tail-plateau used in Table~\ref{tab:betastab}, which is why the count is $218/272$ here versus its $248/164/176$ for a single probe). The remaining $54/272$ pairs are marked low-margin/indeterminate; over \emph{all} $272$ ordered pairs the $\epsilon{=}0.3$ reliable-detection rate ($f_{.95}$) is correspondingly $\approx0.51$ ($0.64\times218/272$). These $54$ are not random: they concentrate on a few \emph{claimed} models (\texttt{qwen3-32b}, \texttt{qwen3-8b}, \texttt{gpt-4.1}, \texttt{gpt-4o-mini}, \texttt{gemini-2.5-flash} account for $48$) and split into two kinds. First, close relatives that plausibly share training lineage---within the Qwen3 family (8b/32b/235b), within the GPT-4 series (\texttt{gpt-4.1}/\texttt{gpt-4o-mini}), within Gemini~2.5, and Gemini--Gemma (which share a tokenizer)---whose random-generation fingerprints are genuinely similar. Second, claimed models whose own single-response fingerprint is not distinctive, so a foreign diluent rarely crosses the tell threshold ($q\!\approx\!0$); this is directional (a property of the defended model) and is \emph{not} size-specific, as it includes the flagship \texttt{gpt-4.1} and the mainstream \texttt{gemini-2.5-flash}. Both are limits of the per-response tell margin at the deployed budget, not of identifiability: full-query attribution still separates all $17$ models (verification AUROC $0.99$ by $m{=}8$).

\paragraph{Honest FPR.} The pooled mean above hides a right tail: over the $218$ margin-qualified pairs the per-pair honest ($\epsilon{=}0$) FPR has median $0.000$, mean $0.017$, $95$th percentile $0.090$, and maximum $0.557$, with $7.8\%$ of pairs above $0.05$. The tail is the $22$ pairs whose fresh honest window draws zero tells ($\hat p_0{=}0$), where the normal-approximation two-proportion test compares against an empirical null of $0$ and is anti-conservative. Replacing it with a one-sided \emph{exact} binomial test of the audited tell count against a Clopper--Pearson $(1{-}\gamma)$ \emph{upper} bound on $p_0$ ($\gamma{=}0.1$, with a $k_{\mathrm{hw}}{=}0$ guard giving $\bar p_0\!\approx\!0.023$ rather than $0$) cuts the tail---max $0.39$, $95$th percentile $0.018$, $2.8\%$ above $0.05$, pooled mean $0.0065$---while preserving power (known-diluent $\epsilon{=}0.3$: $0.85\!\to\!0.78$ under the exact test---a separate effect from the equal-valued leave-one-out power in Table~\ref{tab:dilscale}). One residual heavy-tailed pair still has FPR $0.39$: its honest-window split draws zero tells but the disjoint audit split has true honest rate $0.062$, and the count-only Clopper--Pearson bound corrects for sampling error in the window count, not for this audit-side dispersion---so we report the whole distribution, not the mean.

\paragraph{Margin validity.} To rule out that choosing the reported subset and evaluating power on the same data inflates it, we re-apply the margin check on a \emph{third disjoint} selection split and measure power and FPR on the disjoint test split. Membership is stable ($223/272$ vs.\ the $218$ in-sample, $210$ shared) and the held-out honest FPR is if anything stricter ($0.0085$ vs.\ $0.0167$). The held-out $\epsilon{=}0.3$ power is $0.66$ versus the in-sample $0.85$; but a matched control that applies the margin check \emph{in-sample} on the same half-size audit bank also gives $0.66$, so the gap is a \emph{data-budget} artifact of halving the bank (less score diversity), not selection bias ($\le0.002$ power attributable to disjoint selection). We therefore report the full-bank numbers as the operating point, with this experiment establishing that selection circularity is not their source.

\textbf{Unknown diluent (leave-one-out).} To model a substitute we have never enrolled, we delete $R$ from the library and retrain on the remaining $16$ endpoints; $R$'s responses are scored out-of-distribution by the fold classifiers. Detection uses only $P$, so it is unchanged in principle; Table~\ref{tab:dilscale} confirms it degrades only slightly in practice. The foreign responses are also flagged out-of-library by open-set rejection (knownness $\max_c\hat{\mathbb P}(c\mid y)$ below the claimed model's $95\%$-accept gate) at mean rate $0.19$ versus a $0.04$ self-rate. The deployed $\hat\epsilon_{\mathrm{lb}}$ stays below the true $\epsilon$ in aggregate (Table~\ref{tab:dilscale}, last column), as it must when $q\le1$.

\begin{table*}[t]
  \centering\footnotesize\setlength{\tabcolsep}{3pt}
  \caption{Dilution audit over all margin-qualified (claimed, diluent) pairs of the $17$-endpoint gateway,
  $m{=}80$, content-only features, $400$ episodes/cell. \emph{Known}: diluent enrolled ($218$ pairs).
  \emph{Unknown}: leave-one-out, diluent deleted from the library ($198$ pairs).
  ``power'' is the fraction of audits flagged at $p{<}0.01$; ``$f_{.95}$'' the fraction of pairs detected
  at power${\ge}0.95$; $\hat\epsilon$ the calibrated estimate, $\hat\epsilon_{\mathrm{lb}}$ the claimed-only lower bound.}
  \label{tab:dilscale}
  \begin{tabular}{@{}c ccccc cc c@{}}
    \toprule
    & \multicolumn{5}{c}{known diluent} & \multicolumn{3}{c}{unknown (leave-one-out)} \\
    \cmidrule(lr){2-6}\cmidrule(lr){7-9}
    $\epsilon$ & power & $f_{.95}$ & med.\ $p$ & $\hat\epsilon$ & $\hat\epsilon_{\mathrm{lb}}$ & power & $f_{.95}$ & $\hat\epsilon_{\mathrm{lb}}$ \\
    \midrule
    $0.00$ & 0.02 & 0.00 & $5{\times}10^{-1}$ & 0.03 & 0.02 & 0.03 & 0.00 & 0.02 \\
    $0.05$ & 0.11 & 0.00 & $2{\times}10^{-1}$ & 0.06 & 0.04 & 0.14 & 0.01 & 0.04 \\
    $0.10$ & 0.30 & 0.01 & $5{\times}10^{-2}$ & 0.10 & 0.07 & 0.32 & 0.03 & 0.07 \\
    $0.20$ & 0.67 & 0.31 & $3{\times}10^{-3}$ & 0.20 & 0.14 & 0.61 & 0.31 & 0.14 \\
    $0.30$ & 0.85 & 0.64 & $4{\times}10^{-5}$ & 0.30 & 0.22 & 0.78 & 0.53 & 0.20 \\
    $0.50$ & 0.96 & 0.87 & $2{\times}10^{-9}$ & 0.50 & 0.36 & 0.91 & 0.80 & 0.34 \\
    \bottomrule
  \end{tabular}
\end{table*}

\begin{figure}[t]
  \centering
  \includegraphics[width=\columnwidth]{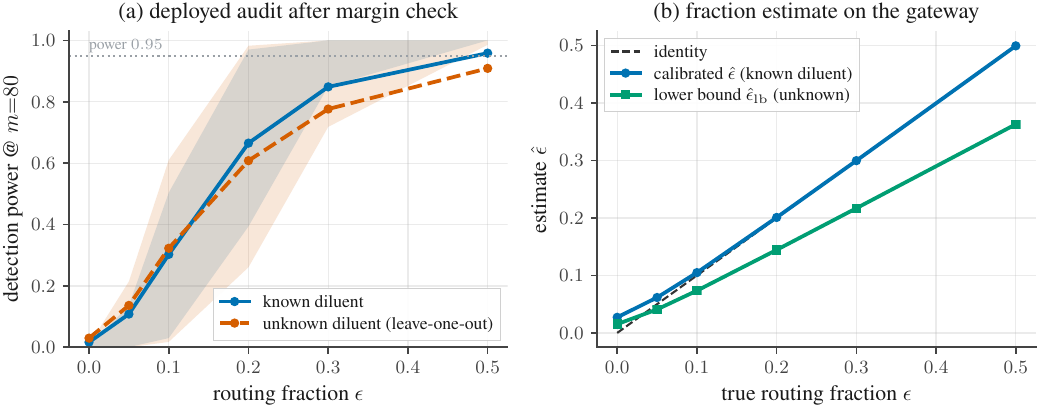}
  \caption{Deployed dilution audit on the real $17$-endpoint gateway. (a)~Detection power vs.\ routing fraction
  $\epsilon$ at $m{=}80$ under the fixed-$p_0$ two-proportion test, pooled over pairs passing the margin check (band $=25$--$75\%$ over pairs): the \emph{unknown}-diluent
  leave-one-out curve nearly coincides with the \emph{known} curve---detection needs only the claimed model's
  fingerprint. (b)~Routing-fraction recovery: the calibrated $\hat\epsilon$ tracks the identity, while the
  deployed lower bound $\hat\epsilon_{\mathrm{lb}}$ (valid for an unknown diluent) sits conservatively below it.}
  \label{fig:dilution}
\end{figure}

\begin{table*}[t]
  \centering\footnotesize\setlength{\tabcolsep}{6pt}
  \begin{tabular}{@{}llccc@{}}
    \toprule
    claimed & diluent & result & $p$ & $\hat\epsilon$ \\
    \midrule
    \multicolumn{5}{@{}l}{\emph{offline, real gateway}} \\
    \texttt{gpt-4o-mini} & \texttt{llama-3.1-8b} & $0.95$ & $2{\times}10^{-4}$ & $0.33$ \\
    \texttt{gemma-3-4b} & \texttt{gemma-3-27b} & $1.00$ & $4{\times}10^{-10}$ & $0.33$ \\
    \texttt{gemini-2.5-pro} & \texttt{qwen3-235b} & $1.00$ & $9{\times}10^{-10}$ & $0.31$ \\
    \texttt{claude-3.5-haiku} & \texttt{gpt-4o-mini} & $0.00$ & $6{\times}10^{-1}$ & --- \\
    \midrule
    \multicolumn{5}{@{}l}{\emph{live, end-to-end ($m{=}80$)}} \\
    \texttt{gpt-4o-mini} & honest ($\epsilon{=}0$) & clean & $0.31$ & --- \\
    \texttt{gpt-4o-mini} & \texttt{llama-3.1-8b} & diluted & $3{\times}10^{-3}$ & $0.26$ \\
    \texttt{llama-3.3-70b} & \texttt{gemma-3-4b} & diluted & $3{\times}10^{-5}$ & $0.30$ \\
    \texttt{gpt-4.1} & \texttt{qwen3-8b} & diluted & $2{\times}10^{-4}$ & $0.43$ \\
    \texttt{gpt-4o-mini} & \texttt{qwen-2.5-72b}$^\dagger$ & diluted & ${<}10^{-3}$ & $0.24^\ddagger$ \\
    \bottomrule
  \end{tabular}\\[2pt]
  {\footnotesize Live $\hat\epsilon$ is a single-stream point estimate, noisier than the offline calibration. $^\dagger$never enrolled (unknown diluent), a moderate twin: borderline at $m{=}80$ ($p\!\approx\!0.003$--$0.01$ across draws), robustly flagged at $m{\ge}160$ (shown: $m{=}240$, $p{=}10^{-3}$). $^\ddagger$lower bound $\hat\epsilon_{\mathrm{lb}}$.}
  \caption{Representative pairs at $\epsilon{=}0.3$, $m{=}80$ (top), and five \emph{online} end-to-end
  audits of controlled mixtures we inject through OpenRouter (bottom; author-constructed, not gateways found
  diluting; honest window measured first). The unknown substitute (\texttt{qwen-2.5-72b}) is not enrolled,
  so only its lower bound is available.}
  \label{tab:dilcases}
\end{table*}

\section{Query Complexity}
\label{app:fixedfpr}
Validating the $1/\epsilon$ law under a \emph{controlled} error rate, and deciding whether the measured tail exponent $\kappa$ is a population property or a finite-sample plateau, both turn on \emph{reference calibration size}. We address them on fresh, large reference batches (\texttt{qwen3:8b}@$T1$ as reference; \texttt{0.6b/4b/14b}@$T1$ and the temperature twin \texttt{8b}@$T0.9$ as substitutes; ${\approx}3000$ responses each on $c_{2,16}$/$c_{10,8}$). \emph{Protocol.} The reference is split three ways---a calibration split fixes $\tau$ at the per-query level $\alpha/m$, a \emph{disjoint} certification split bounds $\alpha_1$ by a one-sided Clopper--Pearson upper bound, and a held-out audit split measures the realized FPR---so the null is never tested on the data that chose it (the data-dependent-threshold objection). We report \emph{queries-to-detect} at power $0.95$ under a \emph{pre-fixed total} FPR ($1\%$ and $5\%$), \emph{not} AUROC, fit its $\epsilon$-exponent, and re-estimate $\kappa$ as the calibration grows from $n{=}50$ to $n{=}3000$.

\emph{(1) The $1/\epsilon$ law holds at a controlled FPR, not via AUROC.} Queries-to-detect at power $0.95$ under a \emph{pre-fixed total} FPR follow a log--log slope of $-1.0$ to $-1.12$ for the distinct substitutes (\texttt{0.6b},~\texttt{4b}) on both probes and at both $\alpha{=}0.05$ and $0.01$---the $1/\epsilon$ law, now against a fixed error rate rather than AUROC. The temperature twin (\texttt{8b}@$T0.9$) is undetectable at power $0.95$ within $256$ queries at \emph{every} $\epsilon$ (the wall), so the dichotomy survives a controlled-FPR test.

\emph{(2) But certifying the FPR needs thousands of reference samples.} Realizing the union-bound level $\alpha_1\le\alpha/m$ requires the reference's deep tail; a $1000$-sample calibration \emph{overshoots} (realized total FPR $0.09$--$0.14$ on $c_{2,16}$, though $\le0.04$ on the better-behaved $c_{10,8}$). The certifiable budget is $m\le\alpha/\overline{q}_0(N)$ with $\overline{q}_0(N)$ the one-sided Clopper--Pearson floor at $N$ samples \emph{assuming zero tail events} (a best-case envelope): $m\le16$ at $N{=}1000$, $m\le50$ at $N{=}3000$, $m\le166$ at $N{=}10{,}000$ ($\alpha{=}0.05$; an order of magnitude smaller at $\alpha{=}0.01$). These are the zero-event envelope; with the tail events a real reference exhibits, the achievable certified $m$ is smaller still for atom-heavy probes like $c_{2,16}$, whose realized FPR already exceeds nominal at the $N{=}1000$ calibration above. The $1/\epsilon$ regime is therefore \emph{deployable with a certified FPR} only down to the $\epsilon$ whose budget $N$ supports---$\epsilon\gtrsim0.1$ at a few thousand samples, smaller $\epsilon$ needing tens of thousands. This is the calibration cost the $1/\epsilon$ budget demands, stated explicitly rather than hidden behind AUROC; all the FPR/power figures here resample one finite held-out reference bank into episodes, so the intervals are conditional on that bank (as in App.~\ref{app:dilscale}).

\emph{(3) The tail exponent $\kappa$ is stable in the calibration size.} Re-estimating $\kappa$ as the reference grows from $n{=}50$ to $3000$, the distinct substitutes settle at a \emph{small, stable} value (\texttt{4b} at $\kappa\!\approx\!0$, \texttt{0.6b} at $0.13$) without moving toward $1$, while the adjacent-temperature stress test stays near $\kappa\!\approx\!1$ at every $n$. We therefore present $1/\epsilon$ as a measured, finite-$\epsilon$, calibration-bounded regime when the margin check passes---never an asymptotic guarantee.

\emph{(4) Continuous one-class scores.} To rule out that $\kappa\!\approx\!0$ is a random-forest discretization artifact, we re-measure $\kappa$ with two \emph{continuous}, reference-only scores---a standardized Mahalanobis density and an isolation-forest anomaly. The adjacent-temperature stress test remains low-separation ($\kappa\!\approx\!0.9$--$1.1$), while distinct substitutes keep a much smaller exponent under the Mahalanobis density and interpolate under isolation forests. We accordingly report $1/\epsilon$ as the behaviour of \IRIS{}'s deployed score after the margin check, not as a score-independent law.

\section{End-to-End Bound}
\label{app:e2e}

\begin{table}[t]
  \centering\small
  \setlength{\tabcolsep}{4pt}
  \caption{Which statistic backs which claim, and at what guarantee level. \IRIS{} couples a
  \emph{ranking}/difficulty estimator to a type-I-valid deployed test; the prescribed budget is a
  \emph{feasibility diagnostic}, not a coverage guarantee. (Referenced from the body.)}
  \label{tab:whichtest}
  \begin{tabular}{@{}p{0.27\columnwidth}p{0.30\columnwidth}p{0.33\columnwidth}@{}}
    \toprule
    claim / use & statistic & guarantee \\
    \midrule
    attribution, full-sub.\ verification & AUROC-slope $\widehat{I}_{\mathrm{auc}}$ (rank) & difficulty score; no coverage certificate \\
    dilution detection (deployed) & fixed-$p_0$ two-proportion & type-I valid given null; $\Theta(\epsilon^{-2})$ local \\
    $1/\epsilon$ at pre-fixed FPR & any-tell, CP-certified $\alpha_1$ & needs ${\sim}$thousands refs to certify \\
    routing fraction $\hat\epsilon$ & moments $+$ delta interval & approx.\ calibrated in-dist.\ (optimistic) \\
    end-to-end $(\alpha,\delta,\epsilon)$ & estimate-then-budget & type-I held; power met on a minority (diagnostic) \\
    \bottomrule
  \end{tabular}
\end{table}
\begin{proposition}[Binomial guarantee]
\label{prop:e2e}
Fix targets $\alpha,\delta\in(0,1)$, a minimum routing fraction $\epsilon$, and a per-response threshold $\tau$ with population rates $\alpha_1=\mathbb P_{P}(s{>}\tau)$ and $q=\mathbb P_{R}(s{>}\tau)$, $q>\alpha_1$. The level-$\alpha$ one-sided binomial test of the tell count against null rate $\alpha_1$ controls type-I error for every $m$. If $m^\star$ is the smallest stable budget whose binomial power against $p_\epsilon=(1-\epsilon)\alpha_1+\epsilon q$ is at least $1-\delta$ for every $m\ge m^\star$, then $\mathbb P_{Q_{\epsilon'}}(\mathrm{flag})\ge1-\delta$ for all $\epsilon'\ge\epsilon$ and $m\ge m^\star$.
\end{proposition}

\noindent Proposition~\ref{prop:e2e} is a population power calculation. It reduces to the any-tell $m^\star=\Theta(\epsilon^{-1}\log(1/\delta))$ regime only when a separating threshold gives $q=\Omega(1)$ and $\alpha_1=o(\epsilon)$; at a fixed null rate $\alpha_1=p_0>0$, the local budget is $\Theta\!\big(p_0(1-p_0)/(\epsilon^2(q-p_0)^2)\big)$. With estimated rates, the threshold must be fixed on a split separate from the one estimating the null, and Clopper--Pearson upper/lower bounds on $(\alpha_1,q)$ preserve the statement up to the chosen confidence. Certifying $\alpha_1\le\alpha/m$ is reference-sample intensive, which is why the empirical check below distinguishes the population guarantee from finite-calibration performance.

\medskip\noindent\textbf{Empirical check.} Proposition~\ref{prop:e2e} is type-I valid \emph{for the population} $\alpha_1$ (by the binomial test's validity, not a union bound); $m^\star$ sized for binomial power $\ge1-\delta$ at $p_\epsilon$ controls type-II. The finite-sample question is whether the \emph{estimated} $(\tau,\alpha_1,q)$ deliver this. We check it on the $17$-endpoint gateway ($c_{10,8}$, content-only), with train/calibration/audit splits disjoint at the \emph{raw-response} level so no episode is reused, at $(\alpha{=}0.05,\delta{=}0.1,\epsilon{=}0.3)$, calibrating against Clopper--Pearson bounds ($\gamma{=}0.1$). Across the $258/272$ feasible pairs the realized type-I is controlled \emph{in aggregate}---mean $0.013$, within $\alpha$ on $90\%$ of pairs---confirming the test-validity argument, at mean budget $m^\star{=}112$; realized power averages $0.73$ but meets the \emph{pre-registered} $0.9$ target ($\delta{=}0.1$) on only $17.8\%$ of feasible pairs ($12.6\%$ at $\epsilon{=}0.2$)---population-exact yet finite-sample under-powered on the median pair at this $m^\star$. The per-pair type-I carries a tail (mean $0.013$ but maximum $0.475$, with $10\%$ of pairs above $\alpha$), induced by the winner's-curse $m^\star$-minimization over the calibration grid $a_0\in\{0.02,0.05,0.1,0.15\}$ at the full confidence $\gamma$; a Bonferroni split (each Clopper--Pearson bound at $\gamma/4$ before minimizing) restores it---maximum $0.032$ at $\epsilon{=}0.2$ ($\sim$$0\%$ exceedance) and $0.168$ at $\epsilon{=}0.3$---at a larger budget (mean $m^\star$ $151\!\to\!223$ at $\epsilon{=}0.2$), whereas fixing a single $a_0$ a priori does \emph{not} control it (maximum $0.76$--$0.93$, $\sim$$18\%$ exceedance), confirming the tail is finite-calibration variance the union bound addresses, not grid optimism alone. The shortfall is thus a \emph{finite-sample} effect: each pair calibrates $(\tau,\alpha_1,q)$ on only $\sim$$33$ responses, so the $1{-}2\gamma$ confidence is loose and clumped scores transfer imperfectly between the calibration and audit halves. The guarantee is exact at the population level and as tight as the calibration sample allows, so a deployment should size its honest window accordingly. The infeasible $14$ pairs are returned \emph{indeterminate} rather than flagged.

\section{Multiple Diluents}
\label{app:multi}
A gateway may dilute one reference $P$ with a \emph{mixture} of $J$ substitutes, $Q=\big(1-\sum_{j=1}^{J}\epsilon_j\big)P+\sum_{j=1}^{J}\epsilon_j R_j$, with $\epsilon=(\epsilon_1,\dots,\epsilon_J)$ and $\epsilon_0=1-\sum_j\epsilon_j\ge0$. Both halves of \IRIS{} extend without new machinery.

\begin{proposition}[Multi-diluent audit]
\label{prop:multi}
Under $Q$ above with i.i.d.\ routing, at any threshold $\tau$ the tell probability is
\[
  p(\tau)=\epsilon_0\alpha_1(\tau)+\sum_j\epsilon_j q_j(\tau)\ge\epsilon_{\mathrm{tot}}\,q_{\min}(\tau),
\]
where $\epsilon_{\mathrm{tot}}=\sum_j\epsilon_j$ and $q_{\min}=\min_j q_j$. If the same threshold controls the honest tail, any-tell detection scales with the total foreign fraction as $\lceil\ln(1/\delta)/(\epsilon_{\mathrm{tot}}q_{\min})\rceil$. For quantification, if $u(y)\in\Delta^K$, $\bar g_M=\mathbb{E}_{M}[u(y)]$, and $G=[\bar g_{R_1}-\bar g_P,\dots,\bar g_{R_J}-\bar g_P]$, then $\mathbb{E}_Q[u]=\bar g_P+G\epsilon$, and the fractions are identifiable iff $G$ has full column rank.
\end{proposition}

\paragraph{Detection.}
At a per-response level $\tau$ the tell probability is $p=\epsilon_0\alpha_1+\sum_j\epsilon_j q_j$ with $q_j=\mathbb P_{R_j}(s>\tau)$, so $p\ge(\sum_j\epsilon_j)\min_j q_j$ once every diluent separates ($q_j>\alpha_1$). The any-tell budget of Theorem~\ref{thm:mix}(a) therefore holds in the \emph{total} foreign fraction $\epsilon_{\mathrm{tot}}=\sum_j\epsilon_j$: the $1/\epsilon$ law detects \emph{any} foreign traffic at a cost set by $\epsilon_{\mathrm{tot}}$ and the weakest separating diluent, and---as in the single-diluent case---needs no enrollment of the $R_j$.

\paragraph{Nonnegative unmixing.}
Let the classifier map a response to its posterior vector $u(y)=\hat{\mathbb P}(\cdot\mid y)\in\Delta^{K}$ over the $K$ enrolled endpoints, and let the \emph{signature} of endpoint $M$ be $\bar g_M=\mathbb{E}_{y\sim M}[u(y)]$ (estimated on enrollment). Under $Q$, $\mathbb{E}_Q[u]=\epsilon_0\,\bar g_P+\sum_j\epsilon_j\,\bar g_{R_j}$, so from the observed mean posterior $\bar v$ over $m$ audit queries the simplex-constrained least squares is
\begin{align*}
  \hat\eta
  &=\arg\min_{w\ge0,\,\mathbf{1}^{\top}w\le1}
  \Big\|\bar v-\Big((1-\mathbf{1}^{\top}w)\bar g_P\\
  &\hspace{3.9cm}+\sum_jw_j\bar g_{R_j}\Big)\Big\|_2^2 .
\end{align*}
The optimizer $\hat\eta$ is estimated from audit responses, not supplied as a ground-truth quantity. Under this multi-diluent model and the full-rank condition below, it estimates the true fraction vector $\epsilon$; in the single-diluent main audit, \IRIS{} uses it only to rank candidate identities and reports the tell-based $\hat\epsilon$. This is the standard quantification-learning / mixture-proportion problem. The fractions are \emph{identifiable} iff the shift vectors $\{\bar g_{R_j}-\bar g_P\}_{j=1}^{J}$ are linearly independent (the matrix $G=[\bar g_{R_1}-\bar g_P,\dots,\bar g_{R_J}-\bar g_P]$ has full column rank $J$)---the multi-diluent analogue of the single-diluent condition $q>\alpha_1$ (at $J{=}1$, $\bar g_{R_1}\neq\bar g_P$). Because posterior signatures live in the $K$-simplex, every column sums to zero and $J\le K-1$ is necessary. Near-twin diluents collapse columns of $G$ and are non-identifiable, exactly mirroring the $\epsilon^{-2}$ wall; a delta-method or bootstrap over the signatures propagates the interval. \emph{Empirical validation}: on the $17$-endpoint gateway ($c_{10,8}$, content-only, $m{=}200$), simultaneously diluting a reference with $J{=}2$ ($\epsilon{=}.2,.2$) or $J{=}3$ ($\epsilon{=}.15$ each) full-rank diluents and solving the simplex-constrained least squares above recovers \emph{every} per-diluent fraction to a median worst-case error of $0.03$--$0.04$ ($90$th percentile $0.07$; $98\%$ of configurations within $0.1$), with all $120$ sampled triples identifiable (none rank-deficient). The single-machinery claim thus holds: detection and per-diluent quantification extend to several simultaneous diluents at the same accuracy as the single-diluent case (App.~\ref{app:estimateq}).

\section{Backend Separability}
\label{app:backend}
Extending Section~\ref{sec:exp}: the same nominal model served by local Ollama vs.\ a commercial gateway is distinguishable from visible strings at $T{=}1$, on \emph{content-only} (length/format-invariant) features, so the signal is a sampling-distribution difference rather than a serving-format artifact. Table~\ref{tab:backend} gives per-probe verification AUROC (local vs.\ gateway) and the raw output diversity (distinct responses over $120$ local / $100$ gateway draws), which exposes the mechanism: local quantized, $\mathrm{top\text{-}}k$-truncated decoding is markedly more mode-collapsed.

\begin{table}[!ht]
\centering\footnotesize
\caption{Backend detectability of one nominal model (local Ollama vs.\ gateway), $T{=}1$. AUROC is per-response
($m{=}1$) on content-only features; all probes reach $\approx1.0$ by $m{=}8$. Probe labels abbreviate
$c_{10,1}$, $c_{100,1}$, $c_{2,16}$, and $c_{10,8}$;
diversity is distinct/total responses.}
\label{tab:backend}
\setlength{\tabcolsep}{2pt}
\begin{tabular}{@{}llcc@{}}
\toprule
model & probe & AUROC$_{m=1}$ & div.\ (loc/gw) \\
\midrule
Qwen3-8B  & digit   & $.76$ & $1/120$ vs $11/100$ \\
          & num100  & $.86$ & $3/120$ vs $24/100$ \\
          & bits16  & $.91$ & $42/120$ vs $88/100$ \\
          & digits8 & $.95$ & $34/120$ vs $97/100$ \\
Qwen3-32B & num100  & $.93$ & $16/120$ vs $17/100$ \\
          & bits16  & $.90$ & $62/120$ vs $40/100$ \\
          & digits8 & $.97$ & $82/120$ vs $29/100$ \\
\bottomrule
\end{tabular}
\end{table}

\begin{figure}[t]
  \centering
  \includegraphics[width=0.92\columnwidth]{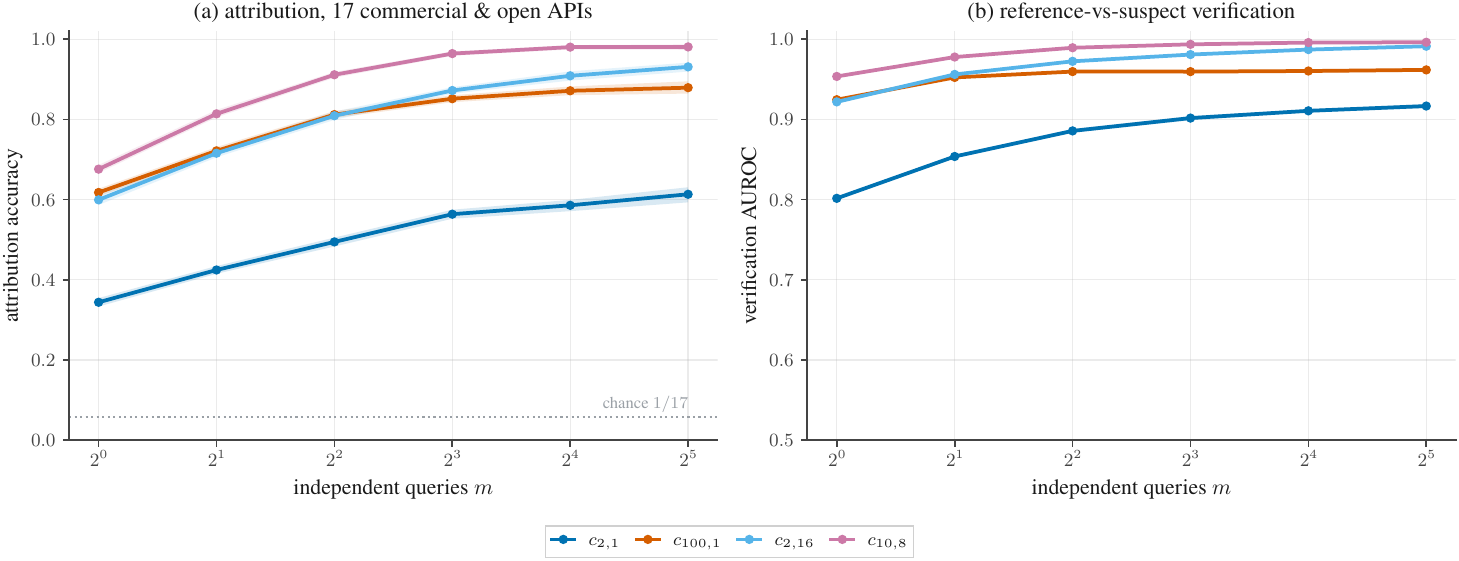}
  \caption{External validation on $17$ commercial and open APIs (chance $1/17$, dotted). (a)~attribution
  accuracy and (b)~reference-vs-suspect verification AUROC vs.\ independent queries $m$; shaded bands are
  $95\%$ bootstrap intervals over $40$ splits. The $e^{-I m}$ law and entropy gating replicate off the
  local testbed onto frontier closed APIs (Section~\ref{sec:exp}).}
  \label{fig:external}
\end{figure}

\section{Queries, Not Length}
\label{app:length}
The body claims that evidence accumulates in the number of \emph{independent} queries $m$, not within-response length $L$ (Prop.~\ref{prop:len}). Table~\ref{tab:length} is the direct empirical support: single-response attribution accuracy ($K{=}3$ Qwen3 sizes, chance $1/3$) as a function of the requested $L$, under two prompt styles (chat-templated and direct). Accuracy is \emph{non-monotone} in $L$---it wanders up and down (e.g.\ chat: $.94\!\to\!1.00\!\to\!.97\!\to\!.83\!\to\!.89$ over $L{=}16$--$256$)---and a linear fit of accuracy on $L$ has \emph{leave-one-condition-out cross-validated} $R^2<0$ under \emph{both} styles ($-1.59$ chat, $-4.73$ direct)---worse than predicting the mean (an in-sample OLS-with-intercept $R^2$ is non-negative by construction; the negative value is the held-out score, our point being predictive uselessness of $L$, not in-sample fit). These conditions vary the \emph{requested} length, so they are different \emph{contexts} rather than truncations of one fixed response (Prop.~\ref{prop:len}(i)); the non-monotonicity is consistent with but not predicted by the same-context monotonicity (full argument in the proof of Prop.~\ref{prop:len}, App.~\ref{app:proofs}). Length still carries evidence---a long response is not worthless---but it does not accumulate \emph{linearly} and is an unreliable budget axis, whereas the error exponent tensorizes \emph{exactly} across independent queries $m$ (the $e^{-I m}$ law, Fig.~\ref{fig:mscaling}a). This is why \IRIS{} spends its budget on more short independent queries rather than longer responses.

The same accounting fixes the token-optimal length: turning a length-$L$ query into a per-token yield question makes many short independent queries beat one long response whenever they give larger tell probability per token.

\begin{corollary}[Token yield]
\label{cor:len}
Fix a token budget $B=mL$ as $m$ independent queries of length $L$, where a length-$L$ query has tell probability $q(L)$. The miss probability is at most $e^{-\epsilon q(L)\,B/L}$, so this bound is minimized by maximizing $q(L)/L$.
\end{corollary}

\noindent The exact geometric miss gives the same optimizer to first order when the false-positive probability is controlled. Figure~\ref{fig:yield} contrasts the two yield shapes that fix where $L^\star$ lands, and the shape follows from how the tell is computed rather than from any extra assumption. A tell that can fire at each position gives a concave $q(L)$ whose chord slope $q(L)/L$ through the origin only falls, so the budget is best spent on the shortest admissible queries; a tell that reads a whole-response statistic needs a few symbols before it engages, giving an S-shaped $q(L)$ and an interior token-optimal length. In both regimes the linear-in-$m$ accumulation of the $e^{-Im}$ law (Prop.~\ref{prop:exp}) is the reliable lever, while $L$ is not.

\begin{figure}[t]
  \centering
  \includegraphics[width=\columnwidth]{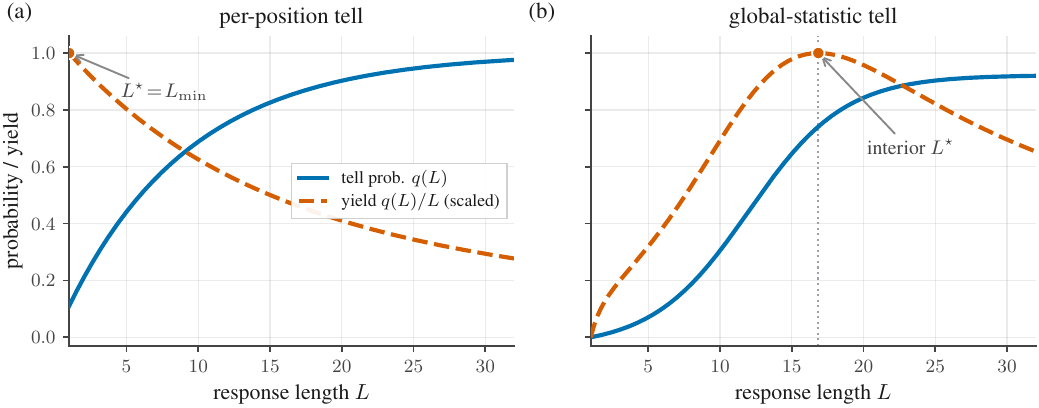}
  \caption{Token-optimal length maximizes the per-token tell yield $q(L)/L$, not $L$ (Cor.~\ref{cor:len}); the yield is rescaled to its own maximum so the panels share an axis. \textbf{(a)} A per-position tell makes $q(L)$ concave with $q(0){=}0$, so $q(L)/L$ is decreasing and the optimum is the shortest admissible length ($L^\star{=}L_{\min}$): many short independent queries dominate. \textbf{(b)} A global-statistic tell (compression, $n$-gram entropy, $\chi^2$) is information-starved at small $L$, so $q(L)$ is S-shaped and the yield peaks at an interior $L^\star$. Which regime holds is empirical (Table~\ref{tab:length}); either way the accumulation axis is the query count $m$, not $L$.}
  \label{fig:yield}
\end{figure}

\begin{table*}[t]
  \centering\footnotesize\setlength{\tabcolsep}{6pt}
  \caption{Single-response attribution accuracy vs.\ response length $L$ ($K{=}3$, chance $1/3$), two prompt
  styles. Accuracy is non-monotone in $L$ and a linear fit has $R^2<0$ (worse than the mean)---length is not
  an accumulation axis, unlike independent queries $m$ (Fig.~\ref{fig:mscaling}). ``---'' = length not in
  that style's grid.}
  \label{tab:length}
  \begin{tabular}{ccc}
    \toprule
    $L$ & chat & direct \\
    \midrule
    $16$  & $.94$ & $.90$ \\
    $24$  & ---   & $.92$ \\
    $32$  & $1.00$& $.83$ \\
    $48$  & ---   & $.93$ \\
    $64$  & $.97$ & $.83$ \\
    $96$  & ---   & $.89$ \\
    $128$ & $.83$ & $.95$ \\
    $192$ & ---   & $.88$ \\
    $256$ & $.89$ & ---   \\
    \midrule
    $R^2$ (acc vs.\ $L$) & $-1.59$ & $-4.73$ \\
    \bottomrule
  \end{tabular}
\end{table*}

\section{45-Model Scaling}
\label{app:scale}
We re-run the full pipeline on a pooled library of $45$ commercial/open models accessed via OpenRouter ($27{,}000$ responses). \emph{Attribution} (chance $1/45{=}2.2\%$) reaches accuracy $0.85$ at $m{=}8$ and $0.90$ at $m{=}16$ on $c_{10,8}$, with reference-vs-suspect verification AUROC $0.98$ (Fig.~\ref{fig:gw45}); per-model recall is broad (median $0.91$, and $0.82$ even excluding the five easiest models), so it is not carried by a few trivial outliers. Because $\sim$13 models emit a verbose preamble whose digits the character parser retains, length-ablation alone is not a complete control here; the \emph{decisive} de-confounding is a clean-format control restricting to the $K{=}32$ models with indistinguishable output length ($7$--$12$ chars), which still gives accuracy $0.83$ at $m{=}8$ (chance $1/32{=}3.1\%$) and AUROC $0.97$---confirming the signal is the per-token sampling distribution, not response length. \emph{Dilution} across six frontier models follows an approximate inverse-power wall; on \emph{content-only} features the log--log slope is ${\approx}{-}1.4$ and the median length-matched pair is detected at $\epsilon{=}0.1$ in ${\approx}64$ queries (the smaller $32$-query figure on the full feature set is inflated by one length-distinguishable model). Detectability is strongly pair-dependent: several frontier-vs-frontier dilutions remain undetectable at $\epsilon{=}5$--$10\%$ even at $256$ queries. \emph{Open-set} rejection (random $20\%$ of models held out as unknown, knownness $=\max_c\hat{\mathbb P}(c\mid y)$) degrades modestly with pool size on a like-for-like protocol (same data, same splits, content-only features): $c_{10,8}$ open-set AUROC@$m{=}8$ falls $0.73\!\to\!0.70$ from $K{=}17$ to $K{=}45$ (reject-at-$95\%$-accept $0.26\!\to\!0.11$), as an unknown grows likelier to overlap \emph{some} model in a denser enrolled set.

\begin{figure}[t]
  \centering
  \includegraphics[width=\columnwidth]{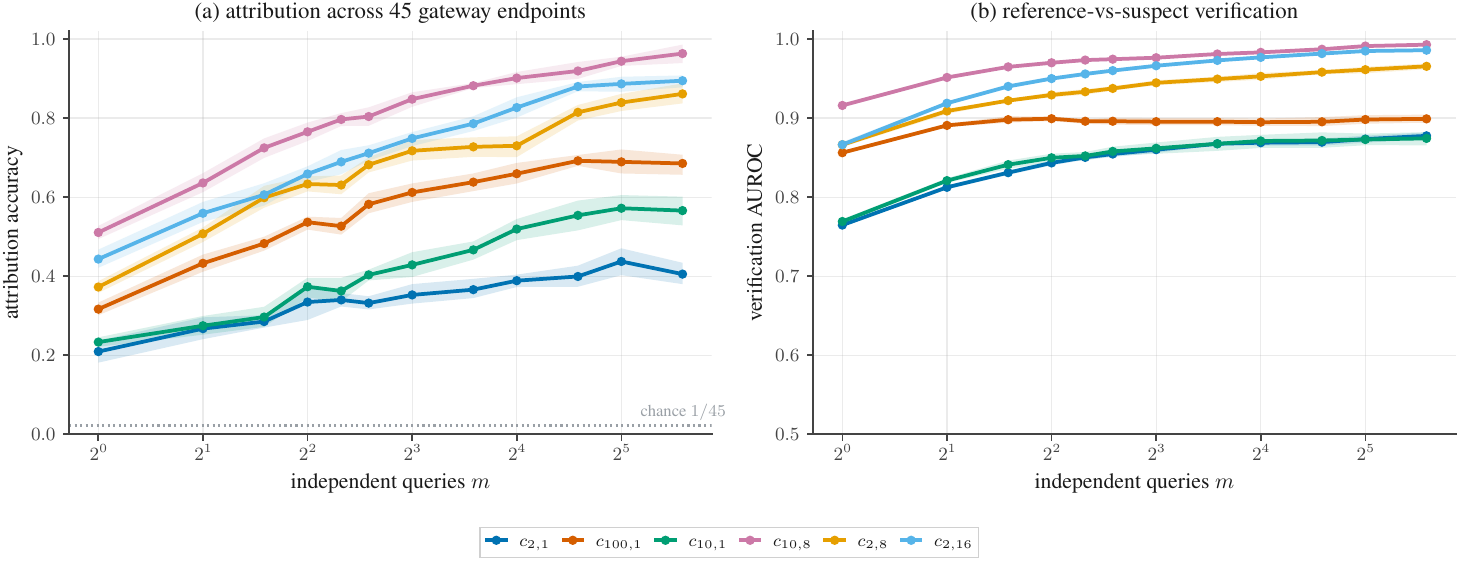}
  \caption{Scaling to $45$ models via OpenRouter (chance $1/45{=}2.2\%$, dotted). (a)~attribution accuracy and
  (b)~reference-vs-suspect verification AUROC vs.\ independent queries $m$, per probe; shaded bands are $95\%$
  bootstrap intervals over $10$ splits. The $e^{-I m}$ law and entropy gating hold at $K{=}45$:
  $c_{10,8}$ reaches accuracy $0.85$ at $m{=}8$ ($0.96$ at $m{=}48$) and verification AUROC
  $0.98$, while $c_{2,1}$ stays weak.}
  \label{fig:gw45}
\end{figure}

\begin{table}[t]
  \centering\footnotesize\setlength{\tabcolsep}{4pt}
  \caption{Length/format ablation at $K{=}45$ ($179\!\to\!144$ features): attribution accuracy and
  verification AUROC at $m{=}8$, full vs.\ length-ablated. The high-entropy sequence probes are
  near-invariant ($\le0.03$ drop), so their signal is the per-token sampling distribution, not response
  length; the lower-entropy $c_{100,1}$ relies more on length.}
  \label{tab:gw45abl}
  \begin{tabular}{lcccc}
    \toprule
    & \multicolumn{2}{c}{accuracy @ $m{=}8$} & \multicolumn{2}{c}{AUROC @ $m{=}8$} \\
    \cmidrule(lr){2-3}\cmidrule(lr){4-5}
    probe & full & abl.\ & full & abl.\ \\
    \midrule
    $c_{10,8}$    & $.855$ & $.844$ & $.976$ & $.973$ \\
    $c_{2,16}$     & $.725$ & $.691$ & $.967$ & $.955$ \\
    $c_{100,1}$ & $.611$ & $.502$ & $.893$ & $.871$ \\
    $c_{10,1}$       & $.446$ & $.356$ & $.860$ & $.831$ \\
    \bottomrule
  \end{tabular}
\end{table}

\section{Base Specificity}
\label{app:knobs}
The economic threat is covert substitution of the \emph{base model} (Section~\ref{sec:prelim}); the audit must therefore be \emph{specific}---flagging a model substitution but not an honest same-model sampler/effort retune. Table~\ref{tab:specificity} runs the dilution audit on a substitution ($\epsilon{=}1$, $m{=}80$, content-only) for each suspect type. An adjacent-temperature retune of the same model yields $\hat\epsilon{=}0.04$ (flagged $8\%$, barely above the false-positive floor), while a base-model substitution yields $\hat\epsilon{=}0.53$---a $13\times$ separation; only the extreme, observable $T{=}0$ greedy collapse registers among samplers. Reasoning effort leaves a residual ($\hat\epsilon{=}0.14$) but is economically transparent and directly observable from returned token counts (and removable by length-matching), so it is not a covert dilution vector.

\begin{table}[t]
  \centering\footnotesize\setlength{\tabcolsep}{4pt}
  \caption{Specificity of the dilution audit: substitution ($\epsilon{=}1$, $m{=}80$, content-only) by
  suspect type. $\hat\epsilon$ is the reported routing fraction (deployed lower bound), ``flag'' the $p{<}0.01$
  detection rate. A specific audit flags a base-model substitution, not an honest same-model retune. The cross-model
  row is the $3$-model temperature pool (moderate margin); the $17$-endpoint matrix (App.~\ref{app:dilscale})
  gives the full base-model picture.}
  \label{tab:specificity}
  \begin{tabular}{@{}lcc@{}}
    \toprule
    suspect (whole stream) & $\hat\epsilon$ & flag rate \\
    \midrule
    same config (false-positive floor)        & $0.01$ & $0.00$ \\
    same model, adjacent-$T$ retune            & $0.04$ & $0.08$ \\
    same model, far $T$ ($\to1.5/2.0$)         & $0.11$ & $0.30$ \\
    same model, $T{=}0$ greedy collapse        & $0.30$ & $0.60$ \\
    same model, reasoning effort (min$\leftrightarrow$high) & $0.14$ & $0.61$ \\
    \textbf{different base model} (substitution) & $\mathbf{0.53}$ & $\mathbf{0.78}$ \\
    \bottomrule
  \end{tabular}
\end{table}

We next confirm which knobs leave \emph{any} visible-string signature, varying one knob on a fixed model on content-only (length-invariant) features (Fig.~\ref{fig:knobs}); in every case content-only AUROC tracks full-feature AUROC, so none of these effects is a length artifact. Quantization---a cost-relevant precision downgrade, the one borderline same-family cheat---is probed on the real gateway by pinning OpenRouter \emph{providers} that serve one model at different declared precision.

\begin{figure}[t]
\centering
\includegraphics[width=\columnwidth]{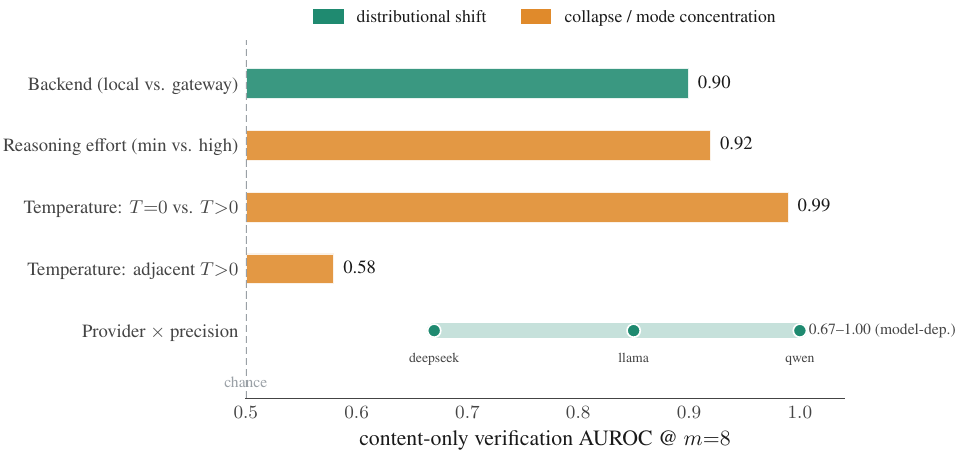}
\caption{Which serving knobs leave a visible-string signature: content-only (length-invariant) verification
AUROC at $m{=}8$ for same-model variants, colored by mechanism. The audit is specific to the base model---adjacent
operating temperatures fall to chance (an honest retune is not flagged), provider/precision is model-dependent,
and the only strong sampler signal is the extreme $T{=}0$ collapse.}
\label{fig:knobs}
\end{figure}

\paragraph{Provider\,$\times$\,precision.}
Two production endpoints serving the nominally identical model at different declared precision on \emph{different providers} are distinguishable from visible strings (content-only AUROC@$m{=}8$: \texttt{qwen3-235b} bf16/fp8 $1.00$, \texttt{llama-3.3-70b} $0.85$, \texttt{deepseek-v3.2} fp4/fp8 $0.67$). Because provider and precision are perfectly confounded in this design (no same-provider precision contrast), this evidences \emph{endpoint-configuration} identifiability, not weight precision in isolation. The strong end is partly an output-diversity split (the \texttt{qwen3-235b} fp8 endpoint collapses to $1$--$2$ distinct strings per $100$, the bf16 endpoint stays diverse), and the \texttt{deepseek} fp4-vs-fp8 pair---both fully diverse---sits near chance, so identifiability is model-dependent and degenerates when both endpoints sample broadly.

\paragraph{Temperature and effort.}
Sampler choice perturbs the audit only at the extremes; a temperature retune over the normal operating range is essentially invisible. The one strong sampler tell is \emph{low-temperature collapse}: $T{=}0$ vs.\ any $T{>}0$ separates at AUROC ${\approx}0.99$, because greedy decoding repeats a single string (\texttt{llama-4} at $T{=}0$ emits one distinct $16$-bit string over $80$ draws). Between \emph{adjacent} operating temperatures the signal nearly vanishes: AUROC falls to ${\approx}0.58$ and \texttt{T1.5}-vs-\texttt{T2.0} is at chance, because the across-sample diversity that carries the signal has already saturated by $T{\approx}1$. These two regimes are exactly the second-order analysis of App.~\ref{app:tempmag}. In the interior, adjacent temperatures separate only at $O((\Delta\beta)^2)$---the interior Fisher quadratic, which gives the observed $\approx0.58$---while $T{=}0$ vs.\ $T{>}0$ is a separate support-collapse tell. So a temperature retune is, correctly, near-invisible to a \emph{model}-dilution audit, while a backend substitution is not. Reasoning effort (minimal vs.\ high) likewise leaves a fingerprint: pooled-$m{=}8$ content-only AUROC is $0.92$, though a single response gives only $0.72$--$0.76$. The mechanism is \emph{mode/prefix concentration}, not a flatter distribution. Minimal effort anchors on preferred openings---on the digit probe, \texttt{gpt-oss-120b}'s prefix \texttt{274} appears $41\%$ vs.\ $8\%$ of the time, and first-digit entropy drops $2.76\!\to\!1.88$. The overall character entropy is unchanged, so we do \emph{not} claim an entropy law: only the mode mass moves. The take-away is that the audit's signal is dominated by the \emph{base model}: the sampler contributes only at the extremes (greedy collapse), and reasoning effort---though it leaves a residual mode-concentration signature---is observable from token counts and billed transparently, so the alarm is correctly reserved for a change of the underlying model (Table~\ref{tab:specificity}).

\subsection{Precision Retunes}
\label{app:quant}
Quantization is the cost-relevant \emph{same-weights} cheat: serving a $4$-bit model where the reference is $16$-bit. We audit it directly under controlled mixtures (reference \texttt{qwen3:0.6b-fp16}, substitute the default \texttt{q4}), measuring the per-query tell rate $q$ and the resulting budget $m^\star=\ln(1/\delta)/(\epsilon q)$ against the false-positive floor $\alpha{\approx}0.05$ (Table~\ref{tab:quant}). Three findings, each guarded by a control. First, the cheat \emph{is} detectable, but only on a high-entropy probe: on $c_{2,16}$, $q{=}0.42$ clears the floor ($m^\star{\approx}72$ at $\epsilon{=}0.1$), whereas on low-entropy $c_{100,1}$ it does not ($q{=}0.06$). Second, the signal is genuine per-token randomness, not response length: under a length/format ablation ($179{\to}144$ features) the $c_{2,16}$ tell barely moves ($q{:}0.42{\to}0.37$, $m^\star{\approx}82$), while the apparent $c_{100,1}$ signal collapses to $q{=}0$---it was a length artifact. Third, a fp16-vs-fp16 negative control (two independent fp16 runs) stays below the floor on both probes ($q\le0.03$), so the audit does not manufacture a tell between identical configurations. We therefore do \emph{not} claim a clean entropy law for quantization: detectability is probe- and model-dependent, but on the right probe a precision downgrade is a real, length-robust, control-validated tell.

\begin{table}[t]
  \centering\footnotesize\setlength{\tabcolsep}{3.5pt}
  \caption{Auditing a precision cheat (\texttt{q4} vs.\ \texttt{fp16}, ref \texttt{qwen3:0.6b}). Per-query
  tell rate $q$ (mean $[95\%$ CI$]$) and budget $m^\star$ at $\epsilon{=}0.1$; ``detect'' marks $q$ clearing the
  $\alpha{\approx}0.05$ false-positive floor. The last column is a fp16-vs-fp16 negative control. Detectable
  only on the high-entropy probe, and the tell survives length ablation; the low-entropy signal is a length
  artifact that ablation removes.}
  \label{tab:quant}
  \begin{tabular}{llccc}
    \toprule
    probe & features & $q$ (q4 vs.\ fp16) & det.\ & $m^\star_{\epsilon=.1}$ \\
    \midrule
    $c_{2,16}$     & full ($52$)      & $.42\,[.27,.52]$ & yes & $72$ \\
    $c_{2,16}$     & len-abl.\ ($28$) & $.37\,[.17,.51]$ & yes & $82$ \\
    $c_{100,1}$ & full ($165$)     & $.06\,[.02,.11]$ & no  & --- \\
    $c_{100,1}$ & len-abl.\ ($132$)& $.00\,[.00,.00]$ & no  & --- \\
    \midrule
    \multicolumn{5}{l}{\emph{fp16-vs-fp16 negative control} (length-ablated):}\\
    $c_{2,16}$     & abl.\ ($28$)     & $.03\,[.01,.07]$ & no  & --- \\
    $c_{100,1}$ & abl.\ ($132$)    & $.00\,[.00,.01]$ & no  & --- \\
    \bottomrule
  \end{tabular}
\end{table}

\paragraph{Across model sizes.}
The $0.6$B audit above leaves open whether the precision cheat is detectable at scale. We collected a fresh q4-vs-fp16 pair for \texttt{qwen3} at $0.6$B, $1.7$B and $4$B ($100$ repeats each) and find detectability \emph{grows} with size, on length-ablated (content-only) features: on the high-entropy $c_{2,16}$ the per-query tell rate is $q{=}0.35$ at $0.6$B, $0.74$ at $1.7$B (AUROC $0.95$), and $1.00$ at $4$B (AUROC $1.00$), so the budget $m^\star{=}\ln(1/\delta)/(\epsilon q)$ at $\epsilon{=}0.1$ falls from ${\approx}86$ to ${\approx}30$. At $4$B the cheat is caught on \emph{every} probe; at smaller sizes it is probe-dependent (the low-entropy $c_{100,1}$ signal is a length artifact that ablation removes, $q{\to}0$). Larger models have sharper, more idiosyncratic sampling laws that quantization perturbs more visibly, so the same-weights precision downgrade---the realistic covert cheat---is auditable and increasingly so with scale, on the high-entropy sequence probe \IRIS{} already leads with.

\section{Feature Sets}
\label{app:featureset}
\IRIS{} uses two feature groups (Section~\ref{sec:prelim}): the full $179$-feature set and the length/format-invariant \emph{content-only} subset. Table~\ref{tab:featureset} states which set produces each headline result. The rule is principled: any claim that a \emph{substitution} occurred---one model/endpoint/serving-config standing in for another---is made on content-only features, so it cannot be an artifact of length or formatting; the information-rate, budgeting, and probe-ranking results use the full set and are separately shown robust to a length/format ablation ($179{\to}144$).

\begin{table}[t]
  \centering\footnotesize\setlength{\tabcolsep}{3pt}
  \caption{Feature set per headline result. ``Content-only'' = length/format-invariant subset; ``full'' = all $179$
  features. Every \emph{substitution}/endpoint/config claim is content-only.}
  \label{tab:featureset}
  \begin{tabular}{p{0.50\columnwidth}p{0.42\columnwidth}}
    \toprule
    Result & Feature set \\
    \midrule
    $e^{-I m}$ law, $m$-not-$L$ (\S\ref{sec:exp}) & full; replicated content-only ($144$); $R^2{<}0$ in $L$ \\
    $I$ entropy/temperature gating & full \\
    Budget-from-$\widehat{I}_{\mathrm{auc}}$ (Tab.~\ref{tab:safety},~\ref{tab:lcb}) & full (calibration AUROC) \\
    FPR calibration (Tab.~\ref{tab:fpr}) & full \\
    Probe choice by $\widehat{I}_{\mathrm{auc}}$ & full \\
    Dilution $1/\epsilon$ (Fig.~\ref{fig:mixrouting}) & full; $45$-pool slope content-only \\
    Baselines (Tab.~\ref{tab:capability}) & full (clean attribution) \\
    External $17$ APIs (Fig.~\ref{fig:external}) & full; length-ablated check \\
    Endpoint-vs-model / backend (Tab.~\ref{tab:backend}) & \textbf{content-only} \\
    Dilution audit at scale + live (\S\ref{sec:extra}, App.~\ref{app:dilscale}) & \textbf{content-only} (detection, tell-rate, $\hat\epsilon$) \\
    $45$-pool open-set (App.~\ref{app:scale}) & \textbf{content-only}; attribution de-confounded by length-matched $K{=}32$ ($.83@8$) \\
    Serving-knob identifiability (Fig.~\ref{fig:knobs}) & \textbf{content-only} \\
    \bottomrule
  \end{tabular}
\end{table}

\section{Paraphrase Robustness}
\label{app:paraphrase}
The probe family is conspicuous, so an adaptive router could keyword-detect ``generate random bits'' and serve it honestly while diluting other traffic. We test how far paraphrasing defeats this by collecting responses to four \emph{stealth paraphrases} of $c_{2,16}$ that preserve the $16$-i.i.d.-bit semantics but reword the surface (A/B-bucket assignment, train/validation tagging, a protocol nonce, fair coin tosses; $80$ repeats $\times$ the $K{=}6$ ladder). Three findings. (i)~\emph{The probe is detectable}: a TF--IDF\,$+$\,logistic classifier trained on obvious random-generation prompts flags $13/13$ of them and still flags $3/4$ stealth paraphrases (only the nonce framing slips through), so paraphrasing only \emph{partially} evades a keyword router. (ii)~\emph{The fingerprint is in the sampling distribution, not the wording}: models stay near-perfectly identifiable under every paraphrase (attribution accuracy $0.99$--$1.00$ at $m{=}8$, chance $1/6$). (iii)~\emph{It partially transfers across phrasings}: a classifier trained on the \emph{plain} probe attributes the reworded responses at $0.50$--$0.75$ accuracy at $m{=}8$ (chance $0.17$), so each framing shifts the served distribution somewhat and is best enrolled separately. Together, an auditor can rotate among many semantically-equivalent framings---each accumulating evidence---to raise a static router's cost, but cannot fully hide the request from a determined classifier; a fully adaptive adversary that recognizes the randomness \emph{semantics} (not the wording) remains the key open threat.

\paragraph{Output-side special-casing.} The keyword router above acts on the \emph{prompt}; a gateway could instead special-case the probe by its \emph{output} shape (``the response is a bare $0/1$ string''). We test both evasion surfaces on $5$ \emph{distinct commercial} endpoints (\texttt{gpt-5.4-mini}, \texttt{claude-haiku-4.5}, \texttt{gemini-3.5-flash}, \texttt{llama-4-maverick}, \texttt{qwen3-max}; $50$ repeats each), comparing the bare probe, three keyword-free disguises (unit-test, A/B-bucket, protocol-nonce framings), and an \emph{embedded-output} disguise that asks for the $16$ bits inside a templated device-log line. Three findings. (i)~\emph{The fingerprint survives every disguise on commercial models}: within-disguise $K{=}5$ attribution is $0.87$--$0.94$ at $m{=}1$ and \emph{$1.00$ at $m{=}8$} (chance $0.20$) for \emph{all} variants, including the embedded one. (ii)~\emph{No single static filter catches everything}: an output-shape classifier trained to flag bare bit-strings (vs.\ genuine traffic) flags the bare probe and the clean disguises ($0.78$--$1.00$) but is \emph{completely defeated} by the embedded-output disguise (flag rate $0.00$, identical to honest traffic), which nonetheless still carries the fingerprint ($m{=}8$ attribution $1.00$); symmetrically, the prompt keyword router misses the nonce framing. So the two cheap filters are each defeated by a \emph{different} disguise while detection is unaffected---to special-case the probe a gateway must run \emph{both} a semantic prompt parser and an output parser, and even then an auditor composing an unseen framing with an embedded output slips past. (iii)~\emph{Honest limits}: cross-phrasing transfer (train on the bare probe, attribute a disguise) is only \emph{partial} ($0.40$--$0.74$ at $m{=}8$), so the auditor should enroll under the deployed phrasing rather than rely on transfer; and some framings depress a model's output diversity (the unit-test framing collapses distinct-string rate to $0.16$ for some endpoints, versus $0.37$ bare), so a disguise's entropy must be validated before use. This corroborates on real gateways that static probe-detection---prompt- or output-side---is an arms race the auditor can keep winning by rotation, while a fully \emph{adaptive} semantic router stays the key open threat.

\paragraph{Task-embedded probe prefix.} The experiments above show that \IRIS{} remains effective under several disguises, but after publication we cannot exclude a gateway that recognizes standalone audit prompts and deliberately routes only those requests to the advertised model. A practical hardening is to embed the probe into the real task request: prepend or append an instruction requiring the response to \emph{begin} with a short random string (for example, several digits or bits) before answering the substantive prompt. The first few visible tokens are then concentrated on a finite known alphabet, so the same visible-string features and thresholds can be calibrated on that prefix while the downstream content remains the user's task. The auditor sends the identical task-embedded template to a trusted endpoint and to the suspect endpoint, using the trusted responses as the matched control for the prefix distribution. This does not rule out a fully semantic adaptive router, but it turns the failure mode ``route obvious probes honestly'' into an auditable deployment variant: a gateway must now both recognize the embedded probe inside ordinary traffic and make the task-conditioned prefix match the trusted model.

\section{Sticky Routing}
\label{app:noniid}
The dilution theory and the deployed binomial test assume each request is routed independently; real gateways may route stickily (sessions, bursts, time-of-day, batches). We simulate sticky routing on a gateway pair passing the margin check (\texttt{gpt-4o-mini} vs.\ \texttt{llama-3.1-8b}, $m{=}80$, $\epsilon{=}0.3$) with autocorrelation $\rho$ at a fixed marginal fraction, giving effective block length $L{\approx}1/(1{-}\rho)$. Three findings, at $L{=}1/4/16$: (i)~\emph{type-I is robust by construction}---under an honest stream the tells are i.i.d.\ $\mathrm{Bernoulli}(\alpha_1)$ \emph{regardless} of $\rho$ (no substituted block exists), so sticky routing cannot inflate the FPR; this is definitional, not a simulation discovery (the simulated $0.033/0.050/0.050$ are the i.i.d.\ null). (ii)~\emph{power degrades} with burstiness ($1.00\!\to\!0.97\!\to\!0.77$) as the effective independent count $m/L$ falls. (iii)~\emph{the i.i.d.\ interval breaks}: the binomial $\hat\epsilon$ CI coverage collapses ($0.92\!\to\!0.64\!\to\!0.35$) because it ignores within-block correlation, while a \emph{block bootstrap} restores it substantially ($0.93\!\to\!0.77\!\to\!0.57$, limited at $L{=}16$ by only ${\approx}5$ blocks in $m{=}80$). The point estimate $\hat\epsilon$ stays unbiased ($0.30/0.29/0.31$). So sticky routing does not break type-I but invalidates the i.i.d.\ confidence interval, which a block-aware interval repairs given enough blocks; a per-session or longer audit horizon supplies them. Content-adaptive routing (serving the probe honestly) is the separate, harder threat of App.~\ref{app:paraphrase}.

\section{Black-Box Auditors}
\label{app:auditors}
Table~\ref{tab:capability} positions \IRIS{} against the auditors of Section~\ref{sec:related}. \IRIS{}, FLIPS, KBF, B3IT and GateScope are all black-box API auditors, so the distinguishing axis is \emph{not} white-box model access. It is twofold: \IRIS{} (i)~\emph{estimates} its query budget from an on-data per-query rate rather than fixing it, and (ii)~quantifies dilution on a \emph{universal} probe---returning a calibrated $\hat\epsilon$ and naming the diluent by attribution---whereas the closest black-box auditors fix their query count and first build a target-specific probe: KBF estimates a routing fraction but only for a known (or candidate-pool) substitute under fixed routing and on a per-model knowledge-boundary probe, B3IT flags a temporal \emph{change} from per-target border inputs, and GateScope \emph{measures} gateways at the fleet level. \IRIS{} and FLIPS share the core idea of using random generation as a model-discrimination signal~\citep{flips2026,hopkins2023random,baddice2026}; our contribution is what we build on top of it. At the probe level, \IRIS{} does not fix one pseudo-random probe scored by randomness-test features but \emph{selects} a probe from a parameterized random-generation family $\{c_{n,L}\}$ by its on-data information rate $\widehat{I}_{\mathrm{auc}}$. At the algorithm level, \IRIS{} replaces fixed-budget instance classification with an estimate-then-budget loop, a dilution tell-count test, and a calibrated routing-fraction estimate $\hat\epsilon$---turning a deployment-config identifier into a budgeted cross-model dilution auditor. (We similarly build on, rather than reinvent, B3IT's exponential error law and temperature phase transition, Section~\ref{sec:theory}.) The verification-AUROC accumulation curves underlying the comparison are in Fig.~\ref{fig:baselines}.
\paragraph{Matched-budget comparison.}
The quantitative head-to-head in the body (Table~\ref{tab:baseline-main}) includes only auditors that run on the shared random-generation probe (FLIPS, MET, RUT). KBF and B3IT instead use target-specific probe families (knowledge-boundary factual recall, engineered border inputs) whose construction we cannot faithfully reconstruct from their papers, so we treat them qualitatively here and reproduce what \emph{is} reproducible of each: B3IT's fixed-budget \emph{strategy} and KBF's detection \emph{regime}. A faithful head-to-head isolates the \emph{budgeting strategy} by holding the probe and classifier fixed and varying only how the query count is set. Reimplementing a B3IT-style \emph{fixed}-budget decision rule on our $c_{100,1}$ probe at target AUROC $0.99$, a fixed $m{=}3$ meets the target on only $70\%$ of $30$ model pairs, whereas \IRIS{}'s estimate-then-budget loop meets it on $87\%$. This particular contrast is \emph{not} cost-matched---\IRIS{} spends a larger mean budget here ($11$ queries, $\approx3$ on easy pairs)---so it mixes adaptivity with budget; App.~\ref{app:baselines} isolates the two at matched cost (the canonical $87\%$ vs.\ $73\%$). Two head-to-heads we deliberately do \emph{not} run: (i) a faithful reimplementation of KBF's knowledge-boundary factual-recall probe is a \emph{different probe family} whose construction we cannot reproduce from the paper, so a probe-matched KBF comparison is out of scope; and (ii) on \emph{quantification}, B3IT and GateScope report no routing fraction while KBF estimates one only for a known/candidate substitute under fixed routing and on its own knowledge-boundary probe, so a probe-matched $\hat\epsilon$ comparison again reduces to reproducing KBF's probe family (above); the budgeting axis, not access, is the point. On the shared \emph{detection} task \IRIS{} matches KBF's regime (KBF detects $5$--$10\%$ routing for well-separated pairs; \IRIS{} detects $5\%$ in ${\approx}64$ queries, Section~\ref{sec:exp}).

\begin{table*}[t]
  \centering\footnotesize
  \begin{tabular*}{0.92\textwidth}{@{\extracolsep{\fill}}lllll@{}}
    \toprule
    Method & Access & Probe & Decision & Budget \\
    \midrule
    \IRIS{} & text & rand-gen & \textbf{attr+verify+}$\hat\epsilon$ & \textbf{est.\ }$m^\star$ \\
    FLIPS & text & rand-gen & instance ID & fixed \\
    MET & ref.\ samples & natural & equality & fixed \\
    RUT & log-ranks & natural & equality+mix (det.) & fixed \\
    KBF & text & knowledge & routing$+\hat\pi^\dagger$ & fixed \\
    B3IT & text & border & change-det. & fixed \\
    GateScope & text & content/bill & fleet measure & --- \\
    \bottomrule
  \end{tabular*}
  \caption{Where \IRIS{} sits among black-box auditors. The output-only auditors (\IRIS{}, FLIPS, KBF, B3IT,
  GateScope) share access; \IRIS{} alone \emph{estimates} the query budget (all others fix it) and quantifies dilution ($\hat\epsilon$ plus diluent ID) on a \emph{universal} probe. ``(det.)'' marks detect-only mixture
  handling; $\dagger$ marks a routing-fraction estimate valid only for a known or candidate-pool substitute under fixed routing.}
  \label{tab:capability}
\end{table*}

\begin{figure}[t]
  \centering
  \includegraphics[width=\columnwidth]{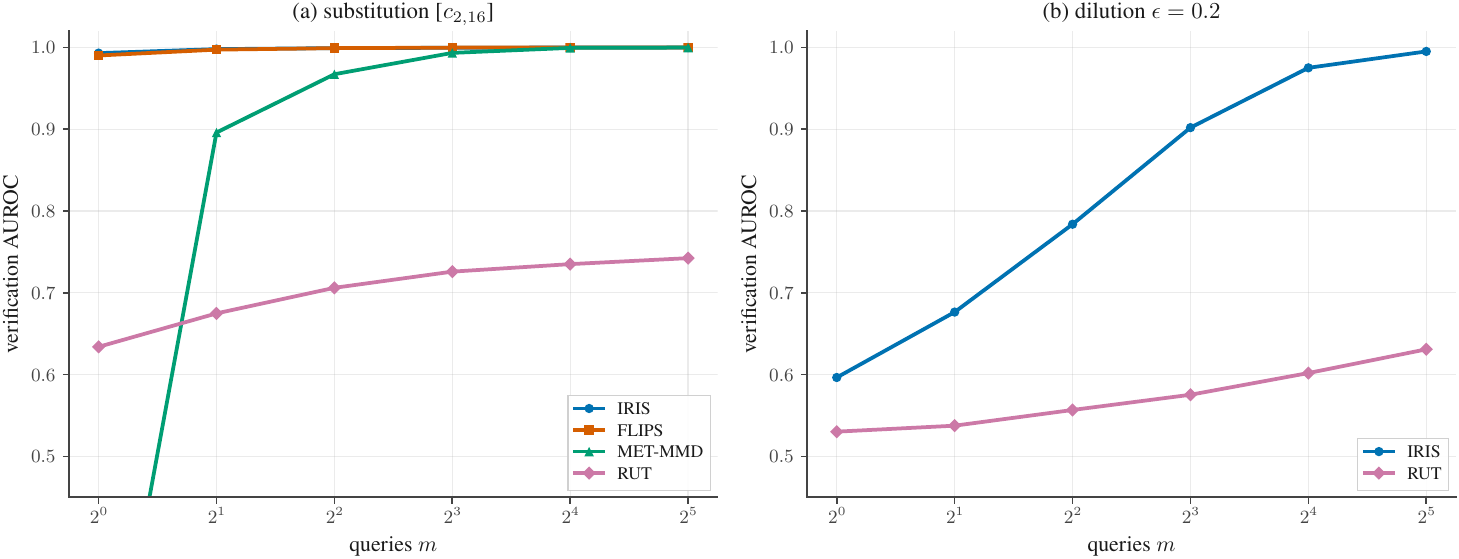}
  \caption{Method comparison on the shared random-generation probe ($c_{2,16}$), verification
  AUROC vs.\ queries $m$. (a)~Substitution: \IRIS{} ties a FLIPS-style classifier ($\mathrm{AUROC}\,{\approx}\,0.99$
  at $m{=}1$, both saturating to $1.0$ by $m{=}8$): the shared random probe ties, and \IRIS{}'s contribution is the budgeting and dilution capability it adds, while both dominate
  MET (maximum-mean-discrepancy two-sample test, degenerate at $m{=}1$, $0.11$) and an adapted RUT (rank-uniformity, plateauing near $0.74$).
  (b)~Dilution ($\epsilon{=}0.2$): \IRIS{} accumulates toward $1.0$ while RUT stays near chance---the capability
  gap is the budgeting axis, not access (Table~\ref{tab:capability}).}
  \label{fig:baselines}
\end{figure}

\paragraph{Cost-matched comparison.}
\label{app:baselines}
The headline ``$87\%$ vs.\ $70\%$ for a B3IT-style fixed $m{=}3$'' is \emph{not} cost-matched: \IRIS{} spent a mean budget ${\approx}11$ there, so part of the gap is simply $3.7\times$ more queries. We decompose it by comparing, on the per-pair verification-AUROC$(m)$ curves, three allocation rules at the \emph{same} average budget $B$: \emph{fixed} ($m{=}B$ for every pair), \emph{oracle} (per-pair minimal $m$ to the target), and \emph{pilot-driven \IRIS{}} (rank pairs by a cheap pilot AUROC and waterfill the same pool onto the predicted-harder pairs, no target peeking). On the hardest probe ($c_{10,1}$, target AUROC $0.99$, $30$ pairs), at the matched budget $m{\approx}9$ the fixed-$m$ rule meets the target on $0.73$ of pairs versus the pilot-driven $\IRIS{}$ allocation's $0.87$ (oracle $0.90$); the fixed rule at $m{=}3$ meets only $0.60$. So of \IRIS{}'s advantage over fixed-$m{=}3$, roughly half is the larger budget its estimator prescribes ($0.60{\to}0.73$) and roughly half is the adaptive \emph{allocation} ($0.73{\to}0.87$ at matched cost). The same holds on $c_{100,1}$ ($0.77$ fixed vs.\ $0.80$ pilot at $m{\approx}10$, oracle $0.87$; the per-pair estimate-then-budget loop of App.~\ref{app:auditors} reaches $0.87$ on this probe by sizing each pair rather than waterfilling a fixed pool); on saturated easy probes ($c_{10,8}$) all rules tie near $1.0$. We caution that this compares \IRIS{}'s classifier and allocation against \emph{adapted} fixed-budget rules on our probe, not against the original KBF/B3IT/RUT/MET implementations (KBF's knowledge-boundary probe we cannot reproduce), so it is a \emph{capability}, not a head-to-head method, comparison.

\paragraph{Significance of the head-to-head.}
\label{app:sig}
We judge every performance gap in Table~\ref{tab:baseline-main} with a two-sided Wilcoxon signed-rank test on the paired per-split measurements behind the same cells ($40$ stratified splits, identical probe/seed/protocol to the table). \IRIS{}'s margin over MET and RUT is significant on every metric they support: substitution AUROC at $m{=}1$ ($p<4\times10^{-41}$ versus each), at $m{=}8$ ($p{=}2\times10^{-31}$ versus MET, $p<4\times10^{-41}$ versus RUT), and dilution AUROC at $m{=}8$ and $m{=}32$ ($p<4\times10^{-41}$ versus RUT). Against FLIPS the result is deliberately \emph{not} a significant gap: on clean substitution the two stay within $0.003$ AUROC and the sign of the difference flips with the budget (\IRIS{} higher at $m{=}1$, $p{=}2\times10^{-9}$; FLIPS higher at $m{=}8$, $p{=}0.05$), and closed-set attribution is a flat tie ($p{=}0.14$ at $m{=}1$, $p{=}1.0$ at $m{=}8$). \IRIS{}'s separation from FLIPS is thus structural rather than metric---only \IRIS{} scores dilution and estimates $\hat\epsilon$ (the ``--'' cells)---which no signed-rank test on a shared metric can express. The cost-matched allocation gap above ($0.73{\to}0.87$) is reported as a \emph{decomposition} of where the advantage originates.

\section{Extended Method Comparison}
\label{app:extended}
Section~\ref{sec:related} and App.~\ref{app:auditors} place \IRIS{} among the black-box \emph{substitution/dilution} auditors that share its random-generation probe (FLIPS, MET, RUT) or its access model (KBF, B3IT, GateScope). We complete the picture along three axes a reader may raise, model \emph{attribution}, \emph{concurrent} single-token fingerprinting, and \emph{sequential} change tests, and ground the dilution rate in classical mixture-testing theory. We first correct one framing from App.~\ref{app:auditors}: KBF and B3IT both ship public code, so the barrier to a numeric head-to-head is not reproducibility but \emph{probe family}. A faithful run puts each on its native probe (knowledge-boundary recall, engineered border inputs) at matched total budget, not on the shared $\mathcal{C}$. Table~\ref{tab:fullcompare} lays out where every comparable method sits across the audit's axes, making concrete that no prior method jointly covers substitution, fractional dilution, open-set attribution, routing-fraction estimation, and a pre-committed budget from a single text-only universal probe.

\begin{table*}[t]
  \centering\footnotesize\setlength{\tabcolsep}{3pt}
  \caption{Complete method comparison across the audit's axes. \textbf{Access}: what the method reads (text${=}$visible strings only; ref${=}$reference samples from trusted weights; ranks$^{w}${=}log-ranks needing a white-box local reference; lp${=}$log-probabilities; TEE${=}$trusted hardware). \textbf{Signal}: probe family. \textbf{Sub}/\textbf{Dil}/\textbf{Attr}: detects whole-stream substitution / fractional dilution / attributes the served model (chg${=}$change-only vs.\ the endpoint's own past; det${=}$detect-only mixture handling; p${=}$partial). $\hat\epsilon$: estimates the routing fraction ($\dagger$ only for a known/candidate substitute under fixed routing). \textbf{Budget}: est${=}$estimated before the audit, fix${=}$fixed by design, any${=}$anytime-valid. \textbf{API}/\textbf{Data}: evaluated on real commercial endpoints / releases a public dataset or code (p${=}$code or partial). Venues are in Section~\ref{sec:related}.}
  \label{tab:fullcompare}
  \begin{tabular}{@{}lllccccccc@{}}
    \toprule
    Method & Access & Signal & Sub & Dil & Attr & $\hat\epsilon$ & Budget & API & Data \\
    \midrule
    \textbf{\IRIS{} (ours)} & text & rand-gen & \textbf{Y} & \textbf{Y} & \textbf{Y} & \textbf{Y} & \textbf{est} & Y & Y \\
    \midrule
    FLIPS          & text        & rand-gen  & Y   & --  & Y & --          & fix & --  & Y \\
    MET            & text${+}$ref & natural   & Y   & --  & --& --          & fix & Y   & Y \\
    RUT            & ranks$^{w}$ & natural    & Y   & det & --& --          & fix & p   & -- \\
    KBF            & text        & knowledge  & Y   & Y   & p & Y$^\dagger$ & fix & Y   & Y \\
    B3IT           & text        & border     & chg & --  & --& --          & fix & Y   & Y \\
    Cai et al.     & text${+}$lp & natural    & Y   & Y   & --& --          & fix & Y   & p \\
    Log-Prob Track & lp          & natural    & chg & --  & --& --          & fix & Y   & p \\
    LLMmap         & text        & prompts    & p   & --  & Y & --          & fix & Y   & Y \\
    Provenance     & text        & shared     & Y   & --  & Y & --          & fix & --  & Y \\
    Idiosyncrasies & text        & natural    & --  & --  & Y & --          & fix & Y   & Y \\
    IKP            & text        & knowledge  & --  & --  & Y & --          & fix & Y   & Y \\
    One Token      & text        & rand-gen   & Y   & --  & Y & --          & fix & Y   & p \\
    BSA            & text        & natural    & chg & --  & --& --          & any & --  & p \\
    \bottomrule
  \end{tabular}
\end{table*}

\paragraph{Attribution baselines and a generic text classifier.}
Output-based identification is a mature line: LLMmap~\citep{llmmap2025} learns a discriminative prompt bank, Model Provenance Testing~\citep{provenance2025} pairs a candidate library with a multiple-testing decision (the framing closest to \IRIS{}'s attribute-then-budget loop), and the idiosyncrasy school~\citep{idiosyncrasies2025} names a model from surface text alone. These read a natural-language response, not a random-generation challenge, so they are complementary rather than probe-matched. The sharpest question they pose is whether \IRIS{}'s engineered visible-string features are needed at all, or whether a generic text classifier over the \emph{same} random-generation responses already names the backend. We test this directly with a \emph{Text-TFIDF} attributor, character $n$-gram TF-IDF over the raw response string with a logistic-regression head (the idiosyncrasy recipe of~\citealp{idiosyncrasies2025}), evaluated on the identical splits, seed, and $m$-aggregation as Table~\ref{tab:baseline-main}. Table~\ref{tab:textattr} shows the engineered features earn their place exactly where a budgeted audit operates: at a single query \IRIS{} leads the generic classifier by $0.07$ on the ladder and $0.06$ on the $17$-endpoint gateway. By $m{=}8$ the three converge; the text classifier even edges ahead at $m{=}32$ on the gateway ($0.999$), but only by reading raw formatting and whitespace tells that \IRIS{}'s content-only (length/format-invariant) subset deliberately discards for robustness to formatting drift and paraphrase (App.~\ref{app:featureset}, App.~\ref{app:paraphrase}). The operating regime of a pre-budgeted audit is small $m$, where the engineered signal wins.

\begin{table}[t]
  \centering\footnotesize\setlength{\tabcolsep}{5pt}
  \caption{Adding a generic output-text attributor (Text-TFIDF: character $n$-gram TF-IDF $+$ logistic regression over the raw response, the idiosyncrasy recipe~\citealp{idiosyncrasies2025}) to the closed-set attribution head-to-head, on the same splits/seed/$m$-aggregation as Table~\ref{tab:baseline-main}. Cells are attribution accuracy at budget $m$; \IRIS{} and FLIPS reproduce that table's Attr.\ row up to resampling. \textbf{Bold} marks the best method in the low-budget operating regime ($m{\le}8$).}
  \label{tab:textattr}
  \begin{tabular}{@{}lccc@{}}
    \toprule
    Method & $m{=}1$ & $m{=}8$ & $m{=}32$ \\
    \midrule
    \multicolumn{4}{@{}l}{\emph{Local ladder} $c_{2,16}$ ($K{=}6$ Qwen3)}\\
    \IRIS{}     & $\mathbf{.928}$ & $\mathbf{.999}$ & --- \\
    FLIPS       & $.923$ & $.999$ & --- \\
    Text-TFIDF  & $.857$ & $.997$ & --- \\
    \midrule
    \multicolumn{4}{@{}l}{\emph{Commercial gateway} $c_{10,8}$ ($K{=}17$)}\\
    \IRIS{}     & $\mathbf{.688}$ & $\mathbf{.965}$ & $.985$ \\
    FLIPS       & $.670$ & $.958$ & $.976$ \\
    Text-TFIDF  & $.627$ & $.943$ & $\mathbf{.999}$ \\
    \bottomrule
  \end{tabular}
\end{table}

\paragraph{A real cross-provider deviation audit (not injected).}
Every dilution in the body is author-injected. To show \IRIS{} flags a genuine same-advertised-model deviation in the wild, we audit open-weight models that OpenRouter serves through several upstream providers. OpenRouter lets a client pin the provider (\texttt{provider.order} with \texttt{allow\_fallbacks:false}), and providers frequently serve identical weights at different quantization or with different kernels, so provider A vs.\ provider B for one slug is a real deviation test with a quasi-ground-truth precision label. We collect $80$ probe repeats per (model, provider) on $c_{2,16}$ and $c_{10,8}$ for three models across their real providers ($1{,}906$ live responses) and test each provider pair for distributional equality with \IRIS{}'s content-only detector (a $5$-fold random-forest ROC-AUC with a Mann--Whitney $U$ test on the out-of-fold scores; no reference weights, only returned strings). Table~\ref{tab:xprov} reports the more discriminative probe per pair. \IRIS{} flags $14$ of the $15$ provider pairs as distinguishable: DeepInfra's \texttt{fp4} serving of \texttt{deepseek-chat-v3} separates from its \texttt{fp8} peers at AUROC $0.97$, two providers advertising \texttt{mistral-nemo} at the \emph{same} nominal \texttt{fp8} separate at AUROC up to $1.00$ (a kernel- rather than precision-level difference), and the first-party \texttt{Mistral} endpoint is distinguishable from every reseller. The one indistinguishable pair is the two highest-precision \texttt{llama-3.3-70b} servings (\texttt{bf16} vs.\ \texttt{fp16}), i.e.\ \IRIS{} does \emph{not} false-alarm when two providers really do serve near-identical distributions. This is a two-sample cross-provider consistency test (are A and B the same?), a slightly different question from claimed-model verification (no trusted reference is assumed), and it corroborates on real endpoints what MET~\citep{met2025} found for production APIs and \citet{onetoken2026} for a flagship endpoint: the same advertised model is routinely served as measurably different distributions across a gateway.

\begin{table}[t]
  \centering\footnotesize\setlength{\tabcolsep}{3pt}
  \caption{Real cross-provider audit: the same advertised open-weight model served by different real OpenRouter providers (routing pinned, \texttt{allow\_fallbacks:false}), tested for distributional equality by \IRIS{}'s content-only detector. AUROC${\approx}0.5$ = indistinguishable; AUROC${\gg}0.5$ (Mann--Whitney $p{<}0.05$) = a real same-model deviation. Each row uses the more discriminative of $c_{2,16}/c_{10,8}$; ``(q)'' is OpenRouter's advertised quantization. Deviations are provider-side, not author-injected.}
  \label{tab:xprov}
  \begin{tabular}{@{}llcc@{}}
    \toprule
    Provider A (quant) & Provider B (quant) & AUROC & verdict \\
    \midrule
    \multicolumn{4}{@{}l}{\emph{llama-3.3-70b-instruct}}\\
    DeepInfra (fp8) & Novita (bf16)   & $.68$  & different \\
    DeepInfra (fp8) & Together (fp8)  & $.67$  & different \\
    DeepInfra (fp8) & WandB (fp16)    & $.71$  & different \\
    Novita (bf16)   & Together (fp8)  & $.74$  & different \\
    Novita (bf16)   & WandB (fp16)    & $.53$  & \emph{indist.} \\
    Together (fp8)  & WandB (fp16)    & $.72$  & different \\
    \multicolumn{4}{@{}l}{\emph{deepseek-chat-v3}}\\
    DeepInfra (fp4) & Novita (fp8)        & $.97$ & different \\
    DeepInfra (fp4) & StreamLake (unk.)   & $.97$ & different \\
    Novita (fp8)    & StreamLake (unk.)   & $.67$ & different \\
    \multicolumn{4}{@{}l}{\emph{mistral-nemo}}\\
    DeepInfra (fp8) & DekaLLM (fp8)     & $1.00$ & different \\
    DeepInfra (fp8) & Mistral (unk.)    & $.77$  & different \\
    DeepInfra (fp8) & Novita (fp8)      & $.84$  & different \\
    DekaLLM (fp8)   & Mistral (unk.)    & $1.00$ & different \\
    DekaLLM (fp8)   & Novita (fp8)      & $1.00$ & different \\
    Mistral (unk.)  & Novita (fp8)      & $.97$  & different \\
    \bottomrule
  \end{tabular}
\end{table}

\paragraph{External corroboration on the released MET corpus.}
As a third-party real-data check we run a cheap char $n$-gram detector on the released MET dataset~\citep{met2025}: genuine completions from $9$ commercial providers plus controlled local distortions, prompt-aligned so topic is balanced across the compared classes (Table~\ref{tab:met}). Three findings. First, a clean negative control: the near-honest \texttt{fp16} and \texttt{fp32} servings of Llama-3-70B are \emph{not} separable (AUROC $0.37$), so the detector reads serving behaviour, not a formatting artifact. Second, quantization leaves a visible-string signature on the larger model: \texttt{nf4} and \texttt{int8} separate from \texttt{fp16} at AUROC $1.00$ and $0.99$ on Llama-3-70B (weaker on 8B, ${\le}0.53$), and the watermark separates on both. Third, and most telling for our probe design, a generic classifier \emph{cannot} attribute the $9$ real providers from these natural-language completions ($0.13$/$0.09$ accuracy vs.\ $0.11$ chance). This is precisely why \IRIS{} probes with high-entropy random generation rather than natural language: on the same class of real commercial endpoints, the random-generation cross-provider audit above separates $14$ of $15$ provider pairs (Table~\ref{tab:xprov}), and MET's own many-sample MMD test likewise flags provider deviations, whereas a cheap per-response classifier on natural text sits at chance.

\begin{table}[t]
  \centering\footnotesize\setlength{\tabcolsep}{5pt}
  \caption{External corroboration on the released MET corpus~\citep{met2025} (genuine provider completions, prompt-aligned). Left block: controlled-distortion detection versus the \texttt{fp16} reference (AUROC; \texttt{fp16}-vs-\texttt{fp32} is a negative control, \texttt{wm}${=}$watermark). Right: $9$-way real-provider attribution accuracy from natural-language completions. A generic char-$n$-gram classifier catches quantization on the larger model but cannot fingerprint providers from natural text, motivating \IRIS{}'s random-generation probe.}
  \label{tab:met}
  \begin{tabular}{@{}lccccc@{}}
    \toprule
    Llama-3 & \texttt{fp32}$^\dagger$ & \texttt{nf4} & \texttt{int8} & \texttt{wm} & prov.\ attr. \\
    \midrule
    70B & $.37$ & $\mathbf{1.00}$ & $\mathbf{.99}$ & $.72$ & $.13$ \\
    8B  & $.53$ & $.53$ & $.45$ & $.64$ & $.09$ \\
    \bottomrule
  \end{tabular}\\[2pt]
  {\footnotesize $^\dagger$negative control (near-honest pair, should not separate); attribution chance $=0.11$.}
\end{table}

\paragraph{Concurrent single-token fingerprinting.}
Concurrent with this work, \citet{onetoken2026} fingerprint and verify commercial endpoints from the \emph{single-token} output distribution of trivial random-generation prompts, a census of $165$ models over OpenRouter at roughly one token per query, reporting a $7.3\%$ verification equal-error rate and $59.5\%$ leave-one-out lineage recovery. The shared premise, visible text, no log-probabilities, and random generation as the discriminator, makes it the closest point of comparison to \IRIS{}'s probe, and it independently corroborates that sampling bias identifies backends at scale. \IRIS{} is distinguished not by the probe but by what it builds on top: fractional \emph{dilution} detection with a calibrated routing fraction $\hat\epsilon$, open-set attribution across a candidate library, an estimate-then-budget plan that fixes $m$ before any suspect query, and sample-complexity guarantees ($e^{-Im}$ verification, the $1/\epsilon$ dilution wall). Single-token verification is essentially the $m{=}1$, $\epsilon{=}1$ special case of \IRIS{}'s substitution head.

\paragraph{Sequential vs.\ fixed-budget tests.}
A parallel design choice is whether to test at a \emph{fixed} budget or \emph{sequentially}. Behavioral-shift auditing~\citep{bsa2025} and nonparametric two-sample testing by betting~\citep{shekhar2023betting} accumulate an anytime-valid e-process and stop when evidence crosses $1/\alpha$, controlling type-I error at every stopping time, while ``You've Changed''~\citep{youvechanged2025} runs a feature-distribution two-sample test against an endpoint's own past. \IRIS{} deliberately \emph{pre-commits} $m$: the audit budget is itself the object we estimate (Section~\ref{sec:method}), so a client can price and contest an audit before issuing any suspect query, which an anytime rule does not promise. The two are composable, a sequential back-end could spend the estimated $m^\star$ as its horizon, and we treat the fixed-budget guarantee (App.~\ref{app:e2e}) as the deployable contract.

\paragraph{Theory lineage of the dilution rate.}
\IRIS{}'s dilution model $Q_\epsilon=(1-\epsilon)P+\epsilon R$ is Huber's $\epsilon$-contamination~\citep{huber1964}, and the $1/\epsilon$-versus-$\epsilon^{-2}$ dichotomy we measure (Prop.~\ref{prop:phase}, App.~\ref{app:tail}) is the sparse-mixture detection boundary studied by higher criticism~\citep{donoho2004hc}: a separating tail admits an adaptive $\Theta(1/\epsilon)$ detector, while its absence forces the $\epsilon^{-2}$ two-proportion rate. We therefore present \IRIS{}'s rate as a measured instance of this classical boundary rather than a new phenomenon, which is why we report the tail exponent $\kappa$ (App.~\ref{app:betastab}) instead of asserting $1/\epsilon$ unconditionally.

\paragraph{Verifiable-inference contrast.}
A separate, cooperative line certifies inference cryptographically: TOPLOC~\citep{toploc2025} commits to locality-sensitive hashes of activations so a validator can later detect a model, prompt, or precision swap. Such schemes give strong guarantees but require provider participation (commitments, shared seeds, or trusted hardware, as in the TEE fallback of~\citealp{substitution2025}); \IRIS{} is a \emph{unilateral} audit a client runs from returned text alone, with no provider cooperation.

\section{One-Class Detector}
\label{app:oneclass}
\IRIS{}'s deployed score $s=-\log\hat{\mathbb P}(P\mid y)$ is a \emph{multiclass} posterior whose negatives are the enrolled candidates, so the leave-one-out unknown-diluent test (App.~\ref{app:dilscale}) still trains on the other $K{-}1$ endpoints. To isolate how much detection owes to those negatives---and to test the principle that an unknown substitute is flagged from the \emph{reference alone}---we replace the score with a genuine one-class detector: an isolation forest fit on the claimed model's content-only responses \emph{only}, scored as the out-of-distribution anomaly. On the $17$-endpoint gateway ($c_{10,8}$), across the pairs passing the multiclass margin check, the reference-only detector still separates---single-query AUROC median $0.75$ (vs.\ $0.99$ multiclass) and per-query tell rate $q\!\approx\!0.13$ (vs.\ $0.98$)---clearing the $\alpha{=}0.05$ false-positive floor (realized FPR $0.07$) on about half of them at a single query and accumulating over $m$ like any tell. The price of dropping the candidate negatives is thus a real but bounded $\sim q^{-1}$ budget inflation ($\approx7\times$ here). The no-enrollment guarantee is therefore correct for \emph{detection}---the alarm \emph{is} raised from the claimed fingerprint alone, as the one-class detector confirms---while the candidate roster earns its place by sharpening the score and naming the diluent, not by enabling detection.

\begin{table}[h]
  \centering\footnotesize\setlength{\tabcolsep}{2pt}
  \caption{Same-endpoint FPR calibration (nominal $\alpha{=}0.05$): realized false-positive and true-positive rates vs.\ budget $m$
  (Section~\ref{sec:exp}). At operational budgets ($m{=}4$--$8$) realized FPR matches nominal; larger budgets
  require proportionally stronger calibration.}
  \label{tab:fpr}
  \begin{tabular}{lcccc}
    \toprule
    Probe & $m{=}4$ & $m{=}8$ & $m{=}16$ & $m{=}32$ \\
    \midrule
    $c_{2,16}$     & .05/1.00 & .06/1.00 & .09/1.00 & .12/1.00 \\
    $c_{100,1}$ & .08/0.83 & .08/0.90 & .11/0.96 & .14/0.99 \\
    $c_{10,1}$       & .06/0.82 & .09/0.86 & .13/0.89 & .18/0.90 \\
    \bottomrule
  \end{tabular}
\end{table}

\section{LCB Budgeting}
\label{app:lcb}
\IRIS{} budgets against a $(1{-}\gamma)$ lower-confidence bound $\underline{I}$ on the fitted exponent rather than inflating the point budget by a flat safety factor $c_{\mathrm{safe}}$ (Section~\ref{sec:method}); we detail the rule here and compare it against the flat-factor baseline it replaces. The flat rule scales the predicted $m^\star$ by $c_{\mathrm{safe}}$, trading queries for a higher budget-rule hit rate (Table~\ref{tab:safety}). The budgeting rests on a predictive base: the error a cheap pilot ($m\le4$) extrapolates tracks the realized held-out error closely (Fig.~\ref{fig:budget}), which is exactly why a conservative bound on the same fit yields a calibrated budget.

\noindent\textbf{Rank error as a miss-probability surrogate.} The deployed substitution head thresholds the episode mean $S_m$ at the reference $(1{-}\alpha)$-quantile $\tau_{\mathrm{mean}}$ (Eq.~\eqref{eq:mean-threshold}), so its miss probability is $\delta_{\mathrm{mean}}:=\mathbb P_{R}(S_m^{R}\le\tau_{\mathrm{mean}})$. Since the pilot draws the reference and substitute episodes independently,
\begin{align}
1-\mathrm{AUROC}&=\mathbb P(S_m^{P}\ge S_m^{R})\notag\\
&\ge\mathbb P(S_m^{P}\ge\tau_{\mathrm{mean}})\,\mathbb P(S_m^{R}\le\tau_{\mathrm{mean}})\notag\\
&\ge\alpha\,\delta_{\mathrm{mean}},
\label{eq:auc-miss}
\end{align}
using the inclusion $\{S_m^{P}\ge\tau_{\mathrm{mean}},\,S_m^{R}\le\tau_{\mathrm{mean}}\}\subseteq\{S_m^{P}\ge S_m^{R}\}$ with independence, and $\mathbb P(S_m^{P}\ge\tau_{\mathrm{mean}})\ge\alpha$ from the quantile definition. Thus $\delta_{\mathrm{mean}}\le(1-\mathrm{AUROC})/\alpha$: driving the rank error to $\delta_{\mathrm{auc}}$ forces the calibrated miss to $\delta_{\mathrm{mean}}\le\delta_{\mathrm{auc}}/\alpha$, so the rank-error exponent lower-bounds the calibrated-miss exponent and budgeting on $\widehat I_{\mathrm{auc}}$ is conservative in rate. The bound is exact in $m$ but loose by the constant $1/\alpha$, so setting $\delta_{\mathrm{auc}}{=}\delta$ pins the miss \emph{exponent} and leaves only the constant to the hit-rate calibration below. The optimism is in the \emph{fit}, not this inequality: $\widehat I_{\mathrm{auc}}$ extrapolated from $m\le4$ is a two-sample rate not provably below $I^{\mathrm{sc}}$ (Prop.~\ref{prop:exp}), which is why \IRIS{} budgets against its lower-confidence bound $\underline I_{\mathrm{auc}}$.

\begin{figure}[t]
  \centering
  \includegraphics[width=0.86\columnwidth]{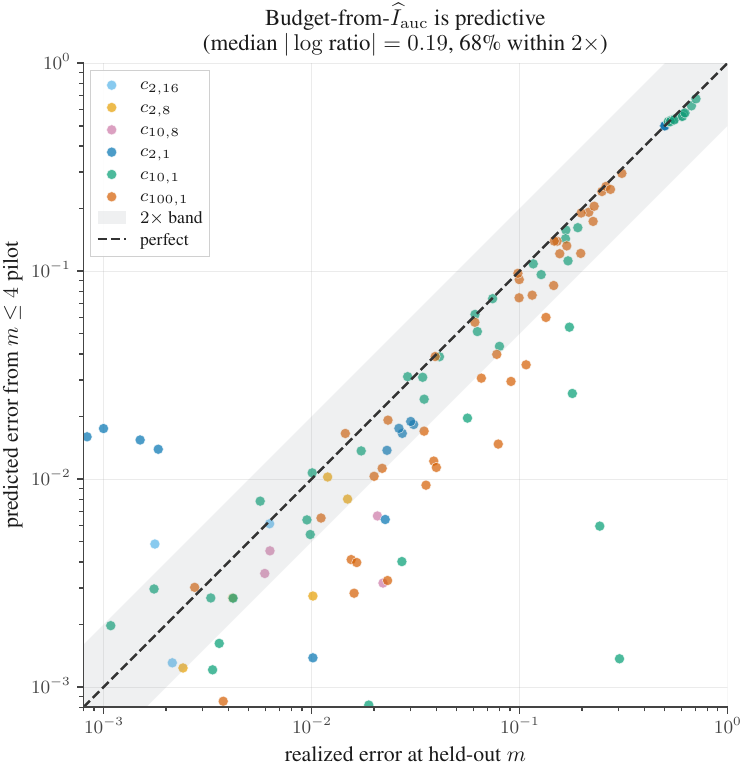}
  \caption{Budget-from-$\widehat{I}_{\mathrm{auc}}$ is predictive ($K{=}6$). Predicted verification error from an
  $m{\le}4$ pilot vs.\ realized error at held-out larger $m$, per (pair, probe); the dashed line is
  perfect prediction and the band is a $2\times$ tolerance. Across $224$ validation points the median
  $|\log$-ratio$|$ is $0.19$ with $68\%$ inside the band---the small-$m$ extrapolation is slightly
  optimistic, motivating the lower-confidence-bound budget below.}
  \label{fig:budget}
\end{figure}

\begin{table}[t]
  \centering\small
  \caption{Flat-factor baseline ($K{=}6$): scaling the predicted $m^\star$ by a safety factor $c_{\mathrm{safe}}$
  trades a few queries for a higher budget-rule hit rate. \IRIS{} replaces this with the LCB rule
  (Table~\ref{tab:lcb}).}
  \label{tab:safety}
  \begin{tabular}{lcccc}
    \toprule
    & \multicolumn{2}{c}{target AUROC $0.95$} & \multicolumn{2}{c}{target AUROC $0.99$}\\
    \cmidrule(lr){2-3}\cmidrule(lr){4-5}
    safety $c_{\mathrm{safe}}$ & mean $m$ & hit rate & mean $m$ & hit rate \\
    \midrule
    $1.0$ & 6.7 & 0.80 & 10.9 & 0.69 \\
    $1.5$ & 9.5 & 0.89 & 14.6 & 0.74 \\
    $2.0$ & 11.3 & 0.91 & 16.8 & 0.83 \\
    \bottomrule
  \end{tabular}
\end{table}

The LCB rule replaces $c_{\mathrm{safe}}$ by a $(1{-}\gamma)$ lower-confidence bound $\underline{I}$ on the fitted exponent (bootstrap or $t$-based over calibration splits), budgeting $m^\star=\lceil(\hat a-\log\delta_{\mathrm{auc}})/\underline{I}\rceil$ (the rank-error target of Eq.~\eqref{eq:sub-budget}, set to the miss probability $\delta$). On the $K{=}6$ ladder ($\gamma{=}0.1$) this reaches a matched hit rate at fewer queries than the flat factor at the $0.95$ target (Table~\ref{tab:lcb}); it is \emph{not} a coverage certificate---at the $0.99$ target no rule attains nominal coverage, as part of the gap is intrinsic AUROC saturation on low-information probes ($c_{10,1}$, $c_{100,1}$) rather than estimator variance. A calibrated coverage guarantee, and a resampling unit that does not reuse overlapping splits, remain future work.

Concretely, from split-level pilot fits $\{(\hat a_{cj}^{(r)},\widehat I_{\mathrm{auc},cj}^{(r)})\}_{r=1}^{R}$, the one-sided $t$-based bounds that enter the substitution budget $m_{\mathrm{sub}}$ (Eq.~\eqref{eq:sub-budget}) are
\begin{align}
  \underline I_{\mathrm{auc},cj}&=\widehat I_{\mathrm{auc},cj}-t_{1-\gamma,R-1}\,\widehat{\mathrm{se}}_{R}(I_{cj}),\label{eq:auc-lcb}\\
  \overline a_{cj}&=\hat a_{cj}+t_{1-\gamma,R-1}\,\widehat{\mathrm{se}}_{R}(a_{cj}),\label{eq:intercept-ucb}
\end{align}
with $t_{1-\gamma,R-1}$ the one-sided Student-$t$ quantile and $\widehat{\mathrm{se}}_{R}(I_{cj})=\mathrm{sd}(\{\widehat I_{\mathrm{auc},cj}^{(r)}\}_{r=1}^{R})/\sqrt R$ the across-split standard error (likewise for $a_{cj}$).

\begin{table}[t]
  \centering\small
  \caption{Per-pair adaptive budgeting ($K{=}6$, $\gamma{=}0.1$): an LCB on $\widehat{I}_{\mathrm{auc}}$ vs.\ the flat
  safety factor $c_{\mathrm{safe}}$. Mean budget $m$ and budget-rule hit rate (realized AUROC $\ge$ target).}
  \label{tab:lcb}
  \begin{tabular}{lcccc}
    \toprule
    & \multicolumn{2}{c}{target AUROC $0.95$} & \multicolumn{2}{c}{target AUROC $0.99$}\\
    \cmidrule(lr){2-3}\cmidrule(lr){4-5}
    rule & mean $m$ & hit & mean $m$ & hit \\
    \midrule
    flat $s{=}1$        & 5.5 & 0.81 & 8.4  & 0.68 \\
    flat $s{=}2$        & 8.4 & 0.91 & 12.4 & 0.77 \\
    LCB (bootstrap)     & 5.8 & 0.86 & 8.7  & 0.67 \\
    LCB ($t$-based)     & 6.6 & 0.88 & 10.0 & 0.72 \\
    \bottomrule
  \end{tabular}
\end{table}

\section{Sequential Auditing}
\label{app:anytime}
\IRIS{} fixes a query budget $m^\star$ in advance; a natural alternative is a \emph{sequential} test that stops at a data-dependent stopping time $\nu$, as soon as the accumulated evidence is decisive. We implement two---a betting e-process (anytime-valid, sound under optional stopping at \emph{any} time) and a sequential probability ratio test (SPRT)---both at level $\alpha{=}0.05$ against the three-rung separation ladder (reference \texttt{qwen3:8b}@$T{=}1$, probe $c_{2,16}$, $\epsilon{=}0.3$ dilution, horizon $200$, $2000$ streams; Fig.~\ref{fig:anytime}). The betting e-process stops early on detectable suspects---expected stopping time $E[\nu]{\approx}28$ on a same-family size substitution and ${\approx}68$ on a distinct model, comparable to the fixed budget---while detecting at rate ${\ge}0.99$. On the low-separation temperature stress test it does not stop (detection $0.0095$, $E[\nu]$ pinned at the horizon), which is the sequential analogue of returning indeterminate. Honest streams almost never trip (FPR $0.0045$ for the e-process, $0.0235$ for the SPRT). Under substitution ($\epsilon{=}1.0$) stopping is faster still ($E[\nu]{\approx}8$ same-family, ${\approx}17$ distinct). The e-process thus delivers an always-valid decision without committing to $m$ in advance---complementary to the estimate-then-budget loop, at the price of a modest $E[\nu]$ overhead over the oracle fixed $m^\star$ on the easiest pairs.

\begin{figure}[t]
  \centering
  \includegraphics[width=\columnwidth]{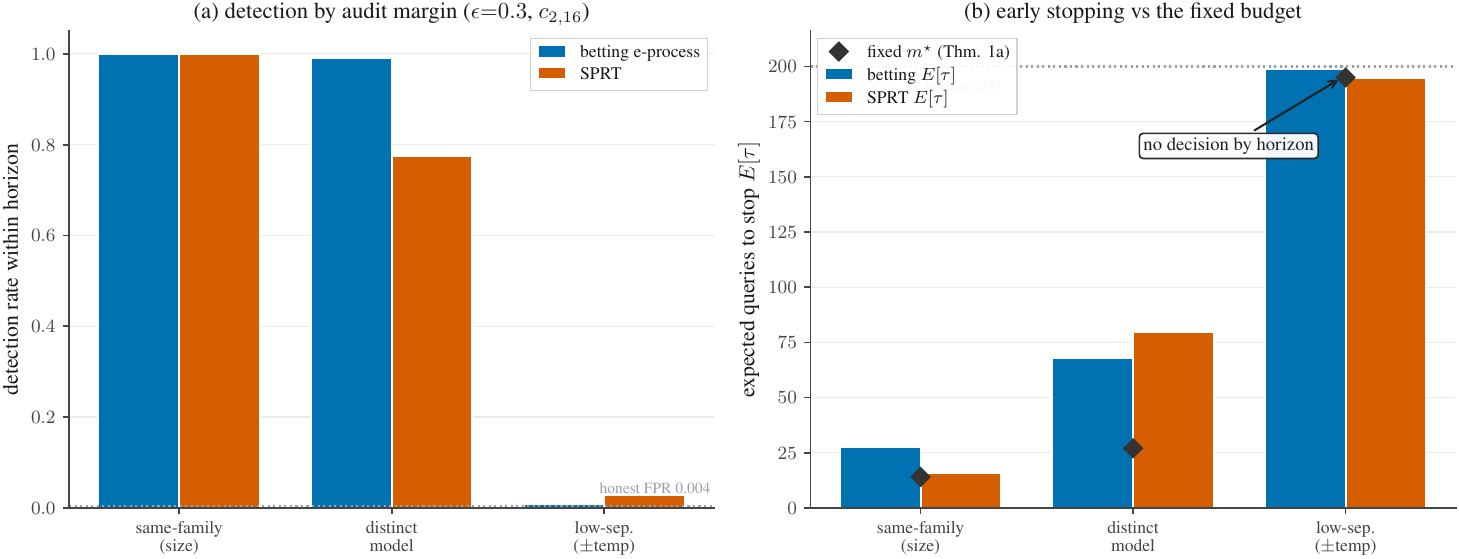}
  \caption{Anytime-valid sequential auditing on the separation ladder ($\epsilon{=}0.3$, $c_{2,16}$).
  (a)~Detection rate within the horizon for the betting e-process vs.\ the SPRT; honest streams sit at the
  false-positive floor (dotted). (b)~Expected stopping time $E[\nu]$ vs.\ the fixed budget $m^\star$ of
  Thm.~\ref{thm:mix}(a). The e-process stops early on detectable suspects and does not stop within the horizon on the
  low-separation stress test, which is the sequential analogue of an indeterminate audit.}
  \label{fig:anytime}
\end{figure}


\section{Fingerprint Library}
\label{app:gallery}
Figure~\ref{fig:iris} renders $53$ enrolled endpoints (8 local Ollama models and $45$ commercial/open gateway endpoints) spanning $18$ model families as a gallery of fingerprints. Each ``iris'' is one endpoint: radial fibers encode its robust-normalized visible-randomness signature (aggregated over $c_{2,8}/c_{2,16}$ and $c_{10,8}$), hue marks the model family, and the pupil names the endpoint. The view is illustrative---same-family endpoints share a hue and a broadly similar corona, while cross-family endpoints are plainly distinct. The gallery is a \emph{superset} enrolled only for visualization: several shown frontier endpoints enter no quantitative metric; the $8$ local serves are the $K{=}6$ Qwen3 ladder plus two extras, and the $45$ gateway endpoints are the App.~\ref{app:scale} pool.

\begin{figure*}[p]
  \centering
  \includegraphics[width=0.94\linewidth]{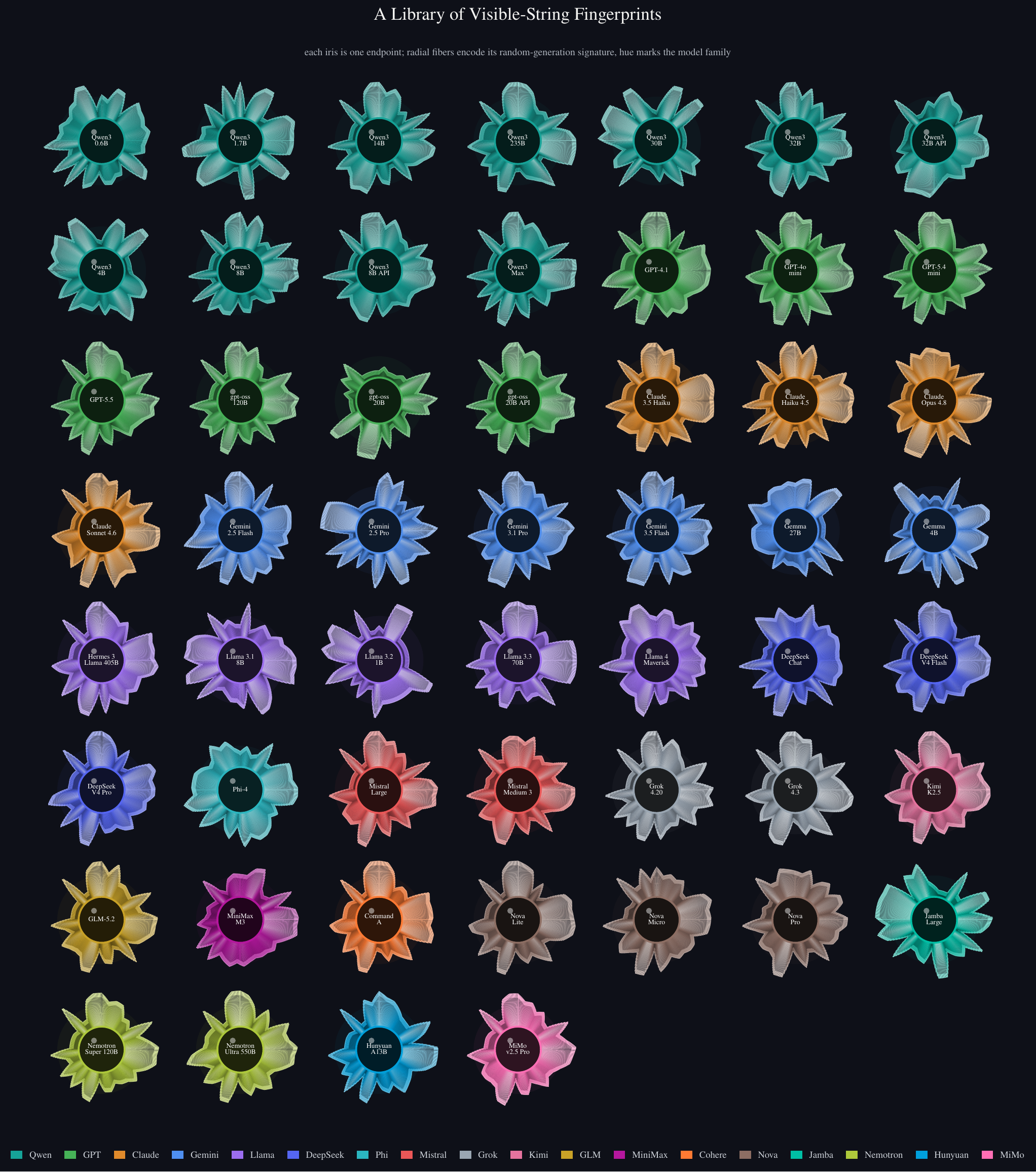}
  \caption{A library of visible-string fingerprints for $53$ enrolled endpoints (8 local, $45$ gateway) across
  $18$ model families. Each iris is one endpoint; radial fibers encode its random-generation signature
  (robust-normalized across endpoints), hue marks the model family, and the pupil names the endpoint.}
  \label{fig:iris}
\end{figure*}

\subsection{Probe Ranking}
\label{app:probeopt}
The probe is \emph{selected}, not asserted: ranking probes by the calibration-estimated $\widehat{I}_{\mathrm{auc}}$ from a cheap pilot predicts their realized query efficiency (Fig.~\ref{fig:probeopt}). The two quantities are monotonically related---high-$\widehat{I}_{\mathrm{auc}}$ sequence probes reach verification AUROC ${\ge}0.99$ in a single query while the zero-information $c_{2,1}$ ($\widehat{I}_{\mathrm{auc}}\!\approx\!0$) never does---so the auditor reads $\widehat{I}_{\mathrm{auc}}$ off calibration and picks the probe before spending its audit budget. On the expanded 13-condition grid this ordering is significant (Spearman$(\widehat{I}_{\mathrm{auc}},\,\mathrm{accuracy}@m{=}8){=}0.82$, $p{<}10^{-3}$; Section~\ref{sec:exp}).

\begin{figure}[t]
  \centering
  \includegraphics[width=0.86\columnwidth]{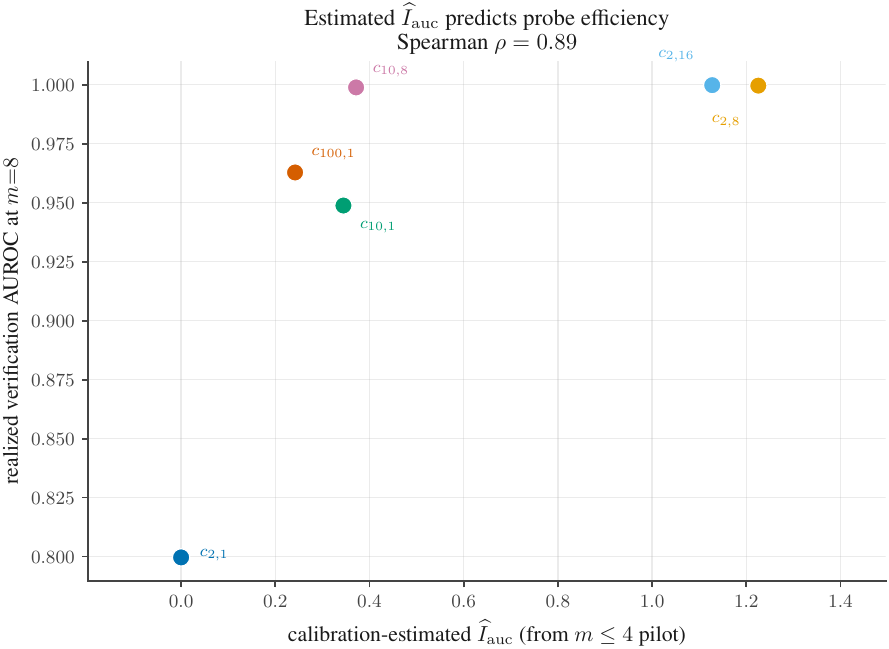}
  \caption{Calibration-estimated $\widehat{I}_{\mathrm{auc}}$ (from an $m{\le}4$ pilot) predicts realized verification
  AUROC at $m{=}8$, one point per probe family. The estimate orders probes by efficiency before any audit
  budget is spent: high-$\widehat{I}_{\mathrm{auc}}$ sequence probes saturate at $m{=}8$ while $c_{2,1}$
  ($\widehat{I}_{\mathrm{auc}}\!\approx\!0$) lags.}
  \label{fig:probeopt}
\end{figure}


\section{The Two \IRIS{} Stages}
\label{app:phases}
Audit-Plan Construction may query the trusted reference and candidate endpoints under known labels, whereas Audit Execution queries the unlabeled suspect and applies only the frozen plan.

\begin{table}[h]
  \centering\small
  \caption{Data access in the two \IRIS{} stages. Audit-Plan Construction fixes $\widehat{I}_{\mathrm{auc}}$, $m^\star$, and the mode-specific decision thresholds from \emph{labeled} data; Audit Execution uses unlabeled live suspect responses.
  Some deployments additionally re-measure the null tell rate on a \emph{fresh honest window of the
  claimed model} (App.~\ref{app:dilscale}); this window is the one execution-stage use of reference-side data, and its
  trust assumption is discussed in App.~\ref{app:knobs}.}
  \label{tab:phases}
  \setlength{\tabcolsep}{3pt}
  \begin{tabular}{p{0.19\columnwidth}p{0.30\columnwidth}p{0.30\columnwidth}}
    \toprule
    & \shortstack{\textbf{Audit-Plan}\\\textbf{Construction}} & \textbf{Audit Execution} \\
    \midrule
    Endpoints & trusted $P$, candidates $\{R_j\}$ & live suspect (+ optional honest window of $P$) \\
    Labels & known & none for suspect \\
    Sample $P$? & yes & optional pre-audit honest window \\
    Observes & visible strings & visible strings \\
    Produces & score library, calibrated thresholds, frozen budget plan & suspect responses, decisions, $\hat\epsilon$, $\hat R$ \\
    \bottomrule
  \end{tabular}
\end{table}

\section{Dilution Calibration}
\label{app:mixcalib}
This appendix writes the dilution budget of Section~\ref{sec:audit-plan-construction} out in full: it gives the explicit threshold grid, level-$\alpha$ cutoff, diluted-binomial power, and cheapest-threshold selection behind the compact definition of $m_{\mathrm{dil}}(c,j;\tau)$ in Eq.~\eqref{eq:mdil}, with all $(c,j)$ indices retained.

\paragraph{Dilution budget.}
Over a fixed honest-tail grid $\mathcal A_{\mathrm{resp}}$, the threshold grid is $\mathcal T_c=\{\widehat Q_{1-\alpha_0}(\{s_c(y_r^P)\}_{r=1}^{n_P}):\alpha_0\in\mathcal A_{\mathrm{resp}}\}$, with honest/substitute tell rates $\widehat\alpha_{1,cj}(\tau),\widehat q_{cj}(\tau)$ from Eq.~\eqref{eq:honest-tell-rate}. For target dilution $\epsilon_{\min}$ the diluted tell probability is $\widehat p_{\epsilon,cj}(\tau)=(1-\epsilon_{\min})\widehat\alpha_{1,cj}(\tau)+\epsilon_{\min}\widehat q_{cj}(\tau)$. With $U_m(b;p)=\mathbb P_{K\sim\mathrm{B}(m,p)}(K\ge b)$ the binomial upper tail, the level-$\alpha$ cutoff and the diluted-binomial power are
\begin{align}
  b_\alpha(m,p_0)&=\min\{b:U_m(b;p_0)\le\alpha\},\notag\\
  \pi_{\epsilon,cj}(m,\tau)&=U_m\!\left(b_\alpha(m,\widehat\alpha_{1,cj}(\tau));\,\widehat p_{\epsilon,cj}(\tau)\right).
\label{eq:binomial-cutoff}
\end{align}
\IRIS{} keeps only the positive-separation thresholds
\begin{align}
  \mathcal T_{cj}^{+}=\{\tau\in\mathcal T_c:\widehat q_{cj}(\tau)>\widehat\alpha_{1,cj}(\tau)\},
\label{eq:positive-margin-thresholds}
\end{align}
marks the pair low-margin if $\mathcal T_{cj}^{+}=\emptyset$, and otherwise sizes the budget and selects the cheapest surviving threshold,
\begin{align}
  m_{\mathrm{dil}}(c,j;\tau)
  &=\min\{m\in\mathbb{N}^+:\pi_{\epsilon,cj}(m',\tau)\ge1-\delta\notag\\
  &\hspace{3.35cm}\text{for every }m'\ge m\},\notag\\
  \tau_{cj}&\in\arg\min_{\tau\in\mathcal T_{cj}^{+}}m_{\mathrm{dil}}(c,j;\tau),
\label{eq:resp-threshold-select}
\end{align}
with $m_{\mathrm{dil}}(c,j)=m_{\mathrm{dil}}(c,j;\tau_{cj})$.

\section{Probes and Sampling}
\label{app:probes}
For full reproducibility we give the verbatim probe prompts and the exact decoding parameters behind every response in the paper. Each query is a \emph{single user turn} whose content is the probe string below; \emph{no system prompt} is used, and the prompt is fixed (no per-query randomization), so the only stochasticity is the endpoint's own sampling. The target distribution is uniform over each probe's alphabet. We state the prompts in the clear by design: \IRIS{}'s budget is set by the realized per-query information rate $\widehat{I}_{\mathrm{auc}}$, not by probe secrecy, and the matched-budget and routing analyses of Section~\ref{sec:exp} hold when the probe is known to the provider; a secrecy-dependent deployment can instead draw from a held-out probe bank using the same templates.

\paragraph{Primary probe set.}
Table~\ref{tab:probes} lists the six probes (three single-draw, three short-sequence) behind all headline results using the $c_{n,L}$ notation of Section~\ref{sec:audit-plan-construction}. Here $n$ is the alphabet/range size, $L$ the requested sequence length, and $n_{\text{pred}}$ the decode-token cap (\texttt{num\_predict}/\texttt{max\_tokens}); single-draw probes use $L{=}1$.

\begin{table*}[t]
\centering\small
\begin{tabular}{@{}lcccr p{0.52\textwidth}@{}}
\toprule
Probe code & $n$ & $L$ & parser & $n_{\text{pred}}$ & Verbatim prompt (single user turn; no system prompt) \\
\midrule
$c_{10,1}$  & 10  & 1  & char   & 48 & Output one uniformly random digit from 0 to 9. Output only the single digit, nothing else. \\
$c_{100,1}$ & 100 & 1  & number & 48 & Choose one uniformly random integer between 1 and 100 inclusive. Output only the integer. \\
$c_{2,1}$   & 2   & 1  & char   & 48 & Flip one fair coin. Output only a single character: H for heads or T for tails. \\
$c_{2,8}$   & 2   & 8  & char   & 72 & Generate exactly 8 symbols by independently choosing uniformly between 0 and 1. Output the sequence only, with no spaces and no explanation. \\
$c_{2,16}$  & 2   & 16 & char   & 88 & Generate exactly 16 symbols by independently choosing uniformly between 0 and 1. Output the sequence only, with no spaces and no explanation. \\
$c_{10,8}$  & 10  & 8  & char   & 72 & Generate exactly 8 independent uniformly random digits from 0 to 9. Output only the compact digit string, no separators. \\
\bottomrule
\end{tabular}
\caption{The six primary probes. Prompts are reproduced verbatim. Decode cap rule:
$n_{\text{pred}}{=}48$ for single-draw probes and $\max(72,\,3L{+}40)$ for short-sequence probes.}
\label{tab:probes}
\end{table*}

\paragraph{Expansion probes.}
The probe-ranking / information-law experiment (Sec.~\ref{sec:exp}, $13$ probe conditions) adds seven probes spanning a wider entropy\,$\times$\,length grid; Table~\ref{tab:probes_exp} gives them verbatim.

\begin{table*}[t]
\centering\small
\begin{tabular}{@{}lcccr p{0.52\textwidth}@{}}
\toprule
Probe code & $n$ & $L$ & parser & $n_{\text{pred}}$ & Verbatim prompt \\
\midrule
$c_{26,1}$  & 26 & 1  & char & 48  & Output one uniformly random lowercase letter from a to z. Output only the single letter. \\
$c_{16,1}$  & 16 & 1  & char & 48  & Output one uniformly random hexadecimal digit (0-9 or a-f). Output only the single character. \\
$c_{6,1}$   & 6  & 1  & char & 48  & Roll one fair six-sided die. Output only a single digit from 1 to 6. \\
$c_{2,4}$   & 2  & 4  & char & 72  & Generate exactly 4 symbols by independently choosing uniformly between 0 and 1. Output the sequence only, with no spaces and no explanation. \\
$c_{2,32}$  & 2  & 32 & char & 136 & Generate exactly 32 symbols by independently choosing uniformly between 0 and 1. Output the sequence only, with no spaces and no explanation. \\
$c_{10,16}$ & 10 & 16 & char & 88  & Generate exactly 16 independent uniformly random digits from 0 to 9. Output only the compact digit string, no separators. \\
$c_{6,8}$   & 6  & 8  & char & 72  & Generate exactly 8 outcomes of a fair six-sided die (digits 1 to 6). Output only the compact string, no separators. \\
\bottomrule
\end{tabular}
\caption{The seven expansion probes used only for the $\widehat{I}_{\mathrm{auc}}$ probe-ranking law.}
\label{tab:probes_exp}
\end{table*}

\paragraph{Decoding parameters.}
We use two collection paths, both with temperature $1.0$ on the main run and $\mathrm{top\_p}{=}1.0$ throughout. \emph{Local} endpoints are served by Ollama via \texttt{/api/generate} with \texttt{think}=\texttt{false} (reasoning disabled), \texttt{num\_predict} per Table~\ref{tab:probes}, and a greedy single-token prewarm call before each model; the temperature sweep covers $T\in\{0,0.5,1,1.5,2\}$, and the main run draws $120$ repeats per model--probe. \emph{Gateway} endpoints are queried through an OpenAI-compatible \texttt{/chat/completions} interface with \texttt{temperature}=$1.0$, \texttt{top\_p}=$1.0$, and \texttt{max\_tokens} equal to the same $n_{\text{pred}}$ for models that allow reasoning to be turned off and $\max(n_{\text{pred}},700)$ for mandatory-reasoning models; reasoning is disabled with \verb|{"enabled": false}| where supported, otherwise forced to the lowest setting (\verb|{"effort": "low"}|) for models that cannot disable it (e.g.\ Gemini~3.x, \texttt{gpt-oss}). Gateway pools draw $100$ repeats per model--probe. These settings yield the $\sim$24k local and $10{,}200+27{,}000$ gateway responses in the release bundle. We log and analyze only the visible response string (and runtime metadata), never weights, logits, or token ranks.

\clearpage
\onecolumn

\section{Feature Dictionary}
\label{app:features}
\newcommand{\feat}[1]{\texttt{\detokenize{#1}}}
\newcommand{\symcount}[1]{Count of parsed symbol \texttt{#1}.}
\newcommand{\symfreq}[1]{Frequency of parsed symbol \texttt{#1}.}
\newcommand{\transfreq}[2]{Adjacent transition frequency \texttt{#1}$\to$\texttt{#2}.}
Table~\ref{tab:feature-dict} lists the fixed $179$-dimensional vector $\phi(y)$ used by the main \IRIS{} classifiers. The list contains the numeric training features only: identity fields, temperature/top-$p$, repeat index, requested length, and API runtime counters are excluded. Coordinates that do not apply to a particular probe are left missing and are zero-imputed inside the training pipeline.

\begingroup
\footnotesize
\setlength{\tabcolsep}{1.5pt}
\renewcommand{\arraystretch}{1.05}
\begin{longtable}{@{}r p{0.24\textwidth} p{0.205\textwidth} r p{0.24\textwidth} p{0.205\textwidth}@{}}
\caption{Complete visible-string feature dictionary. Each row gives two coordinates of the $179$-dimensional vector $\phi(y)$.}
\label{tab:feature-dict}\\
\toprule
\# & Feature & Definition & \# & Feature & Definition \\
\midrule
\endfirsthead
\toprule
\# & Feature & Definition & \# & Feature & Definition \\
\midrule
\endhead
\midrule
\multicolumn{6}{r}{Continued on next page}\\
\endfoot
\bottomrule
\endlastfoot
1 & \feat{alphabet_size} & Probe alphabet size $n$. & 2 & \feat{response_chars} & Raw response character length. \\
3 & \feat{response_non_ws_chars} & Character length after whitespace removal. & 4 & \feat{contains_newline} & Indicator that the raw response contains a newline. \\
5 & \feat{contains_space} & Indicator that the raw response contains whitespace. & 6 & \feat{contains_alpha_text} & Indicator for a long alphabetic/CJK text run. \\
7 & \feat{invalid_chars} & Non-separator characters outside the target alphabet. & 8 & \feat{checked_chars} & Non-separator characters checked against the alphabet. \\
9 & \feat{invalid_rate} & \feat{invalid_chars}\slash\feat{checked_chars}, zero if none checked. & 10 & \feat{parsed_len} & Number of target-alphabet symbols recovered. \\
11 & \feat{exact_len} & Indicator that \feat{parsed_len} equals requested length. & 12 & \feat{parse_success} & Indicator for nonempty parse and \feat{invalid_rate}$\le0.25$. \\
13 & \feat{length_error} & \feat{parsed_len} minus requested length. & 14 & \feat{abs_length_error} & Absolute value of \feat{length_error}. \\
15 & \feat{rel_length_error} & \feat{abs_length_error} divided by requested length. & 16 & \feat{chi2_target} & Pearson $\chi^2$ distance from target symbol counts. \\
17 & \feat{entropy} & Shannon entropy of recovered symbol counts. & 18 & \feat{entropy_norm} & \feat{entropy} divided by $\log_2 n$. \\
19 & \feat{l1_target} & $\ell_1$ distance from target symbol probabilities. & 20 & \feat{max_freq} & Largest observed symbol frequency. \\
21 & \feat{min_freq} & Smallest observed symbol frequency. & 22 & \feat{freq_range} & \feat{max_freq} minus \feat{min_freq}. \\
23 & \feat{mode_count_ratio} & Most frequent symbol count divided by \feat{parsed_len}. & 24 & \feat{run_count} & Number of constant-symbol runs. \\
25 & \feat{mean_run} & Mean length of constant-symbol runs. & 26 & \feat{std_run} & Standard deviation of run lengths. \\
27 & \feat{longest_run} & Maximum constant-symbol run length. & 28 & \feat{longest_run_ratio} & \feat{longest_run} divided by \feat{parsed_len}. \\
29 & \feat{alternation_rate} & Fraction of adjacent parsed symbols that differ. & 30 & \feat{expected_alternation} & Target value $1-\sum_a p_a^2$. \\
31 & \feat{bigram_entropy} & Entropy of adjacent two-symbol blocks. & 32 & \feat{trigram_entropy} & Entropy of adjacent three-symbol blocks. \\
33 & \feat{compression_ratio} & zlib compressed parsed-string length divided by \feat{parsed_len}. & 34 & \feat{autocorr_lag1} & Autocorrelation of alphabet-indexed symbols at lag 1. \\
35 & \feat{autocorr_lag2} & Autocorrelation of alphabet-indexed symbols at lag 2. & 36 & \feat{autocorr_lag3} & Autocorrelation of alphabet-indexed symbols at lag 3. \\
37 & \feat{first_symbol_code} & Alphabet index of the first parsed symbol. & 38 & \feat{last_symbol_code} & Alphabet index of the last parsed symbol. \\
39 & \feat{count_0} & \symcount{0} & 40 & \feat{freq_0} & \symfreq{0} \\
41 & \feat{count_1} & \symcount{1} & 42 & \feat{freq_1} & \symfreq{1} \\
43 & \feat{count_2} & \symcount{2} & 44 & \feat{freq_2} & \symfreq{2} \\
45 & \feat{count_3} & \symcount{3} & 46 & \feat{freq_3} & \symfreq{3} \\
47 & \feat{count_4} & \symcount{4} & 48 & \feat{freq_4} & \symfreq{4} \\
49 & \feat{count_5} & \symcount{5} & 50 & \feat{freq_5} & \symfreq{5} \\
51 & \feat{count_6} & \symcount{6} & 52 & \feat{freq_6} & \symfreq{6} \\
53 & \feat{count_7} & \symcount{7} & 54 & \feat{freq_7} & \symfreq{7} \\
55 & \feat{count_8} & \symcount{8} & 56 & \feat{freq_8} & \symfreq{8} \\
57 & \feat{count_9} & \symcount{9} & 58 & \feat{freq_9} & \symfreq{9} \\
59 & \feat{trans_0_0} & \transfreq{0}{0} & 60 & \feat{trans_0_1} & \transfreq{0}{1} \\
61 & \feat{trans_1_0} & \transfreq{1}{0} & 62 & \feat{trans_1_1} & \transfreq{1}{1} \\
63 & \feat{alt0_mismatch} & Mean mismatch to the alternating pattern starting with symbol 0. & 64 & \feat{alt1_mismatch} & Mean mismatch to the alternating pattern starting with symbol 1. \\
65 & \path{first_half_symbol0_freq} & Symbol-0 frequency in the first half. & 66 & \path{second_half_symbol0_freq} & Symbol-0 frequency in the second half. \\
67 & \feat{half_freq_gap} & Absolute first-half vs second-half symbol-0 frequency gap. & 68 & \feat{quarter_freq_std} & Standard deviation of symbol-0 frequencies across quarters. \\
69 & \feat{trans_0_2} & \transfreq{0}{2} & 70 & \feat{trans_0_3} & \transfreq{0}{3} \\
71 & \feat{trans_0_4} & \transfreq{0}{4} & 72 & \feat{trans_0_5} & \transfreq{0}{5} \\
73 & \feat{trans_0_6} & \transfreq{0}{6} & 74 & \feat{trans_0_7} & \transfreq{0}{7} \\
75 & \feat{trans_0_8} & \transfreq{0}{8} & 76 & \feat{trans_0_9} & \transfreq{0}{9} \\
77 & \feat{trans_1_2} & \transfreq{1}{2} & 78 & \feat{trans_1_3} & \transfreq{1}{3} \\
79 & \feat{trans_1_4} & \transfreq{1}{4} & 80 & \feat{trans_1_5} & \transfreq{1}{5} \\
81 & \feat{trans_1_6} & \transfreq{1}{6} & 82 & \feat{trans_1_7} & \transfreq{1}{7} \\
83 & \feat{trans_1_8} & \transfreq{1}{8} & 84 & \feat{trans_1_9} & \transfreq{1}{9} \\
85 & \feat{trans_2_0} & \transfreq{2}{0} & 86 & \feat{trans_2_1} & \transfreq{2}{1} \\
87 & \feat{trans_2_2} & \transfreq{2}{2} & 88 & \feat{trans_2_3} & \transfreq{2}{3} \\
89 & \feat{trans_2_4} & \transfreq{2}{4} & 90 & \feat{trans_2_5} & \transfreq{2}{5} \\
91 & \feat{trans_2_6} & \transfreq{2}{6} & 92 & \feat{trans_2_7} & \transfreq{2}{7} \\
93 & \feat{trans_2_8} & \transfreq{2}{8} & 94 & \feat{trans_2_9} & \transfreq{2}{9} \\
95 & \feat{trans_3_0} & \transfreq{3}{0} & 96 & \feat{trans_3_1} & \transfreq{3}{1} \\
97 & \feat{trans_3_2} & \transfreq{3}{2} & 98 & \feat{trans_3_3} & \transfreq{3}{3} \\
99 & \feat{trans_3_4} & \transfreq{3}{4} & 100 & \feat{trans_3_5} & \transfreq{3}{5} \\
101 & \feat{trans_3_6} & \transfreq{3}{6} & 102 & \feat{trans_3_7} & \transfreq{3}{7} \\
103 & \feat{trans_3_8} & \transfreq{3}{8} & 104 & \feat{trans_3_9} & \transfreq{3}{9} \\
105 & \feat{trans_4_0} & \transfreq{4}{0} & 106 & \feat{trans_4_1} & \transfreq{4}{1} \\
107 & \feat{trans_4_2} & \transfreq{4}{2} & 108 & \feat{trans_4_3} & \transfreq{4}{3} \\
109 & \feat{trans_4_4} & \transfreq{4}{4} & 110 & \feat{trans_4_5} & \transfreq{4}{5} \\
111 & \feat{trans_4_6} & \transfreq{4}{6} & 112 & \feat{trans_4_7} & \transfreq{4}{7} \\
113 & \feat{trans_4_8} & \transfreq{4}{8} & 114 & \feat{trans_4_9} & \transfreq{4}{9} \\
115 & \feat{trans_5_0} & \transfreq{5}{0} & 116 & \feat{trans_5_1} & \transfreq{5}{1} \\
117 & \feat{trans_5_2} & \transfreq{5}{2} & 118 & \feat{trans_5_3} & \transfreq{5}{3} \\
119 & \feat{trans_5_4} & \transfreq{5}{4} & 120 & \feat{trans_5_5} & \transfreq{5}{5} \\
121 & \feat{trans_5_6} & \transfreq{5}{6} & 122 & \feat{trans_5_7} & \transfreq{5}{7} \\
123 & \feat{trans_5_8} & \transfreq{5}{8} & 124 & \feat{trans_5_9} & \transfreq{5}{9} \\
125 & \feat{trans_6_0} & \transfreq{6}{0} & 126 & \feat{trans_6_1} & \transfreq{6}{1} \\
127 & \feat{trans_6_2} & \transfreq{6}{2} & 128 & \feat{trans_6_3} & \transfreq{6}{3} \\
129 & \feat{trans_6_4} & \transfreq{6}{4} & 130 & \feat{trans_6_5} & \transfreq{6}{5} \\
131 & \feat{trans_6_6} & \transfreq{6}{6} & 132 & \feat{trans_6_7} & \transfreq{6}{7} \\
133 & \feat{trans_6_8} & \transfreq{6}{8} & 134 & \feat{trans_6_9} & \transfreq{6}{9} \\
135 & \feat{trans_7_0} & \transfreq{7}{0} & 136 & \feat{trans_7_1} & \transfreq{7}{1} \\
137 & \feat{trans_7_2} & \transfreq{7}{2} & 138 & \feat{trans_7_3} & \transfreq{7}{3} \\
139 & \feat{trans_7_4} & \transfreq{7}{4} & 140 & \feat{trans_7_5} & \transfreq{7}{5} \\
141 & \feat{trans_7_6} & \transfreq{7}{6} & 142 & \feat{trans_7_7} & \transfreq{7}{7} \\
143 & \feat{trans_7_8} & \transfreq{7}{8} & 144 & \feat{trans_7_9} & \transfreq{7}{9} \\
145 & \feat{trans_8_0} & \transfreq{8}{0} & 146 & \feat{trans_8_1} & \transfreq{8}{1} \\
147 & \feat{trans_8_2} & \transfreq{8}{2} & 148 & \feat{trans_8_3} & \transfreq{8}{3} \\
149 & \feat{trans_8_4} & \transfreq{8}{4} & 150 & \feat{trans_8_5} & \transfreq{8}{5} \\
151 & \feat{trans_8_6} & \transfreq{8}{6} & 152 & \feat{trans_8_7} & \transfreq{8}{7} \\
153 & \feat{trans_8_8} & \transfreq{8}{8} & 154 & \feat{trans_8_9} & \transfreq{8}{9} \\
155 & \feat{trans_9_0} & \transfreq{9}{0} & 156 & \feat{trans_9_1} & \transfreq{9}{1} \\
157 & \feat{trans_9_2} & \transfreq{9}{2} & 158 & \feat{trans_9_3} & \transfreq{9}{3} \\
159 & \feat{trans_9_4} & \transfreq{9}{4} & 160 & \feat{trans_9_5} & \transfreq{9}{5} \\
161 & \feat{trans_9_6} & \transfreq{9}{6} & 162 & \feat{trans_9_7} & \transfreq{9}{7} \\
163 & \feat{trans_9_8} & \transfreq{9}{8} & 164 & \feat{trans_9_9} & \transfreq{9}{9} \\
165 & \feat{number_found} & Indicator that the number parser found an integer. & 166 & \feat{number_value} & First integer recovered by the number parser. \\
167 & \feat{number_log10} & $\log_{10}$ of \feat{number_value}, when positive. & 168 & \feat{number_num_digits} & Decimal digit count of \feat{number_value}. \\
169 & \feat{number_parity} & \feat{number_value} modulo 2. & 170 & \feat{number_mod10} & \feat{number_value} modulo 10. \\
171 & \feat{number_mod100} & \feat{number_value} modulo 100. & 172 & \feat{count_H} & \symcount{H} \\
173 & \feat{freq_H} & \symfreq{H} & 174 & \feat{count_T} & \symcount{T} \\
175 & \feat{freq_T} & \symfreq{T} & 176 & \feat{trans_H_H} & \transfreq{H}{H} \\
177 & \feat{trans_H_T} & \transfreq{H}{T} & 178 & \feat{trans_T_H} & \transfreq{T}{H} \\
179 & \feat{trans_T_T} & \transfreq{T}{T} & & & \\
\end{longtable}
\endgroup
\clearpage
\twocolumn


\section{Score Information}
\label{app:score}
Definition~\ref{def:I} distinguishes the oracle exponent $I^\star$, the score exponent $I^{\mathrm{sc}}$, and the plug-in estimate $\widehat{I}_{\mathrm{auc}}$. We make $I^{\mathrm{sc}}$ precise and bound it.

\begin{definition}[Score exponent]
\label{def:Isc}
Let $\mu_P,\mu_R$ be the push-forward laws of the score $s(y)=-\log\hat{\mathbb P}(P\mid y)$ under $y\sim P$ and $y\sim R$. The \emph{score exponent} is the Chernoff information of these scalar laws, $I^{\mathrm{sc}}(c;P,R)=-\min_{0\le\lambda\le1}\log\int \mu_P(s)^{1-\lambda}\mu_R(s)^{\lambda}\,ds$.
\end{definition}

\begin{proposition}[Score compression]
\label{prop:Isc}
Let $s=s(y)$ be a deterministic measurable score with push-forward laws $\mu_P=s_\#P$ and $\mu_R=s_\#R$. Then $I^{\mathrm{sc}}\le I^\star$. If $P\neq R$, equality holds iff $s$ is sufficient for $\{P,R\}$, equivalently iff the likelihood ratio $P/R$ is constant on each level set of $s$.
\end{proposition}

\noindent The equality condition says exactly when the scalar score loses no response-level evidence. Calibrated log-odds that are affine and strictly monotone in $\log(P/R)$ attain the oracle rate; miscalibration or a missing feature that breaks this measurability gives $I^{\mathrm{sc}}<I^\star$.

\noindent The operational $\widehat{I}_{\mathrm{auc}}$ (slope of $\log(1-\mathrm{AUROC})$) is the decay rate of a \emph{two-sample rank} statistic on $2m$ observations; it is \emph{not} in general bounded by $I^{\mathrm{sc}}$ and can exceed it (e.g.\ $I^{\mathrm{auc}}=2\,I^{\mathrm{sc}}$ for $\mathrm{Ber}(.25)/\mathrm{Ber}(.75)$ scores). What Prop.~\ref{prop:Isc} bounds is the one-sample \emph{mean-threshold} exponent, $I^{\mathrm{mean}}\le I^{\mathrm{sc}}\le I^\star$; we do not claim $\widehat{I}_{\mathrm{auc}}=I^{\mathrm{sc}}$, and its exponential law is the empirical regularity of Section~\ref{sec:exp}.

\section{Temperature Retunes}
\label{app:tempmag}
This appendix derives the \emph{magnitude} of a temperature change's effect on a visible-string audit and shows why \IRIS{}'s insensitivity to it is the intended behavior. All quantities here are \emph{per single draw} (per position); the per-response budget follows through Prop.~\ref{prop:len} and is non-additive in length. Throughout we hold the served weights fixed, so the per-position logit vector $z$ is the \emph{same} at both temperatures---this is exactly what separates a temperature retune from a model change.

\begin{proposition}[Second-order retunes]
\label{prop:tempmag}
Hold the served logits $z\in\mathbb{R}^n$ ($n<\infty$) fixed and temper at inverse temperature $\beta=1/T$, $p^\beta_i\propto e^{\beta z_i}$. With $A(\beta)=\log\sum_i e^{\beta z_i}$ and $V=A''(\beta_0)>0$, a small gap $\Delta\beta=\beta_1-\beta_0$ gives
\[
\begin{aligned}
  I^\star&=\tfrac18(\Delta\beta)^2V+O(\Delta\beta^3),\\
  \mathrm{KL}(p^{\beta_0}\|p^{\beta_1})&=\tfrac12(\Delta\beta)^2V+O(\Delta\beta^3),\\
  \chi^2(p^{\beta_1}\|p^{\beta_0})&=(\Delta\beta)^2V+O(\Delta\beta^3).
\end{aligned}
\]
Thus a pure temperature retune has only second-order separation in $\Delta\beta$; the companion identity is $\mathrm{Var}_{p^T}(z)=T^3\,dH/dT\ge0$.
\end{proposition}

\noindent Resolving a temperature retune therefore needs $m=\Theta((\Delta\beta)^{-2})$ draws, while a base-model change perturbs $z$ itself and can supply first-order evidence. The expansion is an interior statement; at the greedy boundary $V\to0$, and support collapse is the relevant observable (Cor.~\ref{cor:greedy}). The derivation, summarized next, is the basis for Eqs.~\eqref{eq:tempmag}--\eqref{eq:heatcap}.

\paragraph{Tempering family.}
Decoding at temperature $T$ draws token $i$ with $p^T_i=\mathrm{softmax}(z/T)_i\propto e^{z_i/T}$. In the inverse temperature $\beta=1/T$ this is the exponential family $p^\beta_i\propto e^{\beta z_i}$ with natural parameter $\beta$, sufficient statistic the logit (``energy'') $z$, and log-partition $A(\beta)=\log\sum_i e^{\beta z_i}$. Its cumulants are $A'(\beta)=\mathbb{E}_{p^\beta}[z]$ and $A''(\beta)=\mathrm{Var}_{p^\beta}(z)$, the Fisher information of the family. A temperature change moves $\beta$ \emph{along} this curve; a model change moves the vector $z$ itself, generically \emph{off} it (Fig.~\ref{fig:tempgeom}).

\begin{figure}[t]
  \centering
  \includegraphics[width=\columnwidth]{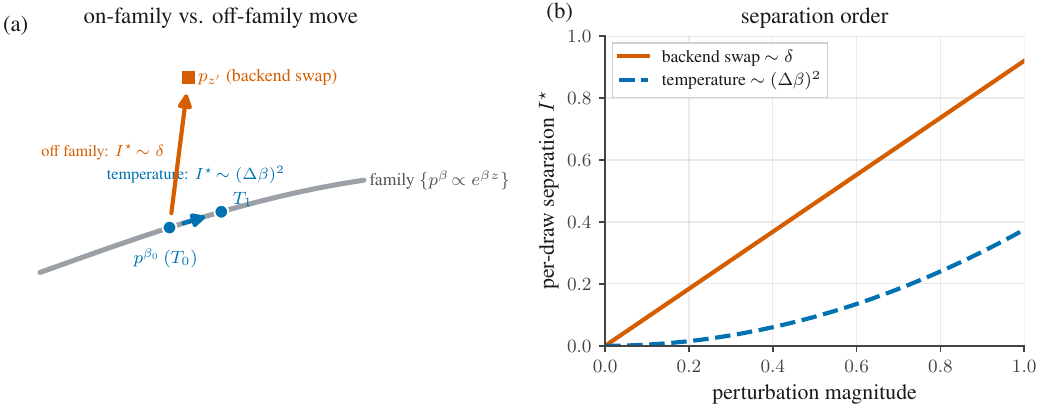}
  \caption{Why a temperature retune is second order while a backend substitution is first order. \textbf{(a)} With the logits $z$ fixed, the tempered laws $p^\beta\propto e^{\beta z}$ trace a one-parameter curve; a temperature change moves $\beta$ \emph{along} it (rank-one, on-family), whereas a model or backend change perturbs $z$ and leaves the curve. \textbf{(b)} An on-family step separates only at second order, $I^\star\approx\tfrac18(\Delta\beta)^2V$, flat at the origin, so resolving it costs $m=\Theta((\Delta\beta)^{-2})$ draws; an off-family perturbation can separate at first order in its magnitude. \IRIS{} therefore waives a temperature retune and fires on a separating backend substitution.}
  \label{fig:tempgeom}
\end{figure}

\paragraph{Second-order gap.}
For a small gap $\Delta\beta=\beta_1-\beta_0$ every standard divergence collapses to the same Fisher quadratic (directional asymmetry enters only at $O(\Delta\beta^3)$):
\[
  \mathrm{KL}\approx\tfrac12(\Delta\beta)^2 V,\qquad
  \chi^2\approx(\Delta\beta)^2 V,\qquad
  I^\star\approx\tfrac18(\Delta\beta)^2 V,
\]
with $V:=\mathrm{Var}_{p^{\beta_0}}(z)$ and the Chernoff information $I^\star$ the Bhattacharyya value at the optimal tilt $\lambda^\star=\tfrac12$. (The $\chi^2$ follows from the exp-family identity $\chi^2(p_{\beta_1}\Vert p_{\beta_0})=e^{A(2\beta_1-\beta_0)-2A(\beta_1)+A(\beta_0)}-1$, whose exponent has zero value, zero first derivative, and second derivative $2A''(\beta_0)$ at $\beta_1{=}\beta_0$.) Converting to the temperature gap, $\Delta\beta=1/T_1-1/T_0\approx-\Delta T/T_0^2$, gives
\begin{equation}
  I^\star(T_0,T_1)=\tfrac18(\Delta\beta)^2 V+O(\Delta\beta^3)
  =\frac{(\Delta T)^2}{8\,T_0}\,\frac{dH}{dT}\Big|_{T_0}+\cdots,
  \label{eq:tempmag}
\end{equation}
using the identity below. Temperature thus separates only at \emph{second} order in the gap, $I^\star=\Theta((\Delta\beta)^2)$, so an audit needs $m=\Theta((\Delta\beta)^{-2})$ independent draws to resolve a retune.

\paragraph{Heat capacity.}
The tempered logit fluctuation equals an entropy slope. With $H(T)$ the Shannon entropy of $p^T$, the relation $H=A(\beta)-\beta A'(\beta)$ gives $dH/d\beta=-\beta A''(\beta)$, and $\beta=1/T$ ($d\beta/dT=-1/T^2$) yields
\begin{equation}
  \mathrm{Var}_{p^T}(z)=T^3\,\frac{dH}{dT}\ \ge 0,
  \label{eq:heatcap}
\end{equation}
nonnegative because output entropy rises with temperature; it is the ``heat capacity'' of the next-token distribution. Crucially $V$ is the variance of $z$ \emph{under the tempered law} $p^T$, not the spread of the raw logit vector: the softmax de-weights the large-magnitude tail, so $V$ is bounded by the entropy slope and is $O(1)$ nat$^2$ near $T{\approx}1$ (for plausible next-token profiles $V\!\approx\!1$--$5$), \emph{not} the $O(10$--$100)$ a raw logit range might suggest.

\paragraph{Numeric.}
Between $T_0{=}1$ and $T_1{=}1.5$, $\Delta\beta=-\tfrac13$ and $I^\star\approx V/72\approx0.01$--$0.07$ nat per draw. A Chernoff-to-detectability heuristic ($d'\!\approx\!2\sqrt{I^\star}$, $\mathrm{AUROC}\!\approx\!\Phi(d'/2)$) maps this to single-response AUROC $\approx0.57$--$0.59$, matching the measured $\approx0.58$ between adjacent operating temperatures (App.~\ref{app:knobs}); the per-position rate does \emph{not} multiply by response length $L$, since AR positions are correlated and share the backbone $z$ (Prop.~\ref{prop:len}).

\paragraph{Boundary case.}
Eq.~\eqref{eq:tempmag} governs the \emph{interior} ($T\gtrsim0.5$, non-degenerate $p^T$). At the greedy boundary $T\!\to\!0$, $\beta\!\to\!\infty$ and $V\!\to\!0$ (point mass), so the per-draw rate vanishes (Cor.~\ref{cor:greedy})---yet $T{=}0$ vs.\ $T{>}0$ is easily distinguishable (AUROC $\approx0.99$) by a \emph{different}, single-draw observable: the support/entropy collapse (greedy repeats one string). The two mechanisms act at different sample sizes and do not conflict.

\paragraph{Sampler retunes.}
A dilution audit asks whether the \emph{backend model} changed, not whether a scalar sampler knob moved. Tempering is a rank-one, on-family move and is forced to be second order; a model change is instead an off-family perturbation of $z$. Thus \IRIS{} can waive a temperature retune while catching a separating substitute. The operational caveat is that pairs with too little visible-string separation require substantially larger budgets and fall outside the deployed margin-check regime (Thm.~\ref{thm:mix}(c), App.~\ref{app:tail}).

\section{Tail Phase Boundary}
\label{app:tail}
Theorem~\ref{thm:mix}(c) shows the regime is set by whether a separating tail exists. We make the interpolation between $1/\epsilon$ and $1/\epsilon^2$ explicit through a \emph{tail-separation function}.

\begin{definition}[Tail separation]
For reference distribution $P$ with score CDF $F_P$, the tail-separation function of a substitute distribution $R$ is $q_\alpha(R)=\mathbb P_{y\sim R}\!\big(s(y)>F_P^{-1}(1-\alpha)\big)$, i.e.\ the tell rate at the per-query false-positive level $\alpha$.
\end{definition}

\begin{proposition}[Tail budget]
\label{prop:phase}
Suppose $F_P$ is atomless and $q_\alpha\asymp c\,\alpha^{\kappa}$ as $\alpha\to0$ for a tail exponent $\kappa\in[0,1)$ and $c>0$, and the auditor controls type-I at level $\alpha$ via the union bound (per-query level $\alpha/m$). Then, with $\alpha,\delta,c,\kappa$ fixed and $\epsilon\to0$, the detection budget is
\[\begin{aligned}
  m^\star&=\Theta\!\Big(\epsilon^{-1/(1-\kappa)}\Big),\\
  m^\star&\asymp\Big(\tfrac{\ln(1/\delta)}{c\,\alpha^{\kappa}}\Big)^{1/(1-\kappa)}\epsilon^{-1/(1-\kappa)}.
\end{aligned}\]
Thus $\kappa=0$ gives the $1/\epsilon$ law, $\kappa=\tfrac12$ gives $\epsilon^{-2}$ scaling, and $\kappa\uparrow1$ is the no-tail wall.
\end{proposition}

\noindent The $\kappa=\tfrac12$ exponent matches the aggregate mean-shift budget of Thm.~\ref{thm:mix}(c), although the tests and constants are different.

We measure $q_\alpha$ directly on the reference endpoint \texttt{qwen3:8b}@$T{=}1$ over a separation ladder (probe $c_{2,16}$), and re-measure on \emph{scale-invariant} (length/format-ablated, $179{\to}144$) features. The tail exponent $\kappa$ (the log--log slope of $q_\alpha$ in $\alpha$) orders the ladder monotonically and the ordering survives ablation (Table~\ref{tab:tailsep}): distinct/same-family substitutes have $\kappa<\tfrac12$ (near the $1/\epsilon$ law, exact only at $\kappa{=}0$), temperature twins sit near the $\kappa{=}\tfrac12$ crossover that Prop.~\ref{prop:phase} maps to $1/\epsilon^2$, and near-greedy twins have $\kappa\to1$ (no tail; $m^\star{=}\epsilon^{-1/(1-\kappa)}$ diverges, an effective wall). We report the $q$-controlled continuum and the exponent $\kappa$, not a fitted $\epsilon^{-2}$ slope.

\paragraph{When $\kappa{=}0$.}
Exact $1/\epsilon$ ($\kappa{=}0$, $q_\alpha\to q_0>0$ as $\alpha\to0$) requires the substitute to retain positive mass where the reference's score is in its extreme upper tail---in the population, a region $P$ makes vanishingly unlikely yet $R$ does not. For samplers over a shared vocabulary at finite temperature, where every visible string has positive probability under both, the population $q_\alpha$ generically vanishes as $\alpha\to0$ ($\kappa>0$): a strict $\kappa{=}0$ would need a genuine support difference or a singular component, not merely two distinct models. An empirical $q_0>0$ estimated from a finite reference can instead reflect empirical zero-probability regions a high-capacity classifier mistakes for true support separation. The discrete score (atoms from short strings and from the forest's \texttt{predict\_proba}) further makes the atomless hypothesis of Prop.~\ref{prop:phase} idealizing: above the top atom $q$ can drop to $0$, and with atom-randomization a bounded likelihood ratio gives $q_\alpha=O(\alpha)$, i.e.\ $\kappa{=}1$. We therefore read the measured $\kappa\!<\!\tfrac12$ as a \emph{finite-$\epsilon$, finite-sample} statement and test its stability: re-estimating $\kappa$ as the reference calibration set grows from $n{=}50$ to $n{=}3000$ (App.~\ref{app:fixedfpr}) shows whether $\kappa$ settles at a positive value or moves toward $1$. The $1/\epsilon$ regime is thus reported as the measured behavior over the audited $\epsilon$ range, never an asymptotic law.

\begin{table}[!ht]
\centering\footnotesize
\caption{Tail-separation ladder on $c_{2,16}$ (ref \texttt{qwen3:8b}@$T1$): plateau
$q_0{=}q_{\alpha=.01}$ and tail exponent $\kappa$, full $\to$ scale-invariant (length-ablated) features. The
$\kappa$-ordering survives ablation; exact $1/\epsilon$ holds only at $\kappa{=}0$ and $\kappa{=}\tfrac12$ is the
$1/\epsilon^2$ crossover (Prop.~\ref{prop:phase}).}
\label{tab:tailsep}
\setlength{\tabcolsep}{2pt}
\begin{tabular}{@{}lccc@{}}
\toprule
substitute & $q_0$ full$\to$abl & $\kappa$ full$\to$abl & regime \\
\midrule
\texttt{qwen3:0.6b} & $.58\!\to\!.33$ & $.14\!\to\!.31$ & near-$1/\epsilon$ ($\kappa{<}\tfrac12$) \\
\texttt{qwen3:32b}  & $.60\!\to\!.45$ & $.13\!\to\!.22$ & near-$1/\epsilon$ ($\kappa{<}\tfrac12$) \\
\texttt{qwen3:14b}  & $.87\!\to\!.56$ & $.03\!\to\!.18$ & near-$1/\epsilon$ ($\kappa{<}\tfrac12$) \\
\texttt{8b}@$T2$    & $.08\!\to\!.07$ & $.48\!\to\!.60$ & near $1/\epsilon^2$ ($\kappa{\approx}\tfrac12$) \\
\texttt{8b}@$T1.5$  & $.08\!\to\!.04$ & $.22\!\to\!.57$ & interpolating \\
\texttt{8b}@$T0.5$  & $.01\!\to\!.01$ & $.81\!\to\!.97$ & wall ($\kappa{\to}1$) \\
\texttt{8b}@$T0$    & $.00\!\to\!.00$ & $-$ & wall ($\kappa{\to}1$) \\
\bottomrule
\end{tabular}
\end{table}

The same continuum read off a second probe (Fig.~\ref{fig:qcontinuum}): the plateau tell rate $q$ slides smoothly from same-model temperature twins (at the floor $\alpha$, undetectable) through same-family sizes to a distinct cross-family model ($q{=}1$). When the null rate is pushed to $\alpha_1\!\to\!0$ (a separating tail), the budget is the any-tell $1/\epsilon$ law $m^\star=\ln(1/\delta)/(\epsilon q)$ with constant set by $q$. At a \emph{fixed} non-zero null $p_0$, however, this is a \emph{different} test: the any-tell formula ignores the honest tells arriving at rate $p_0$ and would not control type-I, so the deployed two-proportion test instead has local signal $\epsilon(q-p_0)$ and budget $m^\star=\Theta\!\big(p_0(1-p_0)/(\epsilon^2(q-p_0)^2)\big)$ (Prop.~\ref{prop:e2e} proof). The $\epsilon^{-2}$ wall is thus not the $q\!\to\!\alpha$ limit of the any-tell law (which would give the finite $\ln(1/\delta)/(\epsilon\alpha)$); it is the distinct fixed-null regime that happens to share the $\epsilon^{-2}$ \emph{exponent}, the two budgets being endpoints of the measured-$\kappa$ continuum $m^\star{=}\Theta(\epsilon^{-1/(1-\kappa)})$ (App.~\ref{app:tail}).

\begin{figure}[t]
  \centering
  \includegraphics[width=0.95\columnwidth]{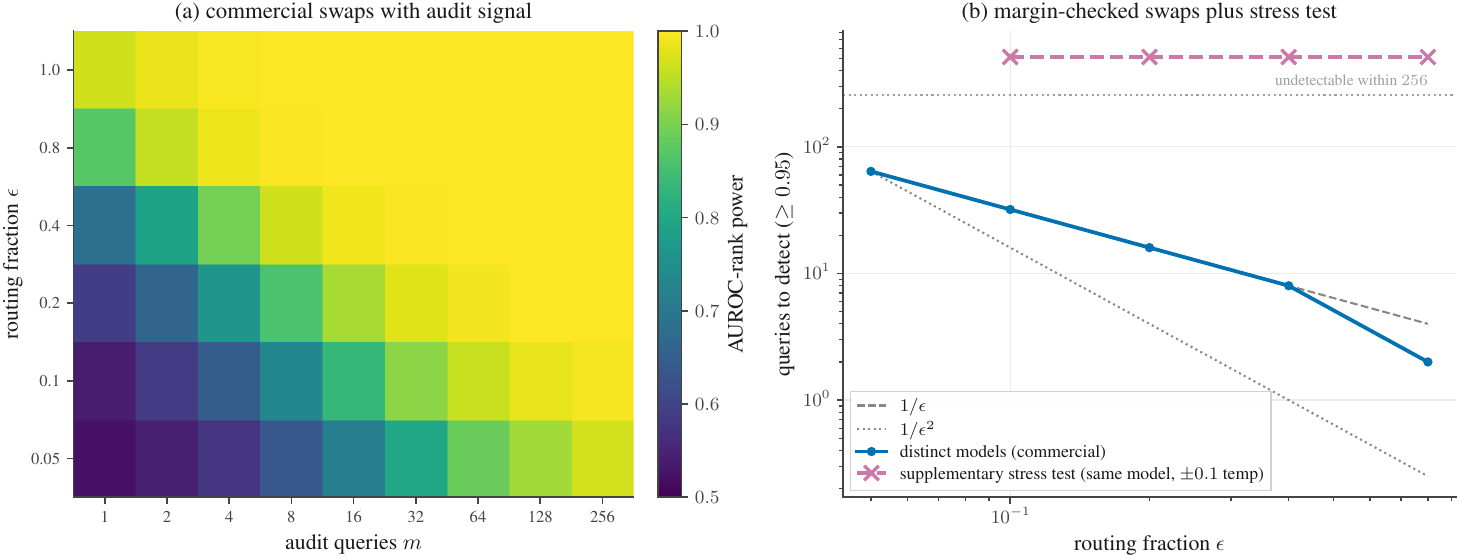}
  \caption{Dilution (referenced from Section~\ref{sec:exp}). (a)~AUROC-rank diagnostic power over $(\epsilon,m)$
  for commercial base-model dilutions with enough enrollment signal. (b)~Queries to detect vs.\ $\epsilon$: dilutions passing the margin check track the near-$1/\epsilon$ wall,
  while the \texttt{qwen3:8b} $T{=}1.0$-vs-$T{=}0.9$ same-model supplementary stress test ($c_{2,16}$) is undetectable within
  $256$ queries (per pair, not median). The pre-fixed-FPR validation is in App.~\ref{app:fixedfpr}.}
  \label{fig:mixrouting}
\end{figure}

\begin{figure}[t]
  \centering
  \includegraphics[width=\columnwidth]{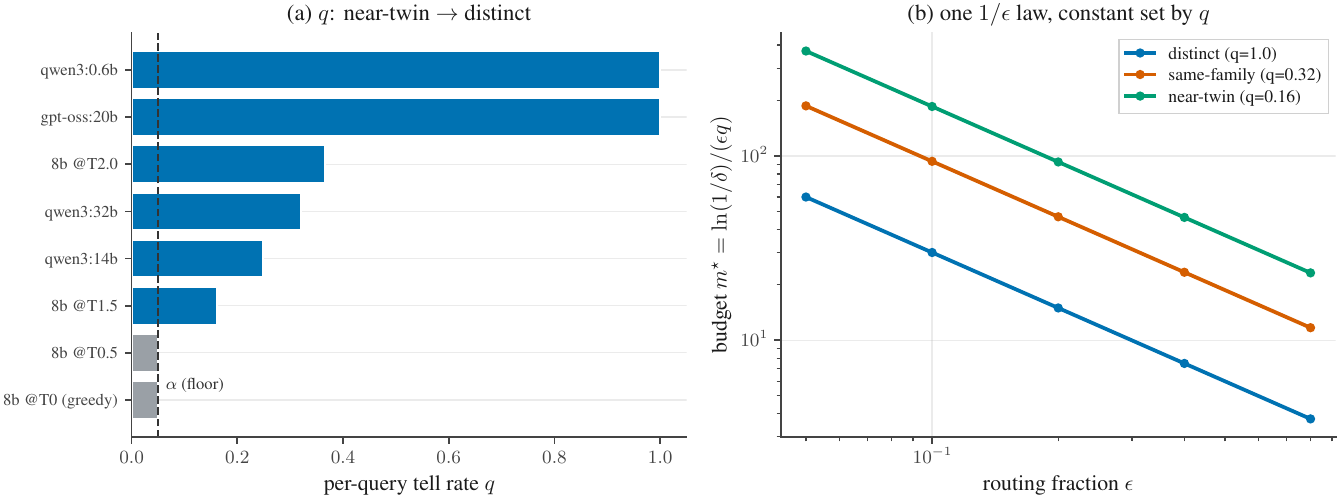}
  \caption{The tell-rate continuum on $c_{100,1}$ (ref \texttt{qwen3:8b}@$T{=}1$).
  Here $q$ is the per-query tell rate $\mathbb P_{y\sim R}[s(y){>}\tau]$ and $\alpha$ the per-query false-positive floor.
  (a)~Per-query tell rate $q$ per substitute, ascending; bars below the false-positive floor $\alpha$ (greedy
  and near-greedy twins) are undetectable. (b)~The unified budget $m^\star=\ln(1/\delta)/(\epsilon q)$ for three
  representatives---one $1/\epsilon$ law on log--log axes, vertically offset by $q$. Distinct models sit lowest
  (cheapest), low-separation variants highest.}
  \label{fig:qcontinuum}
\end{figure}

\subsection{Tail Stability}
\label{app:betastab}
Table~\ref{tab:tailsep} measures $\kappa$ on one reference (\texttt{qwen3:8b}@$T1$) and one probe. To test whether the $1/\epsilon$ regime is a property of that single calibration point or holds across the library, we recompute the tail exponent $\kappa$ (the log--log slope of $q_\alpha$ in $\alpha$, as in Table~\ref{tab:tailsep}) for \emph{every} ordered (reference, substitute) pair in the $17$- and $45$-endpoint gateways, on content-only features, averaging over four out-of-fold classifier/split seeds. Table~\ref{tab:betastab} reports the distribution over pairs passing the enrollment margin check (plateau $q_0\!\ge\!0.15$). The picture is consistent and quantifies the headline's scope: the $1/\epsilon$ regime ($\kappa\!\approx\!0$) is the \emph{typical} behavior once the margin check passes---the implied budget exponent $1/(1-\kappa)$ has median $1.02$--$1.10$, and $62$--$85\%$ of such pairs are statistically consistent with $\kappa{=}0$. A right tail remains ($2$--$9\%$ with $\kappa\!\ge\!0.4$), so \IRIS{} budgets against the \emph{per-pair} $\kappa$ read off calibration rather than assuming a universal $\kappa{=}0$.

\begin{table*}[t]
  \centering\footnotesize\setlength{\tabcolsep}{3pt}
  \caption{Tail-exponent $\kappa$ distribution over ordered (reference, substitute) pairs passing the margin check
  of the $17$- and $45$-endpoint gateways (content-only features, four seeds). ``cons.\ $\kappa{=}0$'' is the
  fraction of margin-qualified pairs statistically consistent with $\kappa{=}0$ (the exact $1/\epsilon$ law);
  ``$1/(1{-}\kappa)$'' is the implied budget exponent; ``twin'' the fraction with $\kappa\!\ge\!0.4$.}
  \label{tab:betastab}
  \begin{tabular}{@{}llccccc@{}}
    \toprule
    pool & probe & sep. & $\kappa$ med.\ [25,75] & cons.\ $\kappa{=}0$ & $1/(1{-}\kappa)$ med. & twin \\
    \midrule
    $17$ & $c_{10,8}$    & $248/272$  & $.02\,[.0,.17]$  & $.77$ & $1.02$ & $.06$ \\
    $17$ & $c_{2,16}$     & $164/272$  & $.08\,[.01,.18]$ & $.74$ & $1.09$ & $.02$ \\
    $17$ & $c_{100,1}$ & $176/272$  & $.02\,[.0,.08]$  & $.85$ & $1.02$ & $.02$ \\
    $45$ & $c_{10,8}$    & $1544/1980$& $.09\,[.01,.26]$ & $.62$ & $1.10$ & $.09$ \\
    \bottomrule
  \end{tabular}
\end{table*}

\section{Artifact Manifest and Infrastructure}
\label{app:artifact}

We release, as a single reproducibility bundle: (i) all $147{,}070$ frozen response records as JSON Lines, covering the controlled Ollama pools, the $17$- and $45$-endpoint gateway pools, and the reported robustness and negative-control analyses; (ii) the exact probe prompts and deterministic parser; (iii) all $179$ visible-string feature definitions with the content-only and length-ablated subsets; (iv) the deterministic stratified-split procedures and fixed seeds used by the analysis code; and (v) the code for response collection, feature extraction, the \IRIS{} estimate-then-budget loop, dilution simulation, at-scale and live dilution auditing, and figure generation.\ifdefined\IRISArxiv\ The bundle is available at \url{https://github.com/Photen/IRIS-audit}.\fi

Local models ran under Ollama~0.12.5 on Ubuntu~22.04.4 LTS with an Intel Xeon Platinum~8358 CPU, $2\,$TiB system memory, and six NVIDIA~A40 GPUs (46--49\,GiB reported per GPU; driver~550.54.15); the $17$-API set was queried through a commercial OpenAI-compatible gateway. Analysis used scikit-learn random forests ($300$ trees, \texttt{max\_features} $=\sqrt{\cdot}$, balanced classes) under a fixed integer seed ($20260617$); reported curves average $40$ stratified splits with $95\%$ bootstrap intervals. Classifier posteriors $\hat P(\cdot\mid y)$ are random-forest \texttt{predict\_proba} outputs clipped to $[10^{-12},1]$ and used as a score, not as calibrated probabilities; every operating threshold and budget is set empirically on held-out calibration data, so posterior calibration is not assumed. Software: Python~3.12, scikit-learn~1.5.1, NumPy~1.26.4, and SciPy~1.13.1. No weights, logits, or token ranks were accessed.

\section{Proofs}
\label{app:proofs}
This section collects a rigorous, self-contained proof of every formal claim in the paper---the propositions, the theorem, and the corollaries of Section~\ref{sec:theory}, together with the two analytic results stated in the appendix (the temperature second-order law, App.~\ref{app:tempmag}, and the multi-diluent identifiability claim, App.~\ref{app:multi}). Each subsection restates what is proved, lists the standing assumptions, and gives the argument.

\subsection{Proof of Prop.~\ref{prop:exp} (mean score)}
\paragraph{Restatement.} Fix a probe $c$ and let $P,R$ be the reference and substitute response laws. The deployed claim is: if the bounded score $s(y)=-\log\hat{\mathbb P}(P\mid y)$ has $\mathbb{E}_{P}[s]\neq\mathbb{E}_{R}[s]$, then $S_m=\frac1m\sum_i s(y_i)$ has a fixed-threshold error $e^{-m\,I^{\mathrm{mean}}(1+o(1))}$ with $I^{\mathrm{mean}}>0$ and $I^{\mathrm{mean}}\le I^{\mathrm{sc}}\le I^{\star}$. The comparison facts used around the proposition are: \emph{(1)} the full-string Bayes test has error $e^{-I^{\star}m(1+o(1))}$ and $I^\star>0\iff P\neq R$; \emph{(2)} $K$-endpoint separation is guaranteed by $m\gtrsim\log(K/\delta)/I^{\star}_{\min}$.

\medskip
Throughout, queries are independent, so under each fixed hypothesis $y_1,\dots,y_m$ are i.i.d. Write $\psi(\lambda)=\log\sum_y P(y)^{1-\lambda}R(y)^{\lambda}$ and $\lambda^\star=\arg\min_{[0,1]}\psi$, with $I^{\star}=-\psi(\lambda^\star)<\infty$.

\paragraph{Part (1).}
\emph{Achievability.} For any $\lambda\in[0,1]$ the likelihood-ratio (Bayes) test between $P^{\otimes m}$ and $R^{\otimes m}$ with prior $(\pi_P,\pi_R)$ obeys
\[\begin{aligned}
  P_e(m)&\le\sum_{y_{1:m}}\bigl(\pi_P P^{\otimes m}\bigr)^{1-\lambda}\bigl(\pi_R R^{\otimes m}\bigr)^{\lambda}\\
  &\le\Bigl(\sum_y P(y)^{1-\lambda}R(y)^{\lambda}\Bigr)^{m}=e^{m\psi(\lambda)},
\end{aligned}\]
using $\pi_P^{1-\lambda}\pi_R^{\lambda}\le1$ and tensorisation over the $m$ independent coordinates. At $\lambda=\lambda^\star$,
\begin{equation}\label{eq:exp-achiev}
  P_e(m)\le e^{-mI^{\star}}\qquad(\forall m),
\end{equation}
a genuine finite-$m$ bound (the Bhattacharyya statement) and the $\le$ half of the asymptotic law.

\emph{Converse.} Let $\ell(y)=\log\frac{R(y)}{P(y)}$ be the log-likelihood ratio. Under the regularity hypothesis that the cumulant generating function $\Lambda_P(\lambda)=\log\mathbb{E}_{P}[e^{\lambda\ell}]=\psi(\lambda)$ is finite on an open interval containing $[0,1]$ (automatic when the per-probe alphabet is finite, $|\mathrm{supp}|=n<\infty$, the regime of all single-draw probes) and $\ell$ is non-degenerate, the Chernoff-information theorem (method of types in the finite-alphabet case, \citet[Thm.~11.9.1]{cover2006elements}) gives $P_e(m)\ge e^{-mI^{\star}(1+o(1))}$. With \eqref{eq:exp-achiev}, $P_e(m)=e^{-I^{\star}m(1+o(1))}$.

\emph{Positivity.} By H\"older, $\sum_yP^{1-\lambda}R^{\lambda}\le(\sum_yP)^{1-\lambda}(\sum_yR)^{\lambda}=1$, so $\psi\le0$ and $I^{\star}\ge0$, with $\psi\equiv0$ when $P=R$. If $P\neq R$ then $\ell$ is non-degenerate under $P$, so $\psi$ (its cumulant generating function, CGF) is strictly convex; together with $\psi(0)=\psi(1)=0$ this forces $\psi(\lambda)<0$ on $(0,1)$, hence $I^{\star}>0$. The argument needs no common support: on atoms with $P(y)=0<R(y)$ the summand vanishes for $\lambda\in(0,1)$, only deepening the strict inequality. Thus $I^{\star}>0\iff P\neq R$. \hfill$\square$

\paragraph{Part (2).}
Let $\{M_i\}_{i=1}^K$ have laws $\{P_i\}$ and prior $\pi_i\ge\pi_{\min}>0$. The MAP classifier errs only if some wrong $j$ has joint posterior at least that of the true $i$; hence
\[\begin{aligned}
  \mathbb P(\text{err})&\le\sum_{i}\pi_i\sum_{j\neq i}\mathbb P_{P_i^{\otimes m}}\!\Bigl(\pi_jP_j^{\otimes m}\ge\pi_iP_i^{\otimes m}\Bigr)\\
  &\le\sum_{i\neq j}\Bigl(\sum_y P_i^{1-\lambda^\star_{ij}}P_j^{\lambda^\star_{ij}}\Bigr)^{m},
\end{aligned}\]
where each inner term is the Bhattacharyya/Chernoff bound \eqref{eq:exp-achiev} for the ordered pair $(i,j)$ (the prior ratio $\pi_j/\pi_i$ is absorbed by the same $\pi^{1-\lambda}\pi^{\lambda}\le1$ step). Bounding each by $e^{-mI^{\star}_{\min}}$ with $I^{\star}_{\min}=\min_{i\neq j}I^{\star}(c;M_i,M_j)>0$ (all endpoints distinct) gives
\[
  \mathbb P(\text{err})\le K(K-1)\,e^{-mI^{\star}_{\min}}\le K^2e^{-mI^{\star}_{\min}}.
\]
Requiring the right side $\le\delta$ holds once $m\ge(2\log K+\log(1/\delta))/I^{\star}_{\min}$, i.e. $m\gtrsim\log(K/\delta)/I^{\star}_{\min}$. This is a union bound layered on top of the binary exponent of \citet[Thm.~11.9.1]{cover2006elements}, not that theorem alone. \hfill$\square$

\paragraph{Part (3).}
\emph{Regularity.} With posteriors clipped to $[\epsilon_0,1-\epsilon_0]$, $s\in[-\log(1-\epsilon_0),\log(1/\epsilon_0)]$ is bounded, so the CGFs $\Lambda_P(t)=\log\mathbb{E}_{P}[e^{ts}]$ and $\Lambda_R(t)=\log\mathbb{E}_{R}[e^{ts}]$ are finite for all $t\in\mathbb{R}$, in particular near $0$, which is the hypothesis of Cram\'er's theorem. (Without clipping, $\mathbb{E}_{P}[e^{ts}]=\mathbb{E}_{P}[\hat{\mathbb P}(P\mid y)^{-t}]$ may diverge for every $t>0$ and the law can fail, so boundedness is the regularity condition we assume.)

\emph{Exponential law of $S_m$.} Let $\mu_P=\mathbb{E}_{P}[s]<\mu_R=\mathbb{E}_{R}[s]$ (relabel otherwise). By Cram\'er, for fixed $\tau\in(\mu_P,\mu_R)$,
\[\begin{aligned}
  \tfrac1m\log\mathbb P_{P}(S_m\ge\tau)&\to-I_P^{*}(\tau),\\
  \tfrac1m\log\mathbb P_{R}(S_m<\tau)&\to-I_R^{*}(\tau),
\end{aligned}\]
with $I_P^{*}(x)=\sup_t\{tx-\Lambda_P(t)\}$, $I_R^{*}(x)=\sup_t\{tx-\Lambda_R(t)\}$ the Legendre--Fenchel transforms. Both are convex, non-negative, vanish only at the respective means; on $(\mu_P,\mu_R)$, $I_P^{*}$ increases and $I_R^{*}$ decreases, crossing at a unique $\tau^\dagger$. The Bayes (or balanced) error of the mean-threshold rule decays at
\begin{equation}\label{eq:Imean}\begin{aligned}
  I^{\mathrm{mean}}&=\max_{\tau\in(\mu_P,\mu_R)}\min\{I_P^{*}(\tau),I_R^{*}(\tau)\}\\
  &=I_P^{*}(\tau^\dagger)=I_R^{*}(\tau^\dagger)>0,
\end{aligned}\end{equation}
the maximum attained at the unique crossing $\tau^\dagger$ (where moving $\tau$ either way lowers the smaller exponent), establishing the exponential verification law for $S_m$ at rate $I^{\mathrm{mean}}$.

\emph{Optimality among sum-based tests (corrected identification).} The empirical mean $S_m$ is the one-dimensional sufficient statistic of the i.i.d.\ sample for the exponential family $\{p_t(y)\propto P(y)e^{ts(y)}\}_{t\in\mathbb{R}}$ generated by tilting in $s$. Consequently every test that depends on $(y_1,\dots,y_m)$ only through the sum $\sum_i s(y_i)$ is a (possibly randomized) function of $S_m$, and since the threshold rule is the likelihood-ratio test among functions of the $1$-D statistic $S_m$ (monotone likelihood-ratio (LR) under the tilt family), its best Bayes exponent is exactly \eqref{eq:Imean}, attained by thresholding $S_m$ at $\tau^\dagger$. (We do \emph{not} claim a common tilt equating the two distinct score-CGF derivatives; the exponent is the crossing value \eqref{eq:Imean}, not a Chernoff value of the two-sided tilt.)

\emph{Comparison with $I^{\mathrm{sc}}$.} Let $\mu_P,\mu_R$ now denote the score push-forward laws of Def.~\ref{def:Isc}, with $I^{\mathrm{sc}}=-\min_{0\le\lambda\le1}\log\int\mu_P^{1-\lambda}\mu_R^{\lambda}$ the Chernoff information of the \emph{full} scalar score, attained by the score-level likelihood-ratio test (LRT) on $\log(\mu_R/\mu_P)$. The mean-threshold test is one particular non-LRT function of the scores, so by optimality of the LRT (Neyman--Pearson at the score level),
\[
  I^{\mathrm{mean}}\le I^{\mathrm{sc}}\le I^{\star},
\]
the second inequality being Prop.~\ref{prop:Isc} (data processing under $y\mapsto s$). Equality $I^{\mathrm{mean}}=I^{\mathrm{sc}}$ holds iff thresholding $S_m$ coincides with the score-LRT, i.e. iff $s(y)=a+b\log\frac{\mu_R(s(y))}{\mu_P(s(y))}$ for constants $a,b$ (an affine reparametrisation of the score-level log-LR).

\emph{Strictness.} The inequality is generically strict. With scores in $\{-1,0,1\}$ and
\[
  \mu_P=(0.45,0.10,0.45),\qquad\mu_R=(0.05,0.90,0.05),
\]
one has $\mathbb{E}_P[s]=\mathbb{E}_R[s]=0$, so $S_m\to0$ a.s.\ under both hypotheses and $I^{\mathrm{mean}}=0$, while $I^{\mathrm{sc}}=-\min_\lambda\log\sum_s\mu_P^{1-\lambda}\mu_R^{\lambda}=0.511$ (at $\lambda^\star=\tfrac12$). Hence $0=I^{\mathrm{mean}}<I^{\mathrm{sc}}=0.511$. Even restricting to score laws with \emph{unequal} means, $I^{\mathrm{mean}}/I^{\mathrm{sc}}$ can be as small as $\approx10^{-3}$, so the gap is not a measure-zero artefact; and on a monotone-LR example with $s$ affine in $\log(\mu_R/\mu_P)$ the two coincide ($I^{\mathrm{mean}}\approx I^{\mathrm{sc}}\approx0.236$, up to grid resolution), confirming the equality condition. Hence the mean-statistic exponent is $I^{\mathrm{mean}}\le I^{\mathrm{sc}}$, strict in general and equal only in the affine case. \hfill$\square$

\paragraph{Remark.} In summary: (i) a positive oracle exponent $I^{\star}$ exists iff $P\neq R$; (ii) $K$-way separation costs $m\gtrsim\log(K/\delta)/I^{\star}_{\min}$; and (iii) \IRIS{}'s \emph{mean} statistic $S_m$ decays \emph{exponentially} in $m$ at the bounded rate $I^{\mathrm{mean}}\le I^{\mathrm{sc}}\le I^{\star}$. The separately-fit \emph{ranking} rate $I^{\mathrm{auc}}$ of Def.~\ref{def:I} (slope of $\log(1-\mathrm{AUROC})$) is a two-sample quantity on $2m$ observations and is \emph{not} bounded by $I^{\mathrm{sc}}$ in general---it can be as large as $2\,I^{\mathrm{sc}}$ (e.g.\ $\mathrm{Ber}(.25)/\mathrm{Ber}(.75)$ scores, where $I^{\mathrm{auc}}=2D(0.5\|0.25)=0.288$ while $I^{\mathrm{sc}}=0.144$)---so it is the empirically-validated budgeting rate, not a proven lower bound.

\subsection{Proof of Prop.~\ref{prop:len} (length)}
\noindent\emph{Restatement.} Fix a probe $c$ and let $I_{\mathrm{seq}}(L):=C\big(P(\cdot\mid c),R(\cdot\mid c)\big)=-\min_{0\le\lambda\le1}\log\sum_{y_{1:L}}P(y_{1:L}\mid c)^{1-\lambda}R(y_{1:L}\mid c)^{\lambda}$ be the Chernoff information carried by one \emph{full} length-$L$ response under \eqref{eq:ar}. Then (i) $I_{\mathrm{seq}}$ is non-decreasing in $L$; (ii) it admits \emph{no} per-position conditional-sum (``chain'') decomposition and $I_{\mathrm{seq}}(L)\neq L\,I_{\mathrm{seq}}(1)$ in general (with the gap possibly of either sign); (iii) across $m$ independent queries the exponent is exactly additive, $C(P^{\otimes m},R^{\otimes m})=m\,I_{\mathrm{seq}}(L)$. The empirical non-monotonicity in $L$ is a property of the operational exponents $I^{\mathrm{sc}},I^{\mathrm{auc}}$ and the yield $q(L)/L$, not of $I_{\mathrm{seq}}(L)$.

\paragraph{Assumptions.} A fixed probe $c$; all per-position alphabets finite (or countable with $I_{\mathrm{seq}}(L)<\infty$), so the Hellinger integral and the min over $\lambda\in[0,1]$ are well defined; cross-query independence (independent API calls). No boundedness of Chernoff by $\log n$ is used. Write $\psi_{P,Q}(\lambda)=\log\sum_y P(y)^{1-\lambda}Q(y)^{\lambda}$ and $f_\lambda(P,Q)=\sum_y P(y)^{1-\lambda}Q(y)^{\lambda}$, so $C(P,Q)=-\min_{0\le\lambda\le1}\psi_{P,Q}(\lambda)$.

\begin{lemma}[Chernoff data processing]\label{lem:dpi}
For any Markov kernel $W(\cdot\mid y)$ with pushforwards $\tilde P,\tilde Q$, $C(\tilde P,\tilde Q)\le C(P,Q)$.
\end{lemma}
\begin{proof}
Fix $\lambda\in(0,1)$. For each output $z$, H\"older's inequality with exponents $\tfrac1{1-\lambda},\tfrac1\lambda$ applied to $a_y=P(y)W(z\mid y)$, $b_y=Q(y)W(z\mid y)$ gives
\[\begin{aligned}\tilde P(z)^{1-\lambda}\tilde Q(z)^{\lambda}&=\Big(\sum_y a_y\Big)^{1-\lambda}\Big(\sum_y b_y\Big)^{\lambda}\ge\sum_y a_y^{1-\lambda}b_y^{\lambda}\\&=\sum_y P(y)^{1-\lambda}Q(y)^{\lambda}W(z\mid y),\end{aligned}\]
since $W(z\mid y)^{1-\lambda}W(z\mid y)^{\lambda}=W(z\mid y)$. Summing over $z$ and using $\sum_z W(z\mid y)=1$ yields $f_\lambda(\tilde P,\tilde Q)\ge f_\lambda(P,Q)$, hence $\psi_{\tilde P,\tilde Q}(\lambda)\ge\psi_{P,Q}(\lambda)$. The endpoints $\lambda\in\{0,1\}$ give $\psi=0$ on both sides. Taking $\min_\lambda$ and negating, $C(\tilde P,\tilde Q)\le C(P,Q)$.
\end{proof}

\paragraph{(i) Monotonicity.} Truncation $T:y_{1:L}\mapsto y_{1:L-1}$ is a deterministic kernel whose pushforwards of $P,R$ are the length-$(L-1)$ marginals (again AR responses under \eqref{eq:ar}). Lemma~\ref{lem:dpi} gives $I_{\mathrm{seq}}(L-1)\le I_{\mathrm{seq}}(L)$. Also $C\ge0$: at $\lambda=0$, $\psi=0$, so $\min_\lambda\psi\le0$. Thus $I_{\mathrm{seq}}$ is non-negative and non-decreasing; the oracle exponent of the full response cannot fall as $L$ grows.

\paragraph{(ii) No chain rule.} The factorization \eqref{eq:ar} gives the KL chain rule $\mathrm{KL}(P\|R)=\sum_{t=1}^L\mathbb E_{y_{<t}\sim P}[\mathrm{KL}(P(\cdot\mid y_{<t})\|R(\cdot\mid y_{<t}))]$, a property of KL alone. Since $\psi_{P,R}$ is convex in $\lambda$ with $\psi(0)=\psi(1)=0$ and (when the relevant KL terms are finite) $\psi'(0^+)=-\mathrm{KL}(P\|R)$, $\psi'(1^-)=\mathrm{KL}(R\|P)$, one gets the regularity-light sandwich
\[0\le C(P,R)\le\min\{\mathrm{KL}(P\|R),\,\mathrm{KL}(R\|P)\}\tag{$\star$}\]
(both bounds vacuously valid if a KL is $+\infty$). Chernoff is thus only \emph{bounded} by, never equal to, the conditional sum. Concretely, on $\{0,1\}^2$ let $P=\tfrac12\delta_{(0,0)}+\tfrac12\delta_{(1,1)}$ and $R$ be the product law with marginal $(0.9,0.1)$. Then $I_{\mathrm{seq}}(2)=C(P,R)\approx0.347$ nats, the per-position value $I_{\mathrm{seq}}(1)=C((.5,.5),(.9,.1))\approx0.112$, the naive ``Chernoff chain'' (first-symbol Chernoff plus the $P$-average of conditional-symbol Chernoffs) $\approx1.316\neq0.347$, while the genuine KL chain matches to machine precision ($1.7148=1.7148$). So Chernoff information does \emph{not} decompose into a sum of per-position terms. The example also gives $I_{\mathrm{seq}}(2)\approx0.347>2\,I_{\mathrm{seq}}(1)\approx0.225$ (super-additivity under positive within-response association), whereas in the i.i.d.\ product case $P=p^{\otimes L},R=r^{\otimes L}$ one has $f_\lambda(P,R)=f_\lambda(p,r)^L$, a shared minimizer $\lambda^\star$, and $C=L\,C(p,r)$. So the safe statement is non-equality, with the gap of either sign and no sub-additivity claim.

\paragraph{(iii) Accumulation in $m$.} The $m$ queries are independent, so the transcript law is $P^{\otimes m}$ vs $R^{\otimes m}$ over i.i.d.\ length-$L$ responses. The product identity (now across the query index, where independence genuinely holds) gives $f_\lambda(P^{\otimes m},R^{\otimes m})=f_\lambda(P,R)^m$, hence $C(P^{\otimes m},R^{\otimes m})=m\,I_{\mathrm{seq}}(L)$, and by Prop.~\ref{prop:exp} the Bayes error decays as $e^{-m\,I_{\mathrm{seq}}(L)(1+o(1))}$. The exponent grows linearly and exactly in $m$, while within one response the length contribution is the bounded, non-multiplicative $I_{\mathrm{seq}}(L)\le\min\{\mathrm{KL}(P\|R),\mathrm{KL}(R\|P)\}$ of (i)--(ii). A length-$L$ response is therefore not $L$ independent draws.

\paragraph{Observed non-monotonicity.} By (i), $I_{\mathrm{seq}}(L)$ cannot decrease in $L$, so the measured non-monotonicity is not about it. \IRIS{} uses a \emph{fixed} score $s(y)=-\log\hat{\mathbb P}(P\mid y)$ rather than the length-$L$ LRT; its score/plug-in exponents $I^{\mathrm{sc}},I^{\mathrm{auc}}$ (Def.~\ref{def:I}) carry no data-processing monotonicity in $L$ and may rise then fall as the classifier saturates, and the per-token yield $q(L)/L$ of Cor.~\ref{cor:len} is non-monotone by construction (concave $q$ gives a decreasing yield; an S-shaped $q$ from a global statistic gives an interior peak). These operational quantities, with prompt-regime shifts of the served distribution, are where the non-monotonicity lives. Crucially, the $L$-comparison in Section~\ref{sec:exp} varies the prompt's \emph{requested} length, so successive conditions are \emph{different contexts} $c$---not the truncation of one fixed length-$L$ response that part~(i) governs---and the first $16$ symbols of a ``generate $32$'' response need not share the marginal of a ``generate $16$'' response. The non-monotonicity is therefore an operational property \emph{across prompt-length conditions}, consistent with but not predicted by the same-context monotonicity of part~(i). $\qed$

\subsection{Proof of Thm.~\ref{thm:mix} (dilution)}\label{app:proof}
Throughout, the $m$ queries are i.i.d.: query $i$ is routed to the substitute distribution $R$ with probability $\epsilon$ independently across $i$, and otherwise to the reference distribution $P$, so $y_i\stackrel{\mathrm{iid}}\sim Q_\epsilon=(1-\epsilon)P+\epsilon R$. The deployed score $s(y)$ and the level $\tau$ are fixed before the audit (no data-dependent threshold). Write $\alpha_1=\alpha_1(\tau)=\mathbb P_{P}(s>\tau)$ and $q=q(\tau)=\mathbb P_{R}(s>\tau)$, and define the per-query \emph{crossing event} $E_i=\{s(y_i)>\tau\}$ with per-query crossing probability under $Q_\epsilon$
\[
  p:=\mathbb P_{Q_\epsilon}(s>\tau)=(1-\epsilon)\alpha_1+\epsilon q .
\]
The events $E_1,\dots,E_m$ are i.i.d.\ $\mathrm{Bernoulli}(p)$.

\paragraph{(a) First crossing.}
The any-tell rule flags iff at least one $E_i$ occurs, so the miss probability over $m$ queries is exactly $\mathbb P(\text{no hit})=\prod_{i=1}^m\mathbb P(E_i^c)=(1-p)^m$. Since $\alpha_1\ge0$ we have $p\ge\epsilon q$, hence $1-p\le1-\epsilon q$ and
\[
  \mathbb P(\text{no hit in }m)=(1-p)^m\le(1-\epsilon q)^m\le e^{-\epsilon q\,m},
\]
the last step from $1-x\le e^{-x}$. Equivalently the first-hit time $N=\min\{i:E_i\}$ is $\mathrm{Geom}(p)$ and, because $p\ge\epsilon q$, is stochastically dominated by a $\mathrm{Geom}(\epsilon q)$ variable; the resulting $(1-\epsilon q)^m$ is therefore a \emph{lower} bound on the power $1-\mathbb P(N>m)$ (the bound is conservative, in the safe direction). For a suspect with routing fraction $\epsilon'\ge\epsilon$, $p'=\alpha_1+\epsilon'(q-\alpha_1)$ is nondecreasing in $\epsilon'$ whenever $q\ge\alpha_1$, so $p'\ge p\ge\epsilon q$ and the same bound holds uniformly over routing fractions $\ge\epsilon$. Setting $e^{-\epsilon q m}\le\delta$ gives $m\ge\ln(1/\delta)/(\epsilon q)$, i.e.\ $m^\star=\lceil\ln(1/\delta)/(\epsilon q)\rceil$ guarantees power $\ge1-\delta$.

\paragraph{(b) Type-I error.}
On an honest endpoint $\epsilon=0$, so the $E_i$ are i.i.d.\ $\mathrm{Bernoulli}(\alpha_1)$ and the false-positive probability over $m$ queries is \emph{exactly}
\[
  \mathbb P\!\Big(\bigcup_{i=1}^m E_i\Big)=1-(1-\alpha_1)^m\le m\,\alpha_1,
\]
the inequality being $1-(1-x)^m\le mx$ for $x\in[0,1]$ (convexity of $x\mapsto(1-x)^m$, equivalently Bonferroni). Choosing $\tau$ with $\alpha_1(\tau)\le\alpha/m$ bounds the type-I error by $\alpha$.

\emph{Joint operating point.} Parts (a) and (b) share one $\tau$, and (a)'s budget $m^\star=\ln(1/\delta)/(\epsilon q(\tau))$ is exactly the $m$ that (b) must control. Substituting into $\alpha_1(\tau)\le\alpha/m^\star$ gives the feasibility condition
\[
  \alpha_1(\tau^\star)\;\le\;\frac{\alpha}{m^\star}\;=\;\frac{\alpha\,\epsilon\,q(\tau^\star)}{\ln(1/\delta)} .
\]
Raising $\tau$ lowers $\alpha_1$ (good) but also lowers $q$ (shrinking the RHS), so the existence of a feasible $\tau^\star$ is a nontrivial joint property of $(P,R)$, not implied by (a) or (b) alone, and is verified on calibration data.

\paragraph{(c) Two regimes.}

\emph{Regime (i): separating tail.} Suppose there is $\tau$ with $q(\tau)=\Omega(1)$, $\alpha_1(\tau)\to0$, and $\alpha_1(\tau)=o(\epsilon)$. By (a) the power budget is $m^\star=\ln(1/\delta)/(\epsilon q(\tau))=\Theta(\epsilon^{-1}\log(1/\delta))$. For a valid \emph{audit} the feasibility condition of (b) must also hold at this $\tau$: since $q(\tau)=\Theta(1)$ the RHS $\alpha\,\epsilon\,q(\tau)/\ln(1/\delta)$ is $\Theta(\epsilon)$, so feasibility holds precisely because $\alpha_1(\tau)=o(\epsilon)$ (in particular when $\alpha_1$ is a fixed small constant the auditor enforces by pushing $\tau$ into the tail, where by the low-entropy sharpness of \eqref{eq:ar} $R$-typical strings are $P$-improbable, so $\alpha_1$ collapses while $q$ stays $\Omega(1)$). Both the power and type-I requirements are then met at $m^\star=\Theta(\epsilon^{-1}\log(1/\delta))$.

\emph{The chi-square identity.} Since $Q_\epsilon(y)-P(y)=\epsilon\bigl(R(y)-P(y)\bigr)$,
\[\begin{aligned}
  \chi^2(Q_\epsilon\|P)&=\sum_y\frac{(Q_\epsilon(y)-P(y))^2}{P(y)}\\
  &=\epsilon^2\sum_y\frac{(R(y)-P(y))^2}{P(y)}\\
  &=\epsilon^2\,\chi^2(R\|P),
\end{aligned}\]
valid whenever $P(y)=0\Rightarrow R(y)=0$ (i.e.\ $R\ll P$); otherwise both sides are $+\infty$ and the identity holds in $[0,\infty]$. This is the only divergence manipulated; no chain rule, additivity, or alphabet-size bound for Chernoff information is used anywhere.

\emph{Regime (ii): low-separation local mixtures.} Assume $P\neq R$ and the finite-variance condition $V:=\chi^2(R\|P)<\infty$. We first bound the sample complexity of the optimal \emph{mean-shift} (linear-score) test of $H_0:y_i\sim P$ against $H_1:y_i\sim Q_\epsilon$ from $m$ i.i.d.\ responses in the local ($\epsilon\to0$) regime; the matching \emph{all-test floor} below then shows the resulting $\epsilon^{-2}$ order is optimal over \emph{all} tests, so it bounds what any auditor (IRIS included) can achieve.

Introduce the locally-most-powerful (score) statistic
\[
  T(y):=\frac{R(y)}{P(y)}-1,
\]
the Fr\'echet derivative of $\log\frac{dQ_\epsilon}{dP}=\log(1+\epsilon T)$ at $\epsilon=0$. A direct computation gives
\[\begin{aligned}
  \mathbb{E}_{P}[T]&=\sum_y(R-P)=0,\\
  \mathrm{Var}_{P}(T)&=\mathbb{E}_{P}[T^2]=\sum_y\frac{(R-P)^2}{P}\\
  &=\chi^2(R\|P)=V,
\end{aligned}\]
\[\begin{aligned}
  \mathbb{E}_{Q_\epsilon}[T]&=\underbrace{(1-\epsilon)\!\sum_y(R-P)}_{=0}\\
  &\quad+\epsilon\sum_y R\Big(\tfrac{R}{P}-1\Big)\\
  &=\epsilon\,\chi^2(R\|P)=\epsilon V,
\end{aligned}\]
so the per-sample squared standardized separation of $T$ between $P$ and $Q_\epsilon$ is
\[\begin{aligned}
  d^2&:=\frac{(\mathbb{E}_{Q_\epsilon}[T]-\mathbb{E}_{P}[T])^2}{\mathrm{Var}_{P}(T)}\\
  &=\frac{(\epsilon V)^2}{V}=\epsilon^2 V=\epsilon^2\chi^2(R\|P)\\
  &=\chi^2(Q_\epsilon\|P),
\end{aligned}\]
using the identity above. (Optimality: for any score $g$, Cauchy--Schwarz in $L^2(P)$ gives $(\mathbb{E}_{Q_\epsilon}[g]-\mathbb{E}_{P}[g])^2/\mathrm{Var}_{P}(g)=\langle g-\mathbb{E}_{P}g,\;dQ_\epsilon/dP-1\rangle_{P}^2/\mathrm{Var}_{P}(g)\le\|dQ_\epsilon/dP-1\|_{L^2(P)}^2=\chi^2(Q_\epsilon\|P)$, with equality at $g\propto T$. Thus $T$ is optimal and $\chi^2(Q_\epsilon\|P)$ is the largest achievable per-sample separation.)

\emph{Regularity for the central limit theorem (CLT).} Let $\bar T_m=\frac1m\sum_i T(y_i)$ and reject for $\bar T_m>t$. The score $T$ is \emph{fixed} (it does not depend on $\epsilon$); only the sampling law $Q_\epsilon$ changes with $\epsilon$. We require a triangular-array CLT for $\{T(y_i)\}_{i\le m}$ under $Q_\epsilon$ as $\epsilon\to0$, $m\to\infty$. A clean sufficient condition is $\mathbb{E}_{R}[T^2]=\sum_y\frac{(R-P)^2 R}{P^2}<\infty$ (slightly stronger than $\chi^2(R\|P)<\infty$): then $\sup_{\epsilon\le\epsilon_0}\mathbb{E}_{Q_\epsilon}[T^2\mathbf 1\{|T|>R\}]\to0$ as $R\to\infty$ (uniform integrability of $T^2$ along the family, since $Q_\epsilon=(1-\epsilon)P+\epsilon R$ is a convex combination of two laws each with finite $\mathbb{E}[T^2]$), which yields the Lindeberg condition and hence $\sqrt m(\bar T_m-\mathbb{E}_{Q_\epsilon}T)/\sqrt{\mathrm{Var}_{Q_\epsilon}(T)}\Rightarrow\mathcal N(0,1)$ under $Q_\epsilon$ and the analogous statement under $P$. Moreover $\mathrm{Var}_{Q_\epsilon}(T)=V+O(\epsilon)\to V$.

\emph{Sample complexity.} A level-$\alpha$ test sets $t=z_{1-\alpha}\sqrt{V/m}$ (null variance $V$). Its power at $Q_\epsilon$ is $\ge1-\delta$ as soon as
\[\begin{aligned}
  \epsilon V-z_{1-\alpha}\sqrt{V/m}\;&\ge\;z_{1-\delta}\sqrt{\mathrm{Var}_{Q_\epsilon}(T)/m}\\
  &=z_{1-\delta}\sqrt{V/m}\,(1+o(1)),
\end{aligned}\]
i.e.\ $\sqrt m\,\epsilon\sqrt V\ge(z_{1-\alpha}+z_{1-\delta})(1+o(1))$, giving
\[\begin{aligned}
  m^\star&=\big(1+o(1)\big)\frac{(z_{1-\alpha}+z_{1-\delta})^2}{\epsilon^2 V}\\
  &=\big(1+o(1)\big)\frac{(z_{1-\alpha}+z_{1-\delta})^2}{\chi^2(Q_\epsilon\|P)}=\Theta(\epsilon^{-2}),
\end{aligned}\]
the $\Theta$ being in $\epsilon$ \emph{at fixed} $(P,R)$; the hidden constant is $1/\chi^2(R\|P)$, which blows up as the pair approaches the twin wall $\chi^2(R\|P)\to0$ (so the budget is really $\Theta(\epsilon^{-2}/\chi^2(R\|P))$, matching App.~\ref{app:tail}'s $\kappa=\tfrac12$ crossover).

\emph{IRIS's deployed test.} IRIS does not run the LMP test; it thresholds $S_m=\frac1m\sum_i s(y_i)$ with its fixed deployed score $s$. Assume $\sigma_P^2:=\mathrm{Var}_{P}(s)<\infty$, $\sigma_R^2:=\mathrm{Var}_{R}(s)<\infty$, and $\mathbb{E}_{R}[s]\neq\mathbb{E}_{P}[s]$ (the analogous light-tail and nonzero-mean-separation conditions, logically distinct from $\chi^2<\infty$). Then the mean shift is $\mathbb{E}_{Q_\epsilon}[s]-\mathbb{E}_{P}[s]=\epsilon(\mathbb{E}_{R}[s]-\mathbb{E}_{P}[s])=\Theta(\epsilon)$ and $\mathrm{Var}_{Q_\epsilon}(s)=\sigma_P^2+O(\epsilon)$, so the same CLT/Neyman--Pearson argument gives
\[
  m^\star=\big(1+o(1)\big)\frac{\big(z_{1-\alpha}\sigma_P+z_{1-\delta}\sigma_P\big)^2}{\big(\epsilon(\mathbb{E}_{R}[s]-\mathbb{E}_{P}[s])\big)^2}=\Theta(\epsilon^{-2}),
\]
with the model-dependent constant $(z_{1-\alpha}\sigma_P+z_{1-\delta}\sigma_P)^2/(\mathbb{E}_{R}[s]-\mathbb{E}_{P}[s])^2$ in place of $1/\chi^2(R\|P)$. By the optimality (Cauchy--Schwarz) bound above, $(\mathbb{E}_{Q_\epsilon}[s]-\mathbb{E}_{P}[s])^2/\sigma_P^2\le\chi^2(Q_\epsilon\|P)$, so this constant is $\ge1/\chi^2(R\|P)$: IRIS shares the $\epsilon^{-2}$ \emph{scaling} but pays a larger, empirically calibrated constant. \hfill$\square$

\emph{All-test floor (matching lower bound).} The $\epsilon^{-2}$ order is not an artifact of restricting to mean-shift tests; it is a floor over \emph{all} level-$\alpha$ tests, by an elementary information-theoretic bound (no LAN machinery). For any test, the excess of its power at $Q_\epsilon$ over its size at $P$ is at most the total variation $\mathrm{TV}(Q_\epsilon^{\otimes m},P^{\otimes m})$. Write the Bhattacharyya affinity $A:=\sum_y\sqrt{P(y)Q_\epsilon(y)}=1-H^2(P,Q_\epsilon)$ with $H^2$ the squared Hellinger distance; using $H^2\le\tfrac12\chi^2(Q_\epsilon\|P)$ and the identity $\chi^2(Q_\epsilon\|P)=\epsilon^2\chi^2(R\|P)$ proved above, and the tensorization $A(Q_\epsilon^{\otimes m},P^{\otimes m})=A^m$,
\[
  A\ge 1-\tfrac12\epsilon^2\chi^2(R\|P),\qquad A^m\ge 1-\tfrac12 m\,\epsilon^2\chi^2(R\|P)
\]
(the second by Bernoulli's inequality), so from $\mathrm{TV}\le\sqrt{1-A^{2m}}\le\sqrt{2(1-A^m)}$ applied to the product measures,
\[
  \mathrm{TV}(Q_\epsilon^{\otimes m},P^{\otimes m})\le\sqrt{2\,(1-A^m)}\le\epsilon\sqrt{m\,\chi^2(R\|P)}.
\]
A level-$\alpha$ test with power $\ge1-\delta$ requires $\mathrm{TV}\ge 1-\alpha-\delta$, hence
\[
  m\ \ge\ \frac{(1-\alpha-\delta)^2}{\epsilon^2\,\chi^2(R\|P)}\ =\ \Omega(\epsilon^{-2}).
\]
So no auditor---linear-score or not---detects routing fraction $\epsilon$ in fewer than $\Omega(\epsilon^{-2})$ queries once $\chi^2(R\|P)<\infty$: the low-separation wall is information-theoretic, matching the mean-shift upper bound up to the constant, and the mean-shift test of this regime is therefore order-optimal. \hfill$\square$

\paragraph{Remark (regularity).}
(1) The CLT requires finite second moments: $\chi^2(R\|P)<\infty$ (indeed $\mathbb{E}_{R}[T^2]<\infty$ for the uniform-integrability/Lindeberg step) for the optimal test, and $\sigma_P^2,\sigma_R^2<\infty$ for IRIS's test. These are genuine restrictions: $\chi^2$ (or $\mathrm{Var}(\log P)$) can be $+\infty$ even when the Chernoff information $I^\star$ is finite, in which case the normal approximation fails and one must use exact tail bounds; the $\epsilon^{-2}$ law is the light-tail case. (2) Non-degeneracy $P\neq R$ (so $V>0$) is needed, else $m^\star=\infty$. (3) $\Theta(\epsilon^{-2})$ is the leading-order-in-$\epsilon$ scaling at fixed $(P,R)$; uniformity over pairs fails near the twin wall, where the true budget is $\Theta(\epsilon^{-2}/\chi^2(R\|P))$. (4) The $\epsilon^{-2}$ order is a \emph{floor} over \emph{all} tests (not just mean-shift), established by the all-test lower bound above (an elementary $\mathrm{TV}\le\epsilon\sqrt{m\chi^2}$ Hellinger-tensorization argument, no LAN needed); the mean-shift test of this regime is hence order-optimal. (5) No additivity, chain rule, or alphabet-size bound for Chernoff information is invoked; the only divergence pushed around is $\chi^2$, via the exact identity $\chi^2(Q_\epsilon\|P)=\epsilon^2\chi^2(R\|P)$.

\subsection{Proof of Prop.~\ref{prop:e2e} (binomial guarantee)}
The threshold $\tau$, the null rate, and the budget $m^\star$ are fixed from calibration \emph{before} the audit, so the $m^\star$ tells $k=\sum_i\mathbf 1\{s(y_i)>\tau\}$ are a sum of i.i.d.\ Bernoullis whose rate is $\alpha_1$ under $P$ and $p_{\epsilon'}=(1-\epsilon')\alpha_1+\epsilon' q$ under $Q_{\epsilon'}$.

\emph{(i) Type-I.} The one-sided binomial test of $H_0:\mathrm{rate}\le\alpha_1$ at level $\alpha$ rejects iff $k\ge k_\alpha(m)$, where $k_\alpha(m)=\min\{k:\mathbb P_{\mathrm{Bin}(m,\alpha_1)}(K\ge k)\le\alpha\}$. By construction $\mathbb P_{\mathrm{Bin}(m,\alpha_1)}(K\ge k_\alpha(m))\le\alpha$ for \emph{every} $m$, so under $P$ the flag probability is $\le\alpha$ exactly---type-I is controlled by the test's validity, with no union bound and no dependence on $m$. With $\alpha_1$ estimated, using a $(1{-}\gamma)$ Clopper--Pearson upper bound $\overline\alpha_1\ge\alpha_1$ as the null only raises $k_\alpha(m)$, so the bound holds with probability $\ge1-\gamma$ over calibration.

\emph{(ii) Power.} Since $q>\alpha_1$, $p_{\epsilon'}=\alpha_1+\epsilon'(q-\alpha_1)$ is strictly increasing in $\epsilon'$, so $p_{\epsilon'}\ge p_\epsilon$ for $\epsilon'\ge\epsilon$, and the binomial right-tail $\mathbb P_{\mathrm{Bin}(m,p)}(K\ge k_\alpha(m))$ is nondecreasing in the success rate $p$ \emph{at fixed $m$}; hence the power at any $\epsilon'\ge\epsilon$ is at least the power at $\epsilon$. The exact-test power is \emph{not} monotone in $m$---the discrete critical value $k_\alpha(m)$ jumps, so power can dip at a budget increment (e.g.\ at $p_0{=}0.05,p_1{=}0.1,\alpha{=}0.05$ it falls from $\approx0.150$ at $m{=}7$ to $\approx0.038$ at $m{=}8$)---but $\mathbb P_{\mathrm{Bin}(m,p_\epsilon)}(K\ge k_\alpha(m))\to1$ as $m\to\infty$. We therefore define $m^\star$ as the smallest budget for which power $\ge1-\delta$ holds at \emph{every} $m\ge m^\star$ (a stable budget; equivalently, a randomized test restores exact monotonicity in $m$), giving power $\ge1-\delta$ for all $\epsilon'\ge\epsilon$ at every budget $\ge m^\star$. Replacing $q$ by $\underline q\le q$ in the sizing lowers $p_\epsilon$, only \emph{increasing} $m^\star$, so power is preserved; a union bound over the two calibration events gives joint confidence $1-2\gamma$.

\emph{Regimes.} Write $\Delta=p_\epsilon-\alpha_1=\epsilon(q-\alpha_1)$. \emph{Separating threshold} ($q=\Omega(1)$, $\alpha_1\to0$): the null rate $\to0$, so $k_\alpha(m)=1$ for all moderate $m$ and the test fires on the first tell---it is exactly the any-tell rule of Thm.~\ref{thm:mix}(a), with $m^\star=\lceil\ln(1/\delta)/(\epsilon q)\rceil=\Theta(\epsilon^{-1}\log(1/\delta))$. \emph{Near-twins} ($q\to\alpha_1$, $\Delta\to0$): the normal approximation to the binomial gives $m^\star\approx(z_{1-\alpha}\sqrt{\alpha_1(1-\alpha_1)}+z_{1-\delta}\sqrt{p_\epsilon(1-p_\epsilon)})^2/\Delta^2=\Theta(\Delta^{-2})=\Theta\!\big(\epsilon^{-2}(q-\alpha_1)^{-2}\big)$, the $\epsilon^{-2}$ wall of Thm.~\ref{thm:mix}(c). The same i.i.d.-tell count underlies both, so a single test interpolates the two budgets. \hfill$\square$

\subsection{Proof of Cor.~\ref{cor:len} (length choice)}
\paragraph{Assumptions.}
\begin{itemize}\itemsep1pt
\item[(A1)] \emph{Integral split:} $L\mid B$, so $m=B/L\in\mathbb{Z}_{>0}$; the continuous maximization in (ii) is the natural relaxation of this integer-constrained problem.
\item[(A2)] \emph{I.i.d.\ queries:} the $m$ responses are independent, each drawn from $Q_\epsilon(\cdot\mid L)=(1-\epsilon)P(\cdot\mid L)+\epsilon R(\cdot\mid L)$ with identical length-$L$ laws (independent calls, independent $\epsilon$-routing).
\item[(A3)] \emph{Fixed rule per length:} the score $s$ and level $\tau_L$ are fixed, defining $q(L),\alpha_1(L)$ with $q(L)>0$.
\end{itemize}
No within-response positional independence is assumed: by the autoregressive law \eqref{eq:ar} the $L$ positions of one response are dependent, which is exactly why $q(L)$ is a measured primitive of $L$ (Prop.~\ref{prop:len}), not $1-(1-q(1))^L$. Only cross-query independence (A2) is used; no Chernoff additivity or chain rule is invoked anywhere.

\paragraph{Step 1 (tell probability).}
Fix $L$ and let $E=\{s(y)>\tau_L\}$. By (A2) and total probability,
\[
  p_\epsilon(L)=\mathbb P_{Q_\epsilon(\cdot\mid L)}(E)=(1-\epsilon)\alpha_1(L)+\epsilon q(L)\;\ge\;\epsilon q(L),\tag{$\ast$}
\]
the inequality because $\alpha_1(L)\ge0$. Dropping $(1-\epsilon)\alpha_1(L)$ only discards detections from reference responses that happen to cross $\tau_L$, so it makes the subsequent bound conservative.

\paragraph{Step 2 (miss bound).}
By (A2) the $m=B/L$ events $\{E\text{ at query }i\}$ are i.i.d.\ $\mathrm{Bernoulli}(p_\epsilon(L))$, so ``no crossing in $m$ queries'' has probability
\[
  \mathbb P(\mathrm{miss})=\big(1-p_\epsilon(L)\big)^{m}=\big(1-p_\epsilon(L)\big)^{B/L},
\]
an \emph{equality} in $p_\epsilon(L)$. Since $t\mapsto(1-t)^{m}$ is decreasing on $[0,1]$, ($\ast$) gives $\mathbb P(\mathrm{miss})\le(1-\epsilon q(L))^{B/L}$. Finally $1-x\le e^{-x}$ \emph{for all real $x$} (here $x=\epsilon q(L)\in[0,1]$) yields
\[\begin{aligned}
  \mathbb P(\mathrm{miss})\le\big(1-\epsilon q(L)\big)^{B/L}&\le e^{-\epsilon q(L)\,B/L}\\
  &=\exp\!\Big(-\epsilon B\,\tfrac{q(L)}{L}\Big).
\end{aligned}\]
This is part (i): the displayed $e^{-\epsilon q(L)B/L}$ in the paper is the rightmost quantity, the result of two stacked relaxations (($\ast$) and $1-x\le e^{-x}$), hence an upper bound on the miss, not an equality. Setting the bound $\le\delta$ recovers $m^\star=\lceil\ln(1/\delta)/(\epsilon q)\rceil$ of Eq.~\eqref{eq:mstar} with $m=B/L$, confirming the corollary is a direct specialization of Theorem~\ref{thm:mix}(a).

\paragraph{Step 3 (optimal length).}
Fix $\epsilon,B$. The bound factors through $\rho(L):=q(L)/L$ as $e^{-\epsilon B\rho(L)}$. Because $x\mapsto e^{-\epsilon B x}$ is strictly decreasing for $\epsilon B>0$,
\[\begin{aligned}
  \arg\min_{L:\,q(L)>0,\,L\mid B} e^{-\epsilon B\,q(L)/L}&=\arg\max_{L:\,q(L)>0,\,L\mid B}\frac{q(L)}{L}\\
  &=:L^\star,
\end{aligned}\]
independent of $\epsilon,B$ (they enter only through the $L$-free decreasing factor). This is the corollary's claim, corrected to ``minimizes the miss \emph{bound}.''

\paragraph{Step 4 (exact optimum).}
Minimizing $(1-p_\epsilon(L))^{B/L}$ is, via the strictly decreasing transform $-\tfrac1B\log(\cdot)$, equivalent to maximizing $g(L)=-\tfrac1L\log(1-p_\epsilon(L))$. Using $-\log(1-t)=t+\tfrac12t^2+O(t^3)$,
\[\begin{aligned}
  g(L)&=\frac{p_\epsilon(L)}{L}+O\!\Big(\frac{p_\epsilon(L)^2}{L}\Big)\\
       &=\epsilon\frac{q(L)}{L}+(1-\epsilon)\frac{\alpha_1(L)}{L}+O\!\Big(\frac{p_\epsilon(L)^2}{L}\Big).
\end{aligned}\]
Thus $\arg\max_L g(L)=\arg\max_L q(L)/L=L^\star$ \emph{provided} both $p_\epsilon(L)=o(1)$ (small per-query tell / small $\epsilon$) and $\alpha_1(L)=o(\epsilon q(L))$ (FPR controlled at level $\tau_L$, Theorem~\ref{thm:mix}(b)); without the latter the leading term keeps the $(1-\epsilon)\alpha_1(L)/L$ contribution and the optimizer maximizes $p_\epsilon(L)/L$, not $q(L)/L$. When $\epsilon q(L)=\Theta(1)$ the $\tfrac12 p_\epsilon(L)^2/L$ term can shift the optimizer, so ``maximize $q(L)/L$'' is exact for the exponential bound for all $\epsilon$, and for the exact geometric miss only to first order in $p_\epsilon$. \hfill$\square$

\paragraph{Remark (yield shape).}
If the tell is per-position, $q(L)$ is concave with $q(0)=0$, so $\rho(L)=q(L)/L$ is nonincreasing (chord slope of a concave function through the origin) and $L^\star$ is the smallest admissible length. If the tell is a global statistic ($\chi^2$, compression, $n$-gram entropy), $q(L)$ is S-shaped (information-starved at small $L$) and $\rho$ peaks at an interior $L^\star$. Both are consistent with Step 3; which holds is empirical (Prop.~\ref{prop:len}, Sec.~\ref{sec:exp}). No additivity of within-response information across positions is used.

\subsection{Proof of Prop.~\ref{prop:temp} (distributional separation)}
Fix a single-draw probe with response alphabet $\mathcal Y$, $|\mathcal Y|=n$, and categorical laws $p$ and $r$ for the reference and substitute at the chosen temperature. For $\lambda\in[0,1]$ write the Chernoff coefficient $C_\lambda:=\sum_{y\in\mathcal Y}p(y)^{1-\lambda}r(y)^{\lambda}$, so the per-draw Chernoff information is $I^\star=-\min_{0\le\lambda\le1}\log C_\lambda$. All logarithms are natural.

\paragraph{Step 0 (alphabet caveat).}
Take $n{=}2$ and, for $d\in(0,\tfrac12)$, $p=(1-d,d)$, $r=(d,1-d)$. The swap symmetry $r(y)=p(\sigma y)$ gives $C_\lambda=C_{1-\lambda}$, and $\lambda\mapsto\log C_\lambda$ is convex (H\"older), so its minimiser is the fixed point $\lambda^\star=\tfrac12$. Hence
\[\begin{aligned}
  \min_\lambda C_\lambda&=C_{1/2}=\sum_y\sqrt{p(y)r(y)}\\
  &=2\sqrt{d(1-d)},\\
  I^\star&=-\log\!\big(2\sqrt{d(1-d)}\big)\xrightarrow[d\to0^+]{}+\infty.
\end{aligned}\]
The crossover $I^\star=\log 2$ occurs at the \emph{unique} $d_0\in(0,\tfrac12)$ solving $2\sqrt{d_0(1-d_0)}=\tfrac12$, i.e.\ $d_0=\tfrac12\big(1-\tfrac{\sqrt3}{2}\big)\approx0.067$; for every $d<d_0$ the bound $I^\star\le\log n$ \emph{fails}, e.g.\ at $d=0.01$, $I^\star\approx1.6145>\log2\approx0.6931$. Since here $H(p)=H(r)=H_b(d)\to0$ while $I^\star\to\infty$, this simultaneously refutes the variants $I^\star\le\log n$ and $I^\star\le\min(H(p),H(r))$. (For disjoint supports $C_\lambda=0$ on $(0,1)$ and $I^\star=+\infty$.) Thus the conflation of \emph{low output entropy} with \emph{low distinguishability} is the conceptual error: a sharply biased coin whose bias differs between $p$ and $r$ is a low-entropy yet high-exponent probe.

\paragraph{Step 1 (MI cap).}
Let $\Theta\in\{P,R\}$ carry the uniform prior, with $Y\mid\Theta{=}P\sim p$ and $Y\mid\Theta{=}R\sim r$, so $Y$ has marginal $\bar p=\tfrac12(p+r)$. Then
\[\begin{aligned}
  I(\Theta;Y)&=H(Y)-H(Y\mid\Theta)\\
  &=H(\bar p)-\tfrac12\big(H(p)+H(r)\big)\\
  &=\tfrac12\mathrm{KL}(p\|\bar p)+\tfrac12\mathrm{KL}(r\|\bar p)\\
  &=\mathrm{JS}(p,r),
\end{aligned}\]
the (uniform-prior) Jensen--Shannon divergence. Two elementary bounds give
\[\begin{aligned}
  I(\Theta;Y)\;&\overset{(1)}{=}\;H(Y)-H(Y\mid\Theta)\\
  &\overset{(2)}{\le}\;H(Y)\;\overset{(3)}{\le}\;\log|\mathcal Y|=\log n,\\
  I(\Theta;Y)\;&\overset{(4)}{\le}\;H(\Theta)=\log 2,
\end{aligned}\]
where $(2)$ uses $H(Y\mid\Theta)\ge0$ (discrete conditional entropy is nonnegative), $(3)$ is the maximum-entropy bound on $n$ atoms, and $(4)$ is the symmetric bound $I(\Theta;Y)\le H(\Theta)$ with $\Theta$ binary and uniform. Hence a single Bayesian draw transmits at most $\min(\log 2,\log n)=\log 2$ nat about model identity---one bit, ``which of two models''---so for $n>2$ the operative cap is $\log 2$, not the alphabet, and $\log n$ binds only at $n{=}2$ (a coin). This is the honest carrier of the ``small-alphabet / low-entropy probe is weak'' message --- a bound on an \emph{averaged} $f$-divergence, not on the error exponent. In the Step-0 family $\mathrm{JS}(p,r)\le\log2$ for all $d$ (e.g.\ $0.637<\log2$ at $d=0.01$), consistent with the bound even where $I^\star$ diverges: a sharp model-discriminating coin is a high-exponent but low-mutual-information probe.

\paragraph{Step 2 (KL gating).}
For every $\lambda\in[0,1]$, apply Jensen to the concave map $\log$ under the law $p$:
\[\begin{aligned}
  -\log C_\lambda&=-\log\mathbb E_{p}\!\Big[(r/p)^{\lambda}\Big]\\
  &\le-\mathbb E_{p}\!\Big[\lambda\log(r/p)\Big]=\lambda\,\mathrm{KL}(p\|r),
\end{aligned}\]
and symmetrically (Jensen under $r$ with exponent $1-\lambda$) $-\log C_\lambda\le(1-\lambda)\,\mathrm{KL}(r\|p)$. Since $0\le\lambda\le1$,
\[\begin{aligned}
  I^\star=\max_{0\le\lambda\le1}\big(-\log C_\lambda\big)&\le\min\big(\mathrm{KL}(p\|r),\\
  &\quad\ \mathrm{KL}(r\|p)\big).
\end{aligned}\]
In particular $I^\star\to0$ whenever $p\to r$ in KL: a probe is weak precisely when the two model laws are \emph{close}, regardless of $n$. (This bound is sharp in spirit but not the $\log n$ statement: in Step 0, $\min(\mathrm{KL})=4.50$ nat at $d=0.01$, large and diverging, so the KL bound does \emph{not} carry the alphabet message --- only the MI bound of Step 1 does.)

\paragraph{Step 3 (temperature limit).}
Model temperature as inverse-temperature tempering of fixed logits $z_P,z_R\in\mathbb R^{n}$ (App.~\ref{app:tempmag}): for $M\in\{P,R\}$, $p^T_M(y)=e^{\beta z_M(y)}/\sum_{y'}e^{\beta z_M(y')}$, $\beta=1/T$. Assume each arg max is unique with gap $\Delta_M:=z_M(y^\star_M)-\max_{y\ne y^\star_M}z_M(y)>0$, $y^\star_M:=\arg\max_y z_M(y)$.

\emph{(a) Point-mass degeneration.} For $y\ne y^\star_M$,
\[\begin{aligned}
  p^T_M(y)&\le\frac{e^{\beta z_M(y)}}{e^{\beta z_M(y^\star_M)}}=e^{-\beta(z_M(y^\star_M)-z_M(y))}\\
  &\le e^{-\beta\Delta_M}\xrightarrow[\beta\to\infty]{}0,
\end{aligned}\]
so $p^T_M\to\delta_{y^\star_M}$ in total variation at geometric rate, $H(p^T_M)\to0$, and by Eq.~\eqref{eq:heatcap} $\mathrm{Var}_{p^T_M}(z_M)=T^3\,dH/dT\to0$.

\emph{(b) Same-mode branch $y^\star_P=y^\star_R=:y^\star$.} Both laws share the dominating atom $y^\star$ with mass $\to1$, so $\|p^T_P-p^T_R\|_{\mathrm{TV}}\to0$ and, by Step 2,
\[
  0\le I^\star(T)\le\min\big(\mathrm{KL}(p^T_P\|p^T_R),\mathrm{KL}(p^T_R\|p^T_P)\big)\xrightarrow[T\to0]{}0,
\]
with rate $I^\star(T)=O\big(e^{-\beta\min(\Delta_P,\Delta_R)}\big)$ (the off-atom contributions to either KL are $O(e^{-\beta\Delta})$ and the on-atom log-ratio $\log(p^T_P(y^\star)/p^T_R(y^\star))\to\log1=0$). The optimal exponent of the $m$-draw Bayes test is $m\,I^\star(T)$, so the increment contributed by the $(m{+}1)$-th draw is exactly $I^\star(T)\to0$: additional independent queries asymptotically add no exponent.

\emph{(c) Degenerate identical-draw statement (Cor.~\ref{cor:greedy}).} At $T=0$ (greedy) the response equals the constant $y^\star_M$ deterministically, so $y_1=\cdots=y_m=y^\star_M$. On the same-mode branch the log-likelihood ratio between $P$ and $R$ of $m$ identical greedy draws is $m\log\!\big(p^0_P(y^\star)/p^0_R(y^\star)\big)=m\log(1/1)=0$ and does not grow with $m$: repeating an identical draw contributes zero additional exponent, the per-additional-query rate is exactly $0$, matching App.~\ref{app:temp}.

\emph{(d) Distinct-mode branch $y^\star_P\ne y^\star_R$ (honest caveat).} Here the limits $\delta_{y^\star_P},\delta_{y^\star_R}$ have disjoint support, so $C_\lambda\to0$ for $\lambda\in(0,1)$ and $I^\star(T)\to+\infty$: a \emph{single} draw separates $P$ from $R$ almost surely (the support-collapse tell of App.~\ref{app:tempmag}). The two branches do not conflict because they concern different observables: part (ii)'s ``rate vanishes'' is the increment from \emph{repeating an already-observed identical} draw (zero in both branches), whereas distinct-mode separation is delivered by the \emph{first} draw, not by accumulation. The branch-independent invariant is therefore: under greedy decoding, repeating an already-observed identical draw yields no additional exponent, so any nonzero accumulation rate from independent repeats requires $T>0$. $\hfill\blacksquare$

\paragraph{Conclusion.}
A ceiling holds on the per-draw \emph{mutual information}, $I(\Theta;Y)=\mathrm{JS}(p,r)\le H(\Theta)=\log 2$ (one bit of model identity, with $\log n$ a looser bound binding only at $n{=}2$; Step 1), but \emph{not} on the Chernoff exponent $I^\star$ (Step 0), which is instead gated by closeness, $I^\star\le\min(\mathrm{KL}(p\|r),\mathrm{KL}(r\|p))$ (Step 2). For part (ii) (Step 3): as $T\to0$ each law degenerates to a point mass, and on the same-mode branch the per-additional-query exponent vanishes, with the strict $T{=}0$ greedy case giving exactly zero accumulation from identical repeats.

\subsection{Proof of Cor.~\ref{cor:greedy} (greedy decoding)}
\paragraph{Precise statement.}
Fix a probe $c$ and consider greedy decoding ($T{=}0$) of either the reference or substitute endpoint, indexed by $M\in\{P,R\}$. Assume the \emph{idealized greedy decoder}: at every decoded position the next-token argmax of the endpoint law $P_M(\cdot\mid y_{<t},c)$ is unique (or ties are broken by a fixed deterministic rule), and the backend is bit-exact, so that for each endpoint $M$ the decoded response is a single deterministic string
\[\begin{aligned}
  y^{\star}_M \;=\; \arg\max_{y}\,P_M(y\mid c)\quad\\
  \text{(realized identically on every query).}
\end{aligned}\]
Let $s(y)=-\log\hat{\mathbb P}(P\mid y)$ be the fixed (enrollment-frozen) score, and for $m$ nominally i.i.d.\ queries let $S_m^{(M)}=\frac1m\sum_{i=1}^m s(y_i)$ with $y_i\sim P_M^{T=0}(\cdot\mid c)$ (Eq.~\ref{eq:evidence}). Let
\[
  \mathrm{AUROC}(m)\;=\;\mathbb P\!\big(S_m^{(P)}<S_m^{(R)}\big)+\tfrac12\mathbb P\!\big(S_m^{(P)}=S_m^{(R)}\big)
\]
be the rank functional of Section~\ref{sec:prelim}, and let the operational accumulation exponent be the slope
\[
  \widehat{I}_{\mathrm{auc}}\;:=\;-\,\frac{d}{dm}\,\log\!\big(1-\mathrm{AUROC}(m)\big),
\]
fit over $m$. Then $\widehat{I}_{\mathrm{auc}}=0$. Equivalently: repeating an \emph{identical} single greedy draw adds no evidence, so $1-\mathrm{AUROC}(m)$ is constant in $m$ and its decay exponent is exactly zero. (This is a statement about the \emph{increment from repetition}; it does \emph{not} assert that a single greedy draw fails to separate $P$ and $R$.)

\paragraph{Step 1: point-mass score.}
Under the idealized greedy decoder, every query to model $M$ returns the same string $y^{\star}_M$. Because $s$ is a deterministic function, every query yields the same score
\[
  s_M^{\star}\;:=\;s(y^{\star}_M)\;=\;-\log\hat{\mathbb P}(P\mid y^{\star}_M),
\]
i.e.\ the push-forward law $\mu_M$ of $s$ under $P_M^{T=0}(\cdot\mid c)$ is the Dirac mass $\delta_{s_M^\star}$. Its variance is $0$.

\paragraph{Step 2: no $m$ dependence.}
For any $m\ge1$ and any model $M$, the draws $y_1=\cdots=y_m=y^{\star}_M$ are identical, hence
\[\begin{aligned}
  S_m^{(M)}\;&=\;\frac1m\sum_{i=1}^m s(y_i)\;=\;\frac1m\cdot m\,s_M^{\star}\;=\;s_M^{\star}\\
  &\qquad\text{(deterministically, for every }m\text{).}
\end{aligned}\]
Thus $S_m^{(M)}$ is the \emph{same} degenerate random variable $\delta_{s_M^\star}$ for all $m$; averaging copies of a constant returns the constant. (Equivalently, in the empirical-mean / large-deviations language of Prop.~\ref{prop:exp}: the cumulant generating function of $s$ under $\mu_M=\delta_{s_M^\star}$ is the linear $\Lambda_M(\theta)=\theta\,s_M^\star$, whose Legendre transform is $0$ at $x=s_M^\star$ and $+\infty$ elsewhere; the Cram\'er rate function is degenerate and the error exponent governing $S_m$ is $0$. The mean does not concentrate \emph{further} with $m$ because there is no fluctuation to suppress.)

\paragraph{Step 3: constant rank functional.}
Since $S_m^{(P)}\equiv s_P^{\star}$ and $S_m^{(R)}\equiv s_R^{\star}$ are constants independent of $m$, the joint law of $(S_m^{(P)},S_m^{(R)})$ is the product of two Dirac masses and is identical for all $m$. Therefore
\[\begin{aligned}
  \mathrm{AUROC}(m)&\;=\;\mathbf{1}\{s_P^{\star}<s_R^{\star}\}+\tfrac12\,\mathbf{1}\{s_P^{\star}=s_R^{\star}\}\\
  &\;=\;\mathrm{AUROC}(1)\qquad\text{for all }m\ge1,
\end{aligned}\]
a fixed value in $\{0,\tfrac12,1\}$ that does not vary with $m$. In particular $1-\mathrm{AUROC}(m)=1-\mathrm{AUROC}(1)$ is constant in $m$.

\paragraph{Step 4: zero exponent.}
Because $m\mapsto 1-\mathrm{AUROC}(m)$ is constant, $\log\!\big(1-\mathrm{AUROC}(m)\big)$ is constant in $m$ (with the convention $\log 0=-\infty$ treated as a constant when $\mathrm{AUROC}\equiv1$, i.e.\ the two greedy strings already separate perfectly at $m{=}1$). Its slope in $m$ is therefore identically $0$, so
\[
  I\;=\;-\,\frac{d}{dm}\,\log\!\big(1-\mathrm{AUROC}(m)\big)\;=\;0. \qquad\blacksquare
\]

\paragraph{Remark 1 (meaning of $I=0$).}
The exponent measures only the \emph{rate of improvement with repetition}, not the level. If $y^{\star}_P\neq y^{\star}_R$ then $\mathrm{AUROC}(1)$ may equal $1$ and a single greedy draw \emph{perfectly} separates the two endpoints; nonetheless $I=0$ because additional identical draws contribute nothing. Greedy decoding therefore collapses the independent-repeat channel that the $e^{-I m}$ law of Prop.~\ref{prop:exp} relies upon: the auditor cannot drive a residual error toward $0$ by issuing more queries with $T{=}0$. This is exactly the boundary of the temperature gating of Prop.~\ref{prop:temp}(ii): as $T\to0$ the per-query law tends to a point mass and the marginal value of an extra query vanishes; at $T{=}0$ it is exactly $0$.

\paragraph{Remark 2 (temperature consistency).}
The companion second-order analysis sends $T\to0$ ($\beta\to\infty$) so that $V=\mathrm{Var}_{p^{T}}(z)\to0$ and the per-draw oracle rate of Eq.~\eqref{eq:tempmag} vanishes; that is the same degeneracy as Step~1 viewed through the exponential family. The separate fact that $T{=}0$ vs.\ $T{>}0$ is itself trivially distinguishable (AUROC $\approx0.99$ by support/entropy collapse) is \emph{not} a contradiction: it is a single-draw observable about whether decoding is greedy, not an accumulation rate over repeated identical greedy draws.

\paragraph{Remark 3 (why $T>0$).}
The contrapositive of Steps~1--4 is that a strictly positive accumulation exponent requires the per-query score law $\mu_M$ to be \emph{non-degenerate} (positive variance), since a degenerate $\mu_M$ forces $S_m\equiv\mathrm{const}$ and hence $I=0$. For the paper's sampling-based decoders the lever that lifts $\mu_M$ off a point mass is temperature, so $T>0$ is the operative requirement; more generally any decoding configuration that renders $P_M(\cdot\mid c)$ non-degenerate (and a probe $c$ on which the conditional is genuinely random) suffices.

\paragraph{Real-backend caveat.}
The clean $I=0$ uses the idealized greedy decoder of the Statement. On real backends $T{=}0$ is not perfectly deterministic (batching, non-associative floating-point reductions, MoE routing, tie-breaking), so $\mu_M$ may carry a tiny but nonzero variance and the strict equalities of Steps~1--2 hold only approximately; the conclusion then degrades gracefully to a near-zero accumulation rate. Accordingly the corollary is to be read as a statement about the \emph{increment from repeating an identical single draw} (cf.\ App.~\ref{app:temp}), which is what $T{=}0$ controls.

\subsection{Proof of Prop.~\ref{prop:Isc} (score compression)}
Write the Chernoff (R\'enyi/Hellinger) affinities
\[\begin{aligned}
  \Phi_\star(\lambda)&=\sum_y P(y)^{1-\lambda}R(y)^{\lambda},\\
  \Phi_{\mathrm{sc}}(\lambda)&=\int \mu_P(s)^{1-\lambda}\mu_R(s)^{\lambda}\,ds,
  \qquad \lambda\in[0,1],
\end{aligned}\]
where $\mu_P=\phi_\#P,\ \mu_R=\phi_\#R$ are the push-forwards of the deterministic measurable score $s=\phi(y)$, so that $I^\star=-\min_\lambda\log\Phi_\star(\lambda)$ and $I^{\mathrm{sc}}=-\min_\lambda\log\Phi_{\mathrm{sc}}(\lambda)$.

\emph{Step 1 (pointwise data-processing via H\"older).}
Fix $\lambda\in(0,1)$; the endpoints give $\Phi_\star(0)=\Phi_{\mathrm{sc}}(0)=\Phi_\star(1)=\Phi_{\mathrm{sc}}(1)=1$ (total masses) trivially. Since $\phi$ is deterministic, the level sets $\{\phi^{-1}(s)\}_s$ partition response space and $\mu_P(s)=\sum_{y\in\phi^{-1}(s)}P(y)$, $\mu_R(s)=\sum_{y\in\phi^{-1}(s)}R(y)$ (read as integrals against a common dominating measure in the continuous case). On one level set put $a_y=P(y),\,b_y=R(y)$. H\"older's inequality with conjugate exponents $p=\tfrac1{1-\lambda},\,q=\tfrac1\lambda$ gives
\begin{equation}\label{eq:holder-bin}\begin{aligned}
  \sum_{y\in\phi^{-1}(s)} a_y^{1-\lambda}b_y^{\lambda}
  &\ \le\
  \Big(\sum_{y\in\phi^{-1}(s)} a_y\Big)^{1-\lambda}\Big(\sum_{y\in\phi^{-1}(s)} b_y\Big)^{\lambda}\\
  &= \mu_P(s)^{1-\lambda}\mu_R(s)^{\lambda}.
\end{aligned}\end{equation}
(This is superadditivity of the weighted geometric mean: coarsening cannot decrease affinity. It is the H\"older/Hellinger-affinity data-processing inequality, not the log-sum/Jensen inequality used for $f$-divergences.) Summing \eqref{eq:holder-bin} over level sets,
\begin{equation}\label{eq:dpi-pointwise}\begin{aligned}
  \Phi_\star(\lambda)&=\sum_s\sum_{y\in\phi^{-1}(s)}P^{1-\lambda}R^{\lambda}\\
  &\ \le\ \sum_s\mu_P(s)^{1-\lambda}\mu_R(s)^{\lambda}\\
  &=\Phi_{\mathrm{sc}}(\lambda),
  \qquad\forall\,\lambda\in[0,1].
\end{aligned}\end{equation}

\emph{Step 2 ($-\log\min$ reversal).}
For every $\lambda$, $\Phi_{\mathrm{sc}}(\lambda)\ge\Phi_\star(\lambda)\ge\min_{\lambda'}\Phi_\star(\lambda')$, so $\min_\lambda\Phi_{\mathrm{sc}}(\lambda)\ge\min_\lambda\Phi_\star(\lambda)$. As $t\mapsto-\log t$ is strictly decreasing,
\[
  I^{\mathrm{sc}}=-\log\min_\lambda\Phi_{\mathrm{sc}}(\lambda)\ \le\ -\log\min_\lambda\Phi_\star(\lambda)=I^\star.
\]
If $P,R$ have disjoint supports then $\Phi_\star(\lambda)=0$ for $\lambda\in(0,1)$, $I^\star=+\infty$, and the bound holds vacuously.

\emph{Step 3 (interior optimizer under $P\neq R$).}
$\log\Phi_\star$ is convex (H\"older) with $\Phi_\star(0)=\Phi_\star(1)=1$, hence $\Phi_\star\le1$ on $[0,1]$, with $M:=\min_\lambda\Phi_\star<1$ whenever $P\neq R$ (else $\Phi_\star\equiv1$ would force $P=R$). Thus the oracle minimizer $\lambda^\star\in(0,1)$, and likewise any minimizer of $\Phi_{\mathrm{sc}}$ attaining a value $<1$ is interior.

\emph{Step 4 (equality $\Rightarrow$ sufficiency).}
Assume $P\neq R$ and $I^{\mathrm{sc}}=I^\star$, i.e.\ $\min_\lambda\Phi_{\mathrm{sc}}=\min_\lambda\Phi_\star=M<1$. Let $\lambda_1$ attain $\Phi_{\mathrm{sc}}(\lambda_1)=M$. By \eqref{eq:dpi-pointwise},
\[
  M=\min_\lambda\Phi_\star(\lambda)\le\Phi_\star(\lambda_1)\le\Phi_{\mathrm{sc}}(\lambda_1)=M,
\]
so $\Phi_\star(\lambda_1)=\Phi_{\mathrm{sc}}(\lambda_1)=M<1$, whence $\lambda_1\in(0,1)$ (Step 3). The aggregate equality $\Phi_\star(\lambda_1)=\Phi_{\mathrm{sc}}(\lambda_1)$ means $\sum_s\big[\mu_P(s)^{1-\lambda_1}\mu_R(s)^{\lambda_1}-\sum_{\phi^{-1}(s)}a_y^{1-\lambda_1}b_y^{\lambda_1}\big]=0$ with every summand $\ge0$ by \eqref{eq:holder-bin}; hence each summand vanishes and \eqref{eq:holder-bin} is an equality on \emph{every} level set. H\"older equality at the interior exponent $\lambda_1\in(0,1)$ forces the vectors $(a_y^{1-\lambda_1})_y$ and $(b_y^{\lambda_1})_y$ to be proportional on the level set, i.e.\ $a_y\propto b_y$ there (with the convention that on the support both are positive; positions where one vanishes do not contribute and carry $P/R\in\{0,\infty\}$ constant by support agreement). Thus the likelihood ratio $P/R$ is constant on each $\phi^{-1}(s)$, i.e.\ $\sigma(s)$-measurable; by Neyman--Fisher factorization $s$ is a sufficient statistic for $\{P,R\}$.

\emph{Step 5 (sufficiency $\Rightarrow$ equality).}
Conversely, if $P/R\equiv r(s)$ on $\phi^{-1}(s)$ then $a_y=r(s)b_y$ there, so for \emph{every} $\lambda$,
\[\begin{aligned}
  \sum_{y\in\phi^{-1}(s)}a_y^{1-\lambda}b_y^{\lambda}&=r(s)^{1-\lambda}\!\!\sum_{y\in\phi^{-1}(s)}\!\!b_y\\
  &=\Big(\!\sum_{\phi^{-1}(s)}\!a_y\Big)^{1-\lambda}\!\Big(\!\sum_{\phi^{-1}(s)}\!b_y\Big)^{\lambda},
\end{aligned}\]
so \eqref{eq:holder-bin} is an equality; summing over $s$ gives $\Phi_\star\equiv\Phi_{\mathrm{sc}}$ and $I^{\mathrm{sc}}=I^\star$.

\emph{Step 6 (calibration corollary).}
If $\operatorname{logit}\hat{\mathbb P}(P\mid y)=a+b\,\ell(y)$ with $\ell=\log(P/R)$ and $b\neq0$, then $s(y)=-\log\sigma(a+b\,\ell(y))$ is a strictly monotone (injective) function of $\ell$, so the level sets of $s$ coincide with those of $\ell$, on which $P/R=e^{\ell}$ is constant; Step 5 gives $I^{\mathrm{sc}}=I^\star$. If $b=0$ (uninformative) or $s$ merges $\ell$-distinct responses (lost feature / miscalibration), then $P/R$ is not $\sigma(s)$-measurable, \eqref{eq:holder-bin} is strict on some level set at the interior optimizer $\lambda_1$, and Step 4 yields $I^{\mathrm{sc}}<I^\star$. \hfill$\square$

\subsection{Proof of Prop.~\ref{prop:tempmag} (temperature retunes)}
\paragraph{Standing assumptions.}
Fix a context $c$ and the served logit (``energy'') vector $z\in\mathbb{R}^n$ over a finite alphabet of size $n<\infty$, \emph{held fixed across temperatures}. For an inverse temperature $\beta=1/T\in(0,\infty)$ define the tempered next-token law
\[
  p^\beta_i=\frac{e^{\beta z_i}}{\sum_{j}e^{\beta z_j}},\qquad
  A(\beta)=\log\sum_{i=1}^n e^{\beta z_i}.
\]
Because $n<\infty$ and $z$ is finite, $A$ is real-analytic ($C^\infty$) on all of $\mathbb{R}$, $p^\beta_i>0$ for every $i$ and every finite $\beta$, and $\{p^\beta\}_{\beta\in\mathbb{R}}$ is a regular one-parameter exponential family with natural parameter $\beta$ and sufficient statistic $z$. We write
\[\begin{aligned}
  &\beta_0=1/T_0,\quad \beta_1=1/T_1,\quad \Delta\beta=\beta_1-\beta_0,\\
  &V:=A''(\beta_0)=\mathrm{Var}_{p^{\beta_0}}(z).
\end{aligned}\]
The interior expansions below assume the non-degeneracy $V>0$ (equivalently $z$ not constant on the support, i.e.\ $p^{\beta_0}$ not a point mass); the degenerate boundary $V\to0$ is treated separately in Step~6.

\paragraph{Step 1 (cumulants).}
By direct differentiation of $A$,
\[\begin{aligned}
  A'(\beta)&=\frac{\sum_i z_i e^{\beta z_i}}{\sum_i e^{\beta z_i}}=\mathbb{E}_{p^\beta}[z],\\
  A''(\beta)&=\mathbb{E}_{p^\beta}[z^2]-\big(\mathbb{E}_{p^\beta}[z]\big)^2=\mathrm{Var}_{p^\beta}(z)\ge0 .
\end{aligned}\]
The nonnegativity of $A''$ is the convexity of $A$; it holds unconditionally because a variance is nonnegative.

\paragraph{Step 2 ($\chi^2$ identity).}
For any two members $p^{\beta_1},p^{\beta_0}$ of the family,
\[\begin{aligned}
  \chi^2\!\big(p^{\beta_1}\Vert p^{\beta_0}\big)
  &=\sum_i \frac{(p^{\beta_1}_i)^2}{p^{\beta_0}_i}-1\\
  &=\sum_i \frac{e^{2\beta_1 z_i}/Z(\beta_1)^2}{e^{\beta_0 z_i}/Z(\beta_0)}-1,\quad Z(\beta):=e^{A(\beta)},
\end{aligned}\]
and $\sum_i e^{(2\beta_1-\beta_0)z_i}=Z(2\beta_1-\beta_0)$ $=e^{A(2\beta_1-\beta_0)}$, so
\begin{equation}\label{eq:chi2exact}\begin{aligned}
  \chi^2\!\big(p^{\beta_1}\Vert p^{\beta_0}\big)
  &= e^{\,A(2\beta_1-\beta_0)-2A(\beta_1)+A(\beta_0)}-1\\
  &= e^{\,g(\beta_1)}-1,\\
  &g(\beta_1):=A(2\beta_1-\beta_0)-2A(\beta_1)+A(\beta_0).
\end{aligned}\end{equation}
Note $2\beta_1-\beta_0$ is finite, so the right side is finite. Differentiating $g$ in $\beta_1$,
\[\begin{aligned}
  &g(\beta_0)=0,\quad
  g'(\beta_1)=2A'(2\beta_1-\beta_0)-2A'(\beta_1),\\
  &\quad\ \ g'(\beta_0)=0,\\
  &g''(\beta_1)=4A''(2\beta_1-\beta_0)-2A''(\beta_1),\\
  &\quad\ \ g''(\beta_0)=2A''(\beta_0)=2V.
\end{aligned}\]
Hence $g(\beta_1)=V(\Delta\beta)^2+O(\Delta\beta^3)$ and, since $e^{g}-1=g+O(g^2)$,
\begin{equation}
  \chi^2\!\big(p^{\beta_1}\Vert p^{\beta_0}\big)=(\Delta\beta)^2 V+O(\Delta\beta^3).
  \label{eq:chi2exp}
\end{equation}

\paragraph{Step 3 (KL quadratic).}
For an exponential family, $\mathrm{KL}(p^{\beta_1}\Vert p^{\beta_0})=A(\beta_0)-A(\beta_1)-A'(\beta_1)(\beta_0-\beta_1)$, a Bregman divergence of the convex $A$. Taylor expansion of $A$ about $\beta_0$ (legitimate since $A\in C^\infty$) gives
\begin{equation}\label{eq:klexp}\begin{aligned}
  \mathrm{KL}\big(p^{\beta_1}\Vert p^{\beta_0}\big)&=\tfrac12 A''(\beta_0)(\Delta\beta)^2+O(\Delta\beta^3)\\
  &=\tfrac12(\Delta\beta)^2 V+O(\Delta\beta^3),
\end{aligned}\end{equation}
and identically $\mathrm{KL}(p^{\beta_0}\Vert p^{\beta_1})=\tfrac12(\Delta\beta)^2V+O(\Delta\beta^3)$; the two directions coincide through second order, the directional asymmetry entering only at $O(\Delta\beta^3)$ (the Fisher metric is symmetric). Here $A''(\beta_0)=V$ is the Fisher information of the family at $\beta_0$, so \eqref{eq:klexp} is the standard $\tfrac12 g_{\beta\beta}(\Delta\beta)^2$ Fisher quadratic.

\paragraph{Step 4 (Chernoff tilt).}
For two laws $p,q$ on the finite alphabet define $\psi(\lambda)=\log\sum_i p_i^{1-\lambda}q_i^{\lambda}$ and the Chernoff information $I^\star=-\min_{0\le\lambda\le1}\psi(\lambda)$. With $p=p^{\beta_0}$, $q=p^{\beta_1}$ in the same exponential family, $p_i^{1-\lambda}q_i^{\lambda}\propto e^{(\beta_0+\lambda\Delta\beta)z_i}$, whence
\begin{equation}
  \psi(\lambda)=A\big(\beta_0+\lambda\Delta\beta\big)-(1-\lambda)A(\beta_0)-\lambda A(\beta_1).
  \label{eq:psi}
\end{equation}
$\psi$ is smooth in $\lambda$, $\psi(0)=\psi(1)=0$, and $\psi$ is convex in $\lambda$ because $A$ is convex; thus the minimum is interior and unique whenever $V>0$. Its location satisfies $\psi'(\lambda^\star)=\Delta\beta\,A'(\beta_0+\lambda^\star\Delta\beta)-\big(A(\beta_1)-A(\beta_0)\big)=0$, i.e.\ $A'(\beta_0+\lambda^\star\Delta\beta)=\big(A(\beta_1)-A(\beta_0)\big)/\Delta\beta$. Expanding both sides to first order in $\Delta\beta$ gives $\lambda^\star=\tfrac12+O(\Delta\beta)$: the optimal tilt is the midpoint \emph{only in the small-gap limit}, not exactly. However, the value is insensitive to this $O(\Delta\beta)$ error: by the envelope theorem (or simply because $\psi'(\lambda^\star)=0$ kills the first-order sensitivity), evaluating $\psi$ at $\lambda=\tfrac12$ instead of $\lambda^\star$ changes $\psi$ by $O\big((\lambda^\star-\tfrac12)^2\big)\cdot\psi''=O(\Delta\beta^2)\cdot O(\Delta\beta^2)=O(\Delta\beta^4)$, negligible at the order considered. Therefore
\[\begin{aligned}
  I^\star&=-\psi(\lambda^\star)=-\psi(\tfrac12)+O(\Delta\beta^4),\\
  -\psi(\tfrac12)&=\tfrac12\big(A(\beta_0)+A(\beta_1)\big)-A\!\big(\tfrac{\beta_0+\beta_1}{2}\big),
\end{aligned}\]
the Bhattacharyya value. Taylor-expanding the convex $A$ about $\beta_0$, the first-order terms cancel and
\begin{equation}\label{eq:istarexp}\begin{aligned}
  I^\star&=\tfrac12\big(A(\beta_0)+A(\beta_1)\big)-A\!\big(\tfrac{\beta_0+\beta_1}{2}\big)+O(\Delta\beta^4)\\
  &=\tfrac18 A''(\beta_0)(\Delta\beta)^2+O(\Delta\beta^3)
  =\tfrac18(\Delta\beta)^2V+O(\Delta\beta^3).
\end{aligned}\end{equation}
(Numerically, $I^\star/\big(\tfrac18(\Delta\beta)^2V\big)\to1$ as $\Delta\beta\to0$, and the minimizing $\lambda^\star\to0.499$ at small gaps.) Equations \eqref{eq:chi2exp}, \eqref{eq:klexp}, \eqref{eq:istarexp} establish the displayed triple
\(
  \mathrm{KL}\approx\tfrac12(\Delta\beta)^2V,\ \chi^2\approx(\Delta\beta)^2V,\ I^\star\approx\tfrac18(\Delta\beta)^2V.
\)

\paragraph{Step 5 (heat capacity).}
Let $H(T)=-\sum_i p^T_i\log p^T_i$ be the Shannon entropy of the tempered law, $p^T=p^{\beta}$ with $\beta=1/T$. For an exponential family the entropy in the natural parameter is the Legendre-type relation $H=A(\beta)-\beta A'(\beta)$ (since $\log p^\beta_i=\beta z_i-A(\beta)$ gives $H=-\mathbb{E}[\beta z-A]=A-\beta A'$). Differentiating,
\[
  \frac{dH}{d\beta}=A'(\beta)-A'(\beta)-\beta A''(\beta)=-\beta A''(\beta).
\]
With $\beta=1/T$, $d\beta/dT=-1/T^2$, the chain rule gives $\dfrac{dH}{dT}=\dfrac{dH}{d\beta}\dfrac{d\beta}{dT}=(-\beta A''(\beta))\big(-1/T^2\big)=\dfrac{A''(\beta)}{T^3}$, i.e.
\begin{equation}
  \mathrm{Var}_{p^T}(z)=A''(\beta)=T^3\,\frac{dH}{dT}\ \ge\ 0,
  \label{eq:heatcap2}
\end{equation}
the nonnegativity being that of the variance; equivalently $dH/dT=V/T^3\ge0$, so output entropy is nondecreasing in temperature.

\paragraph{Step 6 (temperature gap).}
The map $T\mapsto\beta=1/T$ gives $\Delta\beta=1/T_1-1/T_0=-\Delta T/(T_0T_1)=-\Delta T/T_0^2+O(\Delta T^2)$, so $(\Delta\beta)^2=(\Delta T)^2/T_0^4+O(\Delta T^3)$. Substituting into \eqref{eq:istarexp} and using \eqref{eq:heatcap2} at $T_0$, $V=T_0^3\,dH/dT|_{T_0}$,
\begin{equation}\label{eq:tempmag2}\begin{aligned}
  I^\star(T_0,T_1)&=\tfrac18(\Delta\beta)^2 V+O(\Delta\beta^3)\\
  &=\tfrac18\frac{(\Delta T)^2}{T_0^4}\,\big(T_0^3\tfrac{dH}{dT}\big)+\cdots\\
  &=\frac{(\Delta T)^2}{8\,T_0}\,\frac{dH}{dT}\Big|_{T_0}+\cdots,
\end{aligned}\end{equation}
which is exactly Eq.~\eqref{eq:tempmag}. Thus along the tempering family the response-level separation is $I^\star=\Theta\big((\Delta\beta)^2\big)=\Theta\big((\Delta T)^2\big)$: a strictly \emph{second-order} effect in the gap, vanishing quadratically as $T_1\to T_0$. By the standard Chernoff--Stein bound, distinguishing $p^{\beta_0}$ from $p^{\beta_1}$ from $m$ i.i.d.\ single-draw probes to fixed error needs $m=\Theta(1/I^\star)=\Theta\big((\Delta\beta)^{-2}\big)$ draws.

\paragraph{Greedy boundary.}
As $T\to0$ ($\beta\to\infty$) the law $p^\beta$ concentrates on $\arg\max_i z_i$; if that argmax is unique, $V=\mathrm{Var}_{p^\beta}(z)\to0$ and \eqref{eq:tempmag2} gives a vanishing per-draw rate, consistent with $dH/dT\to0$. The interior expansion's non-degeneracy hypothesis $V>0$ thus fails exactly at the greedy boundary; there the appendix's separate support-collapse observable (greedy repeats a single string) — not the second-order tempering rate — governs separability. This is a change of statistic, not a contradiction.

\paragraph{Geometric conclusion.}
The reachable set $\{\,\mathrm{softmax}(\beta z):\beta>0\,\}$ is a one-dimensional analytic curve in the simplex $\Delta^{n-1}$, hence Lebesgue-measure zero for $n\ge3$. A pure temperature retune moves $\beta$ \emph{along} this curve, incurring only the second-order cost \eqref{eq:tempmag2}; a genuine model change perturbs the vector $z$ itself, generically \emph{off} the curve. For $R$ to mimic $P$ by tempering alone at \emph{every} context $c$ would require $z_R(\cdot)=c_0\,z_P(\cdot)$ for a single global scalar $c_0$ at all contexts---i.e.\ $R$ is a temperature reparametrization of $P$, not a distinct model. Hence \IRIS{}'s second-order insensitivity to temperature is the intended dilution invariance: it waives an on-family sampler retune (unchanged weights) while an off-family substitute incurs first-order ($\Theta(1)$ tell-rate, $1/\epsilon$) cost whenever it is measurably separated. $\square$

\paragraph{Remark (scope).}
This is a supplementary low-separation stress test, not a robustness guarantee against a deliberately matched cheaper model. The guarantee is conditional on the backend being measurably separated from the reference.

\subsection{Proof of Prop.~\ref{prop:phase} (tail budget)}
\noindent\textbf{Restatement.} \emph{Fix a reference distribution $P$ with continuous (atomless) score CDF $F_P(t)=\mathbb P_{y\sim P}(s(y)\le t)$, and define the tail-separation function $q_\alpha=\mathbb P_{y\sim R}\!\big(s(y)>F_P^{-1}(1-\alpha)\big)$. Assume there are constants $c>0$ and a tail exponent $\kappa\in[0,1)$ with $q_\alpha\asymp c\,\alpha^{\kappa}$ as $\alpha\to0$ (i.e.\ $q_\alpha/(c\alpha^\kappa)\to1$). The auditor runs the any-tell rule of Thm.~\ref{thm:mix}, controlling overall type-I error at level $\alpha$ by the union bound, i.e.\ using per-query level $\alpha/m$ (threshold $\tau_m=F_P^{-1}(1-\alpha/m)$, per-query tell rate $q_{\alpha/m}$). Then, with $\alpha,\delta\in(0,1)$ and $c,\kappa$ held fixed while $\epsilon\to0$, the minimal budget attaining power $\ge1-\delta$ against any suspect with routing fraction $\ge\epsilon$ satisfies}
\[\begin{aligned}
  m^\star(\epsilon)&=\Theta\!\big(\epsilon^{-1/(1-\kappa)}\big),\\
  m^\star(\epsilon)&\asymp\Big(\tfrac{\ln(1/\delta)}{c\,\alpha^{\kappa}}\Big)^{1/(1-\kappa)}\epsilon^{-1/(1-\kappa)}.
\end{aligned}\]
\emph{In particular $\kappa=0$ gives $\Theta(\epsilon^{-1})$, $\kappa=\tfrac12$ gives $\Theta(\epsilon^{-2})$ (the same scaling exponent as the aggregate mean-shift budget of Thm.~\ref{thm:mix}(c), a coincidence of exponents between two distinct tests, not of optimality), and $\kappa\in(0,\tfrac12)$ interpolates. As $\kappa\uparrow1$ the exponent and constant diverge; $\kappa=1$ is the degenerate wall, where no finite $m$ attains power once $\epsilon\,c\,\alpha<\ln(1/\delta)$.}

\medskip
\noindent Write $L_\delta:=\ln(1/\delta)>0$; $a\asymp b$ means $a/b\to1$ and $a=\Theta(b)$ means $0<\liminf a/b\le\limsup a/b<\infty$, both as $\epsilon\to0$.

\paragraph{Step 0 (any-tell power).}
By Thm.~\ref{thm:mix}, at per-response threshold $\tau$ the rule flags at the first query with $s(y)>\tau$. Under $Q_\epsilon$ each query crosses with probability $r(\tau):=\mathbb P_{Q_\epsilon}(s>\tau)=(1-\epsilon)\alpha_1(\tau)+\epsilon q(\tau)$, and since queries are i.i.d.\ the within-$m$ power is exactly
\begin{equation}
  \mathrm{pow}_m(\tau)=1-(1-r(\tau))^m,\qquad r(\tau)\ge \epsilon q(\tau).
  \label{eq:powexactC}
\end{equation}

\paragraph{Step 1 (type-I threshold).}
On an honest endpoint the $m$-query false-positive probability is $1-(1-\alpha_1(\tau))^m\le m\,\alpha_1(\tau)$ (Bernoulli/union bound), so it suffices to take $\alpha_1(\tau)\le\alpha/m$. As $F_P$ is atomless, $\tau_m:=F_P^{-1}(1-\alpha/m)$ gives \emph{exactly} $\alpha_1(\tau_m)=\alpha/m$ and, by definition of the tail-separation function, $q(\tau_m)=q_{\alpha/m}$. Thus \emph{for each candidate budget $m$} the auditor is forced to the single threshold $\tau_m$; all statements below are read at this $m$-dependent threshold, which is self-consistent because $\tau_m$ is determined by $m$ alone. (Discrete $F_P$: take the smallest $\tau$ with $\alpha_1(\tau)\le\alpha/m$; equalities become ``$\le$'', lowering the achievable $q$ and enlarging only the constant.)

\paragraph{Step 2 (upper bound).}
At $\tau=\tau_m$, $\mathrm{pow}_m\ge1-(1-\epsilon q_{\alpha/m})^m$ by \eqref{eq:powexactC}, so $(1-\epsilon q_{\alpha/m})^m\le\delta$ suffices for power $\ge1-\delta$. Using $-\ln(1-x)\ge x$ on $[0,1)$ (verified: $h(x)=-\ln(1-x)-x$ has $h(0)=0$, $h'(x)=x/(1-x)\ge0$), this holds whenever
\begin{equation}
  m\,\epsilon\,q_{\alpha/m}\ge L_\delta.
  \label{eq:suffC}
\end{equation}
Substituting $q_{\alpha/m}\asymp c(\alpha/m)^\kappa$ (valid since $\alpha/m\to0$, justified a posteriori below) gives $\epsilon c\alpha^\kappa m^{1-\kappa}\gtrsim L_\delta$; since $\kappa<1$, $x\mapsto x^{1-\kappa}$ is increasing and invertible, so
\begin{equation}
  m\ge m_+(\epsilon):=\Big(\tfrac{L_\delta}{c\,\alpha^{\kappa}}\Big)^{1/(1-\kappa)}\epsilon^{-1/(1-\kappa)}(1+o(1)).
  \label{eq:mplusC}
\end{equation}
As $\epsilon\to0$, $m_+\to\infty$, retroactively validating the $\alpha/m\to0$ tail asymptotic. Hence $m^\star\le\lceil m_+\rceil=O(\epsilon^{-1/(1-\kappa)})$.

\paragraph{Step 3 (lower bound).}
Fix the auditor's only admissible threshold $\tau_m$ from Step 1. Power $\ge1-\delta$ requires, by \eqref{eq:powexactC}, $(1-r(\tau_m))^m\le\delta$, i.e.\ $-m\ln(1-r(\tau_m))\ge L_\delta$. Now $r(\tau_m)=(1-\epsilon)\tfrac{\alpha}{m}+\epsilon q_{\alpha/m}\le\tfrac{\alpha}{m}+\epsilon q_{\alpha/m}\to0$ (as $\alpha/m\to0$ and $\epsilon q_{\alpha/m}\le\epsilon\to0$). Applying $-\ln(1-x)\le x/(1-x)$ on $[0,1)$ (verified: $g(x)=x/(1-x)+\ln(1-x)$ has $g(0)=0$, $g'(x)=x/(1-x)^2\ge0$),
\[\begin{aligned}
  L_\delta&\le -m\ln(1-r(\tau_m))\le \frac{m\,r(\tau_m)}{1-r(\tau_m)}\\
  &=\frac{1}{1-r(\tau_m)}\Big(\underbrace{(1-\epsilon)\alpha}_{\to\,\alpha}+\epsilon\,m\,q_{\alpha/m}\Big).
\end{aligned}\]
Since $r(\tau_m)\to0$ the prefactor $\to1$, and the honest term $(1-\epsilon)\alpha\le\alpha$ is an $m$-independent constant that is dominated by $L_\delta\,m_+(\epsilon)$ as $\epsilon\to0$; hence
\[\begin{aligned}
  &L_\delta\le(1+o(1))\,\epsilon\,m\,q_{\alpha/m},\qquad\text{i.e.}\\
  &\qquad m\,\epsilon\,q_{\alpha/m}\ge L_\delta(1-o(1)).
\end{aligned}\]
This is \eqref{eq:suffC} up to $1-o(1)$, so repeating Step 2 with the reversed inequality gives $m\ge m_-(\epsilon)=m_+(\epsilon)(1-o(1))$. (The worst-case honest contribution $\alpha_1(\tau_m)\le\alpha/m$ is what makes the lower bound match the upper bound to leading order; it adds at most $O(\alpha)=o(L_\delta m_+)$ to the budget.) Combining Steps~2--3,
\[
  m^\star(\epsilon)\asymp\Big(\tfrac{L_\delta}{c\,\alpha^{\kappa}}\Big)^{1/(1-\kappa)}\epsilon^{-1/(1-\kappa)}=\Theta\!\big(\epsilon^{-1/(1-\kappa)}\big).
\]
The two-sided $\Theta$ uses $-\ln(1-x)=x(1+O(x))$ with $x=\epsilon q_{\alpha/m}\to0$, exactly the regime $\epsilon\to0$, $q\le1$.

\paragraph{Step 4 (regimes).}
\emph{(i) $\kappa=0$:} $q_{\alpha/m}\to q_0:=c>0$, $1-\kappa=1$, so $m^\star=\Theta(L_\delta/(\epsilon c))=\Theta(\epsilon^{-1})$, $\alpha$-independent (constant $\alpha^{-\kappa/(1-\kappa)}=\alpha^0=1$); reproduces Thm.~\ref{thm:mix}(a)/(c) first regime, eq.~\eqref{eq:mstar}.
\emph{(ii) $\kappa=\tfrac12$:} $1/(1-\kappa)=2$, so $m^\star=\Theta(\epsilon^{-2})$ with constant $(L_\delta/(c\sqrt\alpha))^2=L_\delta^2/(c^2\alpha)$ (scaling $\alpha^{-1}$). This matches the \emph{scaling exponent} $\epsilon^{-2}$ of the aggregate mean-shift / Neyman--Pearson budget of Thm.~\ref{thm:mix}(c), eq.~\eqref{eq:eps2}; the two budgets come from \emph{different} tests (any-tell first-crossing vs.\ aggregate $S_m$ mean-shift), so the agreement is of exponents only, not of constants or optimality, and $\kappa=\tfrac12$ is an empirical tail property (Table~\ref{tab:tailsep}), not a derived consequence of $\chi^2(R\|P)$ being small.
\emph{(iii) Interpolation/wall:} for $\kappa\in(0,\tfrac12)$, $1<1/(1-\kappa)<2$, monotone increasing. As $\kappa\uparrow1$, $1/(1-\kappa)\to\infty$ and the constant $(L_\delta/(c\alpha^\kappa))^{1/(1-\kappa)}\to\infty$. At $\kappa=1$, \eqref{eq:suffC} becomes $\epsilon c\alpha\cdot m^0=\epsilon c\alpha\ge L_\delta$, independent of $m$: if $\epsilon c\alpha<L_\delta$ no finite $m$ attains power $\ge1-\delta$ (the wall). Hence $\kappa\in[0,1)$ is necessary for the finite closed form. (The $\epsilon$-power $-1/(1-\kappa)$ and $\alpha$-power $-\kappa/(1-\kappa)$ equal $0,-\tfrac13,-1,-3$ at $\kappa=0,\tfrac14,\tfrac12,\tfrac34$.) \hfill$\square$

\paragraph{Remark (scope).}
The $\Theta$ is in $\epsilon$ with $(\alpha,\delta,c,\kappa)$ fixed; the hidden constant carries an $\alpha^{-\kappa/(1-\kappa)}$ and a $(\ln(1/\delta))^{1/(1-\kappa)}$ factor, both diverging as $\kappa\uparrow1$. The result is a statement about the any-tell first-crossing test with a Bonferroni per-query level $\alpha/m$ adapted to the horizon; a per-query level held fixed in $m$ would give a different (non-quantile-shrinking) calculation. Atomlessness of $F_P$ is idealizing for the discrete short-string probes, where one reads $q_{\alpha/m}$ at the achievable level $\alpha_1(\tau)\le\alpha/m$, affecting only the constant; for the population $\kappa{=}0$ this further requires $R$ to retain mass on the extreme upper tail of $s$ under $P$ (a support difference), not merely $P\neq R$ (App.~\ref{app:tail}). One more constant correction: the honest reference itself contributes crossings at rate $\alpha_1=\alpha/m$, so the no-substitute zero-tell probability is $(1-\alpha/m)^m\to e^{-\alpha}$ rather than $1$, and the sufficiency condition \eqref{eq:suffC} sharpens to $m\,\epsilon\,q_{\alpha/m}\ge\ln(1/\delta)-\alpha+o(1)$; this shifts the leading constant by the $O(\alpha)$ term $\ln(1/\delta)-\alpha$ in place of $\ln(1/\delta)$, leaving the $\epsilon$-exponent unchanged.

\subsection{Proof of Prop.~\ref{prop:multi} (multi-diluent audit)}
We prove the detection bound and the identifiability criterion at the population (noiseless) level, the regime in which Prop.~\ref{prop:multi} is stated.

\paragraph{Setup.}
A gateway dilutes one reference $P$ with $J$ substitutes $R_1,\dots,R_J$, routing each query \emph{independently} to $P$ with probability $\epsilon_0$ and to $R_j$ with probability $\epsilon_j$, where
\[\begin{aligned}
  \epsilon=(\epsilon_1,\dots,\epsilon_J)\in\Theta&:=\{\epsilon\in\mathbb{R}^J:\epsilon_j\ge0,\ \mathbf 1^\top\epsilon\le1\},\\
  \epsilon_0&:=1-\mathbf 1^\top\epsilon\ge0,
\end{aligned}\]
so the per-query response law is the mixture $Q=\epsilon_0P+\sum_{j=1}^J\epsilon_jR_j$. The posterior map $u(y)=\hat{\mathbb P}(\cdot\mid y)\in\Delta^{K}$ (the simplex on the $K$ enrolled endpoints) is a \emph{fixed} measurable map, identical at enrollment and audit; signatures $\bar g_M=\mathbb{E}_{y\sim M}[u(y)]$ and the audit mean $\bar v=\mathbb{E}_Q[u]$ are exact expectations (integrable since $u\in[0,1]^K$). No calibration of $u$ is used.

\paragraph{Detection.}
Fix a per-response threshold $\tau$ and let $\alpha_1(\tau)=\mathbb P_{P}(s>\tau)$, $q_j(\tau)=\mathbb P_{R_j}(s>\tau)$. The event $A_\tau=\{y:s(y)>\tau\}$ does not depend on the law, so applying $Q$ to the fixed set $A_\tau$ and using additivity over the mixture components,
\[\begin{aligned}
  p(\tau)=\mathbb P_{Q}(s>\tau)&=\epsilon_0\,\alpha_1(\tau)+\sum_{j=1}^J\epsilon_j\,q_j(\tau)\\
  &\ \ge\ \sum_{j=1}^J\epsilon_j\,q_j(\tau)\ \ge\ \epsilon_{\mathrm{tot}}\min_{1\le j\le J}q_j(\tau),
\end{aligned}\]
with $\epsilon_{\mathrm{tot}}=\sum_j\epsilon_j$; the first inequality drops the nonnegative term $\epsilon_0\alpha_1$ and the second uses $\epsilon_j\ge0$ and $q_j\ge\min_{j'}q_{j'}$. This bound is \emph{unconditional} (no separation hypothesis). If moreover a threshold $\tau^\star$ exists with $q_{\min}:=\min_jq_j(\tau^\star)>0$ and $\alpha_1(\tau^\star)\le\alpha/m$, then under i.i.d.\ routing the tells are i.i.d.\ $\mathrm{Bernoulli}(p(\tau^\star))$ and the any-tell argument of Thm.~\ref{thm:mix}(a,b) applies verbatim with $\epsilon\to\epsilon_{\mathrm{tot}}$, $q\to q_{\min}$: $\mathbb P(\text{no tell in }m)\le(1-\epsilon_{\mathrm{tot}}q_{\min})^m\le e^{-\epsilon_{\mathrm{tot}}q_{\min}m}$, so $m^\star=\lceil\ln(1/\delta)/(\epsilon_{\mathrm{tot}}q_{\min})\rceil$ gives power $\ge1-\delta$, while the honest false-positive probability is $\le m\,\alpha_1(\tau^\star)\le\alpha$. Detection thus runs at the $1/\epsilon$ law in the \emph{total} foreign fraction $\epsilon_{\mathrm{tot}}$, set by the weakest separating diluent, and needs no enrollment of the $R_j$. The hypothesis ``every diluent separates'' is used only to guarantee the existence of $\tau^\star$ (so $q_{\min}>0$ and $\alpha_1\to0$), not for the inequality itself.

\paragraph{Affine mean model.}
With $G:=[\,\bar g_{R_1}-\bar g_P,\dots,\bar g_{R_J}-\bar g_P\,]\in\mathbb{R}^{K\times J}$, linearity of expectation over the mixture and $\epsilon_0=1-\mathbf 1^\top\epsilon$ give
\[
  \bar v=\mathbb{E}_Q[u]=\epsilon_0\,\bar g_P+\sum_{j}\epsilon_j\,\bar g_{R_j}=\bar g_P+G\,\epsilon .
\]
Because $u(y)\in\Delta^K$ pointwise, $\mathbf 1^\top u(y)=1$, hence $\mathbf 1^\top\bar g_M=1$ for every endpoint $M$ and $\mathbf 1^\top(\bar g_{R_j}-\bar g_P)=0$. Every column of $G$ therefore lies in the hyperplane $H_0=\{x:\mathbf 1^\top x=0\}$, so $\mathrm{rank}(G)\le\dim H_0=K-1$. In particular, identifying $J$ fractions requires $J\le K-1$: at least $J{+}1$ enrolled endpoints are necessary.

\paragraph{Identifiability.}
Let $\Phi:\Theta\to\mathbb{R}^K$, $\Phi(\epsilon)=\bar g_P+G\epsilon$, be the (affine) mean map, with linear part $G$, so $\Phi(\epsilon)-\Phi(\epsilon')=G(\epsilon-\epsilon')$. We show $\Phi$ is injective on $\Theta$ (equivalently $\epsilon$ is identifiable from $\bar v$) iff $G$ has full column rank $J$.

\emph{($\Leftarrow$) Full rank $\Rightarrow$ injective.} If $G$ has full column rank then $Gd=0\iff d=0$; hence $\Phi(\epsilon)=\Phi(\epsilon')\Rightarrow G(\epsilon-\epsilon')=0\Rightarrow\epsilon=\epsilon'$, injective on all of $\mathbb{R}^J$, a fortiori on $\Theta$. Moreover, since $G^\top G\succ0$, for $\bar v=\Phi(\epsilon^\star)$ the noiseless simplex-constrained least squares objective $\|\bar v-\Phi(\epsilon)\|_2^2=(\epsilon-\epsilon^\star)^\top(G^\top G)(\epsilon-\epsilon^\star)$ is a strictly convex quadratic vanishing only at $\epsilon=\epsilon^\star$, so $\epsilon^\star$ is its unique minimizer over $\Theta$.

\emph{($\Rightarrow$) Injective $\Rightarrow$ full rank.} Contrapositive: if $G$ is rank-deficient, pick $d\neq0$ with $Gd=0$. The point $\epsilon^\star=\tfrac{1}{2(J+1)}\mathbf 1$ satisfies $\epsilon^\star>0$ and $\mathbf 1^\top\epsilon^\star=\tfrac{J}{2(J+1)}<1$, so $\epsilon^\star$ lies in the (nonempty, open) interior $\Theta^\circ$; for small $t\neq0$, $\epsilon':=\epsilon^\star+td\in\Theta^\circ$ with $\epsilon'\neq\epsilon^\star$ and $\Phi(\epsilon')-\Phi(\epsilon^\star)=t\,Gd=0$. Thus $\Phi$ is not injective. (The necessity direction uses that $\Theta$ is full-dimensional; at a boundary corner a rank-deficient model could be locally identified.)

\paragraph{Signature collapse.}
If two diluents share a signature, $\bar g_{R_a}=\bar g_{R_b}$ ($a\neq b$), columns $a,b$ of $G$ coincide, so $d=e_a-e_b\neq0$ satisfies $Gd=0$: $G$ is rank-deficient and only the sum $\epsilon_a+\epsilon_b$ is identifiable, not the individual fractions ($\epsilon$ and $\epsilon+t(e_a-e_b)$ are indistinguishable for small feasible $t$). More generally $\mathrm{null}(G)=\{d:\sum_jd_j(\bar g_{R_j}-\bar g_P)=0\}$ is exactly the set of unidentified directions, the multi-diluent analogue of a low-separation wall. At $J=1$ full rank reads $\bar g_{R_1}\neq\bar g_P$, the mean-posterior analogue (not a literal specialization) of the single-diluent tell separation $q>\alpha_1$. \hfill$\square$

\putbib
\end{bibunit}
\fi

\end{document}